\documentclass{article}

\usepackage{arxiv}

\usepackage[utf8]{inputenc} 
\usepackage[T1]{fontenc}    
\usepackage{hyperref}       
\hypersetup{colorlinks=true, linkcolor=gnblue4, breaklinks=true, urlcolor=gnblue4, citecolor=gnblue4}

\usepackage{url}            
\usepackage{booktabs}       
\usepackage{amsfonts}       
\usepackage{nicefrac}       
\usepackage{microtype}      
\usepackage{lipsum}		
\usepackage{graphicx}
\usepackage[dvipsnames]{xcolor}
\usepackage{natbib}
\usepackage{doi}
\usepackage{titletoc}

\usepackage{bbm}
\usepackage{mathbbol}
\usepackage{imakeidx}

\usepackage{amsmath, amssymb, amsthm}
\usepackage{accents}
\usepackage{upgreek}
\newlength{\dhatheight}

\usepackage{mathtools}
\DeclarePairedDelimiter{\ceil}{\lceil}{\rceil}
\DeclareMathAlphabet{\mymathbb}{U}{BOONDOX-ds}{m}{n}
\usepackage{yhmath}
\usepackage{mathrsfs}
\usepackage{mathdots}
\usepackage{subfig}
\usepackage[makeroom]{cancel}

\usepackage{tikz}
\usetikzlibrary{matrix, backgrounds, calc}
\usetikzlibrary{graphs, quotes}
\usepackage{geometry}
\usepackage{titlesec}
\usepackage{enumitem}

\newtheorem{definition}{Definition}
\newtheorem{remark}{Remark}
\newtheorem{assumption}{Assumption}
\newtheorem{theorem}{Theorem}
\newtheorem{lemma}{Lemma}
\newtheorem{corollary}[theorem]{Corollary}
\newtheorem{proposition}{Proposition}

\newtheorem{conjecture}{Conjecture}
\MakeRobust{\eqref}

\definecolor{gnred1}{RGB}{71,0,0} 
\definecolor{gnred2}{RGB}{117,0,0} 
\definecolor{gnred3}{RGB}{164,0,0} 
\definecolor{gnred4}{RGB}{211,0,0} 
\definecolor{gnred5}{RGB}{255,0,0} 
\definecolor{gnred6}{RGB}{255,42,34} 
\definecolor{gnred7}{RGB}{255,91,89} 

\definecolor{gnblue1}{RGB}{0,36,71}   
\definecolor{gnblue2}{RGB}{0,60,118}  
\definecolor{gnblue3}{RGB}{0,85,164}  
\definecolor{gnblue4}{RGB}{0,108,212} 
\definecolor{gnblue5}{RGB}{0,133,255}  
\definecolor{gnblue6}{RGB}{35,156,255} 
\definecolor{gnblue7}{RGB}{88,177,255} 

\definecolor{gnbrown1}{RGB}{71,27,0}  
\definecolor{gnbrown2}{RGB}{117,45,0} 
\definecolor{gnbrown3}{RGB}{164,62,0} 
\definecolor{gnbrown4}{RGB}{211,80,0} 
\definecolor{gnbrown5}{RGB}{255,97,0} 
\definecolor{gnbrown6}{RGB}{255,127,26} 
\definecolor{gnbrown7}{RGB}{255,155,86} 

\newcommand{\En}{\mathrm{E}}
\newcommand{\prox}{\text{prox}}

\newcommand{\tr}{\mathrm{Tr}}
\newcommand{\real}{\mathbb{R}}

\renewcommand{\real}{\mathbb{R}}

\newcommand{\diag}{\mathrm{diag}}

\newcommand{\argmax}[2] {\mathrm{arg}\max_{#1}#2}
\newcommand{\argmin}[2] {\mathrm{arg}\min_{#1}#2}

\DeclareSymbolFont{bbold}{U}{bbold}{m}{n}
\DeclareSymbolFontAlphabet{\mathbbold}{bbold}

\newcommand{\vect}[1]{\mathbbold{#1}}
\newcommand{\vectorones}[1][]{\vect{1}_{#1}}
\newcommand{\vectorzeros}[1][]{\vect{0}_{#1}}

\title{Competition, stability, and functionality\\ in excitatory-inhibitory neural circuits}

\author{ \href{https://orcid.org/0009-0000-3444-0838}{\includegraphics[scale=0.06]{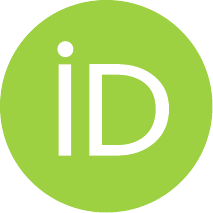}\hspace{1mm}Simone Betteti} \\
RIAS Lab\\ 
	The Italian Institute of\\
	Artificial Intelligence for Industry\\
	Turin, 10129, IT \\
	\texttt{simone.betteti[at]ai4i.it} \\
	\And
	\href{https://orcid.org/0000-0002-0983-1615}{\includegraphics[scale=0.06]{orcid.pdf}\hspace{1mm}William Retnaraj} \\
	Center for Control Systems and Dynamics\\
	University of California at San Diego\\
	San Diego, CA, 92093, US \\
	\texttt{wretnaraj[at]ucsd.edu} \\
	\And
	\href{https://orcid.org/0000-0001-5629-2565}{\includegraphics[scale=0.06]{orcid.pdf}\hspace{1mm}Alexander Davydov} \\
	Department of Mechanical Engineering\\
	Rice University\\
	Houston, TX, 77005, US \\
	\texttt{sd210[at]rice.edu} \\
	\AND
	\href{https://orcid.org/0000-0001-9582-5184}{\includegraphics[scale=0.06]{orcid.pdf}\hspace{1mm}Jorge Cortés} \\
	Center for Control Systems and Dynamics\\
	University of California at San Diego,\\
	San Diego, CA, 92093, US \\
	\texttt{cortes[at]ucsd.edu}
	\And	
    \href{https://orcid.org/0000-0002-4785-2118}{\includegraphics[scale=0.06]{orcid.pdf}\hspace{1mm}Francesco Bullo} \\
	Center for Control, Dynamical Systems and Computation\\
	University of California at Santa Barbara,\\
	Santa Barbara, CA, 93106 US \\
	\texttt{bullo[at]ucsb.edu}
}

\renewcommand{\headeright}{}
\renewcommand{\undertitle}{}
\renewcommand{\shorttitle}{Competition, stability, and functionality in excitatory-inhibitory neural circuits}

\hypersetup{
pdftitle={Competition, stability, and functionality in excitatory-inhibitory neural circuits},
pdfsubject={q-bio.NC, q-bio.QM},
pdfauthor={Simone Betteti, William Retnaraj, Alexander Davydov, Jorge Cortés, Francesco Bullo},
pdfkeywords={Asymmetric neural networks,  Game theory, Excitation-inhibition, Network stability},
}

\begin{document}
\maketitle

\begin{abstract}
Energy-based models have become a central paradigm for understanding computation and stability in both theoretical neuroscience and machine learning. However, the energetic framework typically relies on symmetry in synaptic or weight matrices, a constraint that excludes biologically realistic systems such as excitatory-inhibitory (E-I) networks. When symmetry is relaxed, the classical notion of a global energy landscape fails, leaving the dynamics of \emph{asymmetric} neural systems conceptually unanchored. In this work, we extend the energetic framework to \emph{asymmetric} firing-rate networks, revealing an underlying game-theoretic structure for the neural dynamics in which each neuron is an agent that seeks to minimize its own energy. In addition, we exploit rigorous stability principles from network theory to study regulation and balancing of neural activity in E-I networks. We combine the novel game-energetic interpretation and the sharp stability results to revisit standard examples in theoretical neuroscience, such as the Wilson-Cowan and lateral inhibition models. These insights allow us to study columns of lateral inhibitory microcircuits as contrast enhancers, i.e., as circuits that selectively sharpen subtle differences in the environment through hierarchical excitation-inhibition interplay. Our results bridge energetic and game-theoretic views of neural computation, offering a pathway toward the systematic engineering of biologically grounded, dynamically stable neural architectures. 
\end{abstract}


\section*{Introduction}
Recent advancements in both theoretical neuroscience and machine learning~\citep{KD-HJJ:16,RH-SB-LJ:21,KL-SJJ-KD:25} frame relevant computational properties of neural networks through energy lenses. The energetic framework provides a simple yet powerful interpretive tool to understand neural dynamics and answer scientifically profound questions on emergent cognitive phenomena. The recent Nobel prize awarded to John J. Hopfield~\citep{HJJ:82, CMA-GS:83, HJJ:84} underscores the lasting impact of this approach across topics and disciplines. Drawing inspiration from spin-glass theory~\citep{SD-KS:75, MM-PG-VM:86}, Hopfield introduced an energy function for neural dynamics, proving asymptotic convergence to memory states encoded in the synaptic matrix and retrieved as stable attractors. The energetic framework has since become a core element of neural dynamics investigation~\citep{KD:86,RET-TA:12, BS-BG-BF-ZS:25fr}, casting collective neural activity as a minimization process aimed at retrieving meaningful network states. Lately, the pairing of generative AI technologies and energy-driven dynamics has inspired a wave of theoretical studies on the dynamical properties of large language and diffusion models~\citep{HB-LY-PB:23, AL:24}, enriching the study on their interpretability and predictability. Thus, energy becomes reliability; energy grants robustness; yet, energy requires structure.     

The classical applications of energy-driven neural dynamics rest on a strong structural assumption: the symmetry of the synaptic
matrix~\citep{HD-WJ:87}. In simple terms, symmetry enforces perfect reciprocity-if neuron $i$ projects to neuron $j$ with synaptic strength $w$, then neuron $j$ must project back to neuron $i$ with the same strength. The symmetry constraint is \textit{de facto} a requirement, being itself a structural condition for the stability of the memory patterns. However, the symmetry constraint introduces compromising modeling flaws, thereby widening the gap between theoretical studies and experiments~\citep{PG:86, YH-ZL-HL:13}. Most notably, it clashes with the organization of real neural circuits, where populations of excitatory (E) and inhibitory (I) neurons interact~\citep{HB-DA-HAR:06}. By Dale's law~\citep{EJC-FP-KK:54}, E-neurons project only excitatory outgoing connections, and I-neurons only inhibitory outgoing connections. Under these conditions, the symmetry constraint can only be satisfied if excitatory and inhibitory populations do not interact-a scenario that defeats the very purpose of modeling E-I networks. 

The scientific literature on excitatory-inhibitory (E-I) networks has made substantial progress in characterizing their dynamics and stability, yet it largely lacks a unifying, energy-like framework capable of interpreting the activity of biologically-plausible neural circuits. With few notable exceptions relying on strong structural assumptions~\citep{ASI:77li, FM-TA:95}, energetic interpretations have remained confined to \emph{symmetric} settings. This limitation has created a conceptual divide between biologically grounded, \emph{asymmetric} E-I networks and cognitively oriented models, where optimization and energy minimization provide a core explanatory principle. As a result, the role of energy-based organization in real neuronal circuits remains underdeveloped. At the same time, a broad body of work on inhibition-stabilized E-I circuits~\citep{VC-SH:98,SA-AB-GHC:20,SS-CC:20,MD-DA-IV:24} points to intrinsic balancing mechanisms that regulate activity, prevent runaway excitation, and anchor network dynamics to well-defined operating regimes. Such stabilization phenomena strongly suggest\footnote{In dynamical system theory it is widely known that existence of stable equilibrium points implies the existence of a Lyapunov function by the inverse Lyapunov theorem~\citep[Chapter 4]{KHK:02}} the presence of implicit cost-minimizing principles, even in the absence of symmetry. The inability to reconcile energy-based descriptions with the balancing dynamics of \emph{asymmetric} E-I networks therefore emerges as a central theoretical bottleneck. Addressing this gap is essential to bridge biophysically grounded circuit models with optimization-driven theories of neural computation.

\begin{figure}[tbph!]
 \centering
 \includegraphics[width=14.4cm]{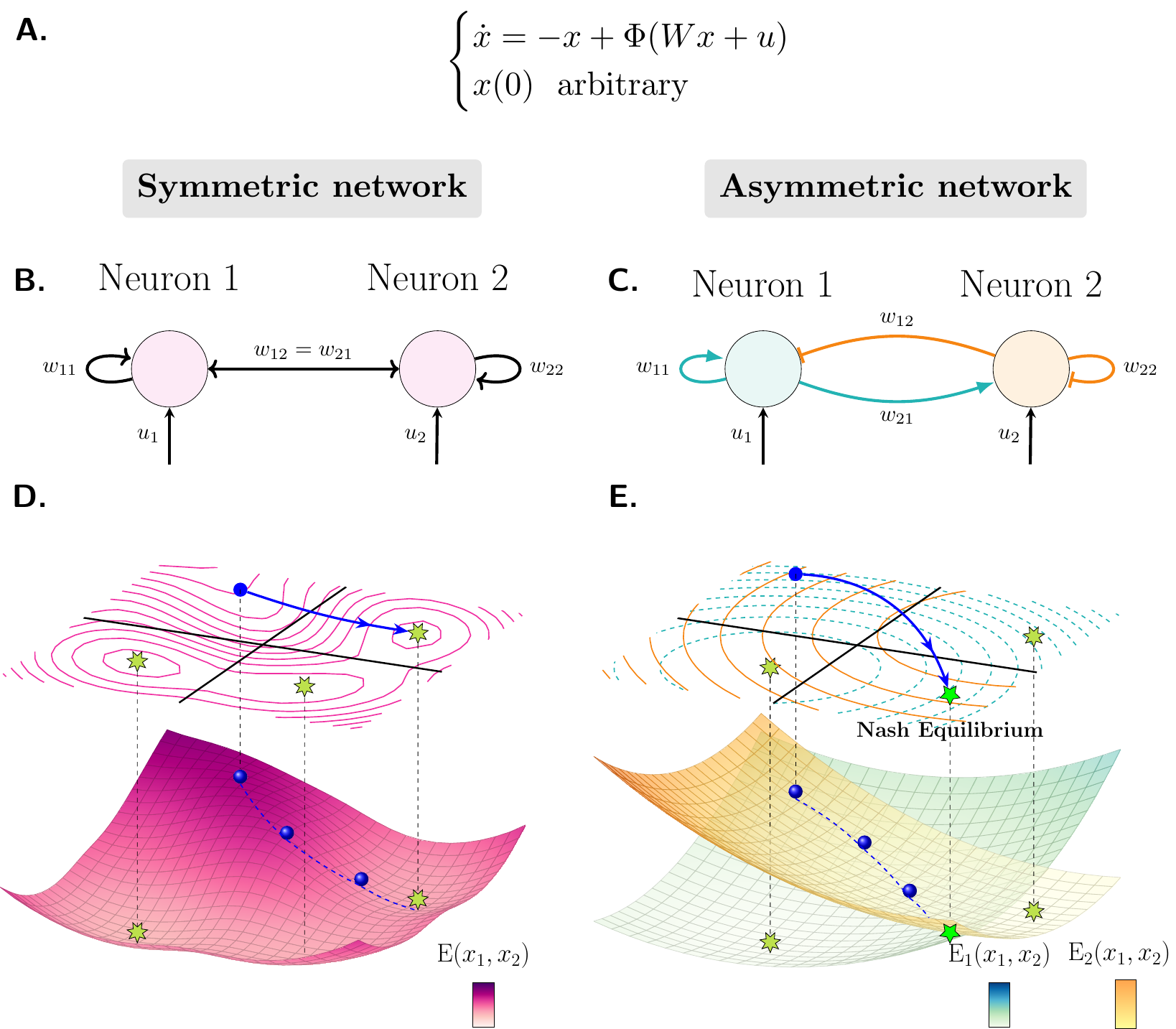}
 \caption{\textbf{A qualitative comparison of neural interactions and collective actions in \emph{symmetric} (left) and \emph{asymmetric} (right) firing-rate networks.} (A) Dynamics of a recurrent neural network with firing rates $x$, activation function $\Phi$, synaptic matrix 
$W$, and external input $u$. Firing-rate models capture the temporal evolution of mean neural activity across interconnected populations.
(B) Graphic visualization of neural interaction in \emph{symmetric} neural networks ($W=W^{\top}$). In \emph{symmetric} neural networks, every pair of neurons interacts reciprocally: if neuron $i$ influences neuron $j$ with strength $w$, neuron $j$ exerts the same influence back on $i$. Symmetry enforces balance and reciprocity across the network, and allows for the existence of a potential for the dynamics directing the flow of the system.
(C) Graphic visualization of neural interaction in \emph{asymmetric} neural networks ($W\neq W^{\top}$). In \emph{asymmetric} networks, neurons contribute to the collective activity through their outgoing synapses, but reciprocal connections need not exist - or may differ in strength. Interactions are therefore unbalanced, reflecting the structure of real excitatory–inhibitory circuits.
(D) Energy in \emph{symmetric} firing-rate networks. \emph{Symmetric} interactions define a global energy function that all neurons jointly minimize. Neural activity converges toward a minimum-energy configuration, corresponding to a stable equilibrium of the dynamics.
(E) Energies in \emph{asymmetric} firing-rate networks. When symmetry is broken, each neuron minimizes its own energy. Their competition shapes a balance point - a Nash equilibrium - where no neuron can unilaterally improve its outcome. This equilibrium need not correspond to a global minimum but represents the best possible compromise given the asymmetry of interactions.}
 \label{fig: graph_abstr}
\end{figure}

In this manuscript, we extend the energetic framework of \emph{symmetric} energy-based networks to \emph{asymmetric} excitatory-inhibitory (E-I) systems (Fig.~\ref{fig: graph_abstr}), and reinterpret classical results on inhibition-stabilized dynamics within this broader setting. Competition is distilled in our idea is to model each neuron, or neuronal population, as a selfish agent optimizing its own energy, so that network dynamics emerge as a non-cooperative game whose meaningful states correspond to Nash equilibria. Stability is addressed within the new energetic, game-theoretic framework through the analysis of inhibitory stabilization in E-I circuits. In multi-stable E-I circuits, such as spatial working memory~\citep{CA:00} modules, inhibition and dissipation interact, jointly constraining excitation and preventing runaway dynamics. Surprisingly, our analysis reveals that circuits exhibiting a unique globally stable equilibrium, such as early visual inhibitory circuits~\citep{OH-FIM-SES:09}, rely on strong dissipation alone to dominate recurrent excitation. Functionality is approached through the study of three canonical neural architectures. First, we revisit the Wilson-Cowan (WC) model~\citep{WHR-CJD:73}, refining the inhibition-stabilized regime into mono-stable and multi-stable phases and showing that oscillations and chaos arise from excitatory dominance. Second, we introduce a novel energetic interpretation of lateral inhibition, showing how E-I networks can implement input discrimination with arbitrary precision~\citep{FY-CMA-KTG:96}. Finally, we extend lateral inhibition to columnar organizations, suggesting that stacked E-I microcircuits may act as contrast-enhancing modules in sensory processing. More broadly, this work provides a unifying language that reconciles biological asymmetry with optimization principles, offering a bridge between mechanistic neuroscience, dynamical systems theory, and the engineering of stable yet flexible neural computation.

\section*{Taking the energy principle beyond symmetry}

\subsection*{One energy drives global neuronal activity}
Firing-rate networks have been a dominant framework in the study of mesoscopic neural activity~\citep{DP-ALF:05,EGB-TDH:10}, offering insight into the dynamics of distributed information processing in the brain, from excitatory-inhibitory interaction to associative memory. They have also inspired relevant machine learning applications, such as echo state networks~\citep{JH-HH:04}, where the recurrence among neurons is exploited to define tailored input-output relationships.
In their most general characterization, firing-rate dynamics obey the ordinary differential equation
    \begin{equation}\label{eq: fr-clas}
    \dot x = -x + \Phi(Wx+u)\\
\end{equation}
where $W\in\real^{N\times N}$ is the synaptic matrix regulating the interaction among neurons, $u\in\real^{N}$ is the external input and $\Phi(x)=(\upphi_{1}(x_{1}),\dots,\upphi_{N}(x_{N}))$ is the activation function with positive components abstracting the average rate of firing of an ensemble of neurons~\citep{DP-ALF:05}. After normalization, the rates of activation for the units vary in the interval $[0,1]$, with the two extremes representing either neural inactivation or neural activation. 

A key determinant of network properties is the structure of the synaptic matrix. In associative memory models, the synaptic interaction among neurons is \emph{symmetric} ($W=W^{\top}$, see Fig.~\ref{fig: graph_abstr}(B)) and typically given by a covariance learning rule~\citep{GW-KWM:02} implementing zero-shot learning from random memories. The symmetry enables the definition of an Hopfield-type energy function~\citep{BS-BG-BF-ZS:25fr}
\begin{equation}\label{eq: En-gb}
    \En(x,u)=-\frac{1}{2}x^{\top}Wx-x^{\top}u+\sum_{i=1}^{N}\int_{0}^{x_{i}}\Phi_{i}^{-1}(s)\ ds
\end{equation}
that ensures convergence of the network activity to the desired patterns from any initial condition $x_{0}$ (see Fig~\ref{fig: graph_abstr}(D)). Crucially, the memories correspond to the minima of the energy landscape: the dynamics~[\ref{eq: fr-clas}] unfold as gradient-like motion descending into valleys centered on the learned patterns.

Recently~\citep{VC-AG-AD-GR-FB:23a,AG-AD-FB:24d}, firing-rate dynamics with \emph{symmetric} synaptic interactions have been re-interpreted as proximal gradient dynamics, a framework that has gained traction in the analysis and optimization of machine learning systems. The idea is that the firing-rate system minimizes a composite scalar cost $\mathrm{E}_{\text{int}}(x,u)+\mathrm{E}_{\text{act}}(x)$. The interaction cost $\mathrm{E}_{\text{int}}(x,u)$ is the potential along which firing-rate dynamics would evolve without physical constraints (e.g., non-negative firing rates). The activation cost $\mathrm{E}_{\text{act}}(x)$ assigns prohibitively high penalty to infeasible regions and null or low cost to feasible regions. The proximal operator of $\mathrm{E}_{\text{act}}$ takes $x-\nabla \mathrm{E}_{\text{int}}(x,u)$ (a gradient step of $\mathrm{E}_{\text{int}}$) as input and outputs its nearest point in the feasible region, thereby regularizing the basic gradient flow of $\mathrm{E}_{\text{int}}$ by smoothly steering it toward feasible regions.
\begin{figure}[!tbph]
    \centering
    \includegraphics[width=\linewidth]{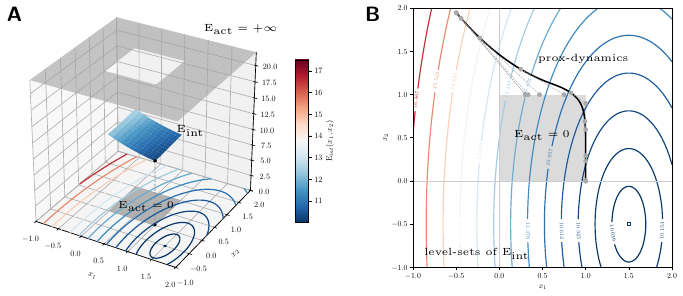}
    \caption{\textbf{Graphic interpretation of the $\operatorname{prox}_{\mathrm{E}_{\text{act}}}$ operator action for any fixed external input $u$.} (A) Energy landscape associated to a quadratic interaction cost $\mathrm{E}_{\text{int}}$ and activation cost $\mathrm{E}_{\text{act}}$ for a saturated activation function. In particular, $\mathrm{E}_{\text{act}}$ is equal to zero over the entire domain of linear activation, and equal to $+\infty$ over the entire saturated region. The dynamics are therefore steered inside (or at the boundary) of the linear region. Generally, the interaction cost $\mathrm{E}_{\text{int}}$ defines the global landscape over which the neural network dynamics evolve and $\mathrm{E}_{\text{act}}(x)$ the physical constraints of the system. (B) The $\operatorname{prox}$ operator smoothly alters the basic gradient descent on $\mathrm{E}_{\text{int}}$, pushing the dynamics towards the point in the feasible region that is closest - in the Euclidean metric - to the gradient step $x-\nabla \mathrm{E}_{\text{int}}$.}
    \label{fig: prox-grad}
\end{figure}
Under this decomposition, the firing-rate dynamics can be rewritten as 
\begin{equation}\label{eq: fr-prox}
    \dot x = -x + \prox_{\mathrm{E}_{\text{act}}}(x-\nabla \mathrm{E}_{\text{int}}(x,u))\\
\end{equation}
where $\prox_{\mathrm{E}_{\text{act}}}$ is the proximal operator with respect to the function $\mathrm{E}_{\text{act}}$ (see SI for a rigorous introduction to proximal operators), satisfying $\prox_{\mathrm{E}_{\text{act}}}(x)=\Phi(x)$. 
Thus, the equilibrium of the firing-rate system (when unique and stable) must minimize the composite cost $\mathrm{E}_{\text{int}}(x,u)+\mathrm{E}_{\text{act}}(x)$. 
Comparing~\eqref{eq: fr-clas} with~\eqref{eq: fr-prox}, for \emph{symmetric} \(W\), we get
\begin{equation}
    x-\nabla \mathrm{E}_{\text{int}}(x,u) = Wx+u ,
\end{equation}
from which one possible cost structure is $\mathrm{E}_{\text{int}}(x,u)=\frac{1}{2}x^{\top}(\mathcal{I}_{N}-W)x-x^{\top}u$, where $\mathcal{I}_{N}$ is the identity matrix. 
Using the definition of the proximal operator, we obtain $\mathrm{E}_{\text{act}}(x)=\sum_{i=1}^{N}(\int_{0}^{x_{i}}\Phi_{i}^{-1}(s)\ ds - \frac{1}{2}x_{i}^{2})$, and finally
\begin{equation}
    \mathrm{E}_{\text{int}}(x,u)+\mathrm{E}_{\text{act}}(x) = \En(x,u).
\end{equation}
Thus, the traditional energetic interpretation of firing-rate networks emerges naturally within the proximal gradient framework. This unification not only confirms the consistency of the two perspectives, but also equips us with a powerful modern tool to view firing-rate dynamics as minimizing interaction and activation costs. Importantly, modeling of many cognitive processes explicitly requires excitatory and inhibitory interactions, in which case the synaptic matrix $W$ is no longer \emph{symmetric}. Developing an analogous framework for such \emph{asymmetric} firing-rate networks is a central open challenge, which we tackle here.

\subsection*{To each neuron its own energy}

Communication in cortical networks is inherently directional~\citep{PG:86, TA-ADJ:88}: each neuron propagates its activity forward and transmits exclusively either excitatory (positive) or inhibitory (negative) signals (see Fig.~\ref{fig: graph_abstr}(C) for a diagrammatic representation of excitatory-inhibitory interaction). 
This biological constraint makes \emph{asymmetric} synaptic matrices unavoidable in realistic models, as the non-reciprocity of neural connections directly breaks symmetry. 
This loss of symmetry has significant mathematical consequences. 
In particular, many tools from dynamical systems theory that rely on \emph{symmetric} interactions no longer apply, making it difficult to prove convergence of trajectories or to rely on a global energy function to interpret the system’s behavior. Without the unrealistic assumption on \emph{symmetric} network interactions, gaining insight on \emph{asymmetric} recurrent neural networks thus becomes a primary challenge.

Remarkably, the proximal operator formulation of the firing-rate dynamics
can be extended to the case of an \emph{asymmetric} network.  The key
difference is that, while symmetry allows us to define a single global
interaction cost $\mathrm{E}_{\text{int}}(x,u)$, asymmetry requires
considering instead a family of interaction costs
$\{\mathrm{E}^{i}_{\text{int}}(x,u_{i})\}_{i=1}^{N}$, one for each neuron
(or neuronal population), where $u_{i}$ is the scalar external input to the
$i^{\text{th}}$ neuron.  The activation costs remain neuron-specific but
structurally identical to the \emph{symmetric} case, with
$\prox_{\mathrm{E}^{i}_{\text{act}}}(x) = \Phi_{i}(x)$. Therefore, we can
rewrite the \emph{asymmetric} firing-rate dynamics for each neuron as
\begin{equation}\label{eq: fr-prox-s}
    \dot x_{i} = -x_{i} + \prox_{\mathrm{E}^{i}_{\text{act}}}(x_{i}-\nabla_{x_{i}} \mathrm{E}^{i}_{\text{int}}(x,u_{i}))\\
\end{equation}
Owing to the equivalence with the single neuron dynamics in~\eqref{eq: fr-clas}, the following relationship must hold
\begin{equation}\label{eq:pre-energy}
    x_{i}-\nabla_{x_{i}}\mathrm{E}^{i}_{\text{int}}(x,u_{i}) = \sum_{j=1}^{N}W_{ij}x_{j}+u_{i}
\end{equation}
from which it easy to see that the interaction cost of each neuron is
$\mathrm{E}^{i}_{\text{int}}(x,u_{i})=-x_{i}[\sum_{j=1}^{N}(1-\frac{1}{2}\updelta_{ij})W_{ij}x_{j}-\frac{1}{2}x_{i}+u_{i}]$. Exploiting
once again the properties of the proximal operator, the activation cost of
each neural is
$\mathrm{E}^{i}_{\text{act}}(x_{i})=\int_{0}^{x_{i}}\Phi_{i}^{-1}(s)\ ds-\frac{1}{2}x_{i}^{2}$. Consequently,
we introduce interaction and activation cost families,
$\{\En_{\mathrm{int}}^{i}\}_{i=1}^{N}$ and
$\{\En_{\mathrm{act}}^{i}\}_{i=1}^{N}$, which jointly induce proximal gradient play
dynamics~\citep{BT-OGJ:98} across the neural network. The sum of the
interaction and activation costs for each neuron allows us to define a
specific neuron energy as
\begin{align}\label{eq: En-pn}
    \En^{i}(x,u_{i})&=\mathrm{E}^{i}_{\text{int}}(x,u_{i})+\mathrm{E}^{i}_{\text{act}}(x_{i})\\
    &=-x_{i}\sum_{j=1}^{N}(1-\frac{1}{2}\updelta_{ij})W_{ij}x_{j}-x_{i}u_{i}+\int_{0}^{x_{i}}\Phi_{i}^{-1}(s)\ ds \nonumber
\end{align}
and therefore a family of energies $\{\En^{i}(x,u_{i})\}_{i=1}^{N}$ over the whole network. Strikingly, the structure of~\eqref{eq: En-pn} closely resembles the global energy~\eqref{eq: En-gb} that characterizes the \emph{symmetric} case. To express this compactly, define the entrywise product $\odot$ for two vectors $a,b \in \mathbb R^{N}$ by $(a \odot b)_{i} = a_{i}b_{i}$, and let $\Tilde W$ be the modified synaptic matrix with entries $\Tilde W_{ij} = (1 - \tfrac{1}{2}\updelta_{ij}) W_{ij}$. Then we can introduce the vector-valued energy
\begin{align}\label{eq: En-vc}
\Tilde \En(x,u) &= (\En^{1}(x,u_{1}), \dots, \En^{N}(x,u_{N})) \nonumber\\
&= -x \odot (\Tilde W x) - x \odot u + F(x),
\end{align}
where $F_{i}(x) = \int_{0}^{x_{i}} \Phi_{i}^{-1}(s) ds$.

Unlike the \emph{symmetric} case, however, the vectorial nature of~\eqref{eq: En-vc} prevents us from interpreting the network states as minima of a single potential function. 
Instead, each neuron contributes its own local energetic perspective, reflecting the fundamentally \emph{asymmetric} structure of neural communication.
Finally, observe how the components of the vectorial energy in~\eqref{eq: En-vc} are affected exclusively by the inputs $u_{i}$ to the corresponding $i^{th}$ neuron.

\subsection*{A game of energies}
Game theory provides a powerful lens for interpreting the dynamics of systems composed of multiple interacting agents with individual objectives. 
In this framework, each agent may choose an action strategy to minimize (or maximize) its own cost function, while the actions of the other agents act as constraints shaping the feasible set of outcomes. 
A Nash equilibrium is a configuration of player actions in which no agent can unilaterally improve its outcome by altering its own strategy, even if the global configuration is not optimal.
The concept of Nash equilibria is of vital importance in the consideration of noncooperative games. 

As described above, \emph{asymmetric} neural networks can be viewed through the perspective of non-cooperative games, which we call \textit{proximal gradient play}, where each neuron plays the role of an agent attempting to minimize its own energy function through gradient play dynamics. 
The Nash equilibria of such networks, if they exist, satisfy
\begin{align}
x^{\star}_{i} &= \arg \min_{x_{i} \in \mathbb R} \En^{i}(x_{i}, x^{\star}_{-i},u_{i}) \nonumber \\
&= \arg \min_{x_{i} \in \mathbb R} \Big( \mathrm{E}^{i}_{\text{int}}(x_{i}, x^{\star}_{-i},u_{i}) + \mathrm{E}^{i}_{\text{act}}(x_{i}) \Big),
\end{align}
where $x^{\star}_{-i}$ denotes the states of all neurons other than $i$. 
In simpler terms, each neuron in the \emph{asymmetric} network strives to jointly minimize the cost $\mathrm{E}^{i}_{\text{int}}(x_{i}, x_{-i},u_{i})$ of its interaction with the other neurons and its own cost of activation $\mathrm{E}^{i}_{\text{act}}(x_{i})$.
Crucially, this does not imply that neurons reach the absolute minimum of their individual energies. 
Rather, they settle into states where no unilateral deviation would yield an advantage (see Fig.~\ref{fig: graph_abstr}(F)), which is an inherently game-theoretic notion of equilibrium.

\section*{Navigating the Wilson-Cowan dynamic regimes}
Excitatory-inhibitory (E-I) neural networks are a central framework for studying \emph{asymmetric} neural dynamics, capturing the distinct roles and interactions of excitatory and inhibitory populations.
The excitatory population exclusively projects positive synaptic weights, whereas the inhibitory kind projects exclusively negative weights. 
This is the principle popularized in Dale’s law~\citep{EJC-FP-KK:54}, which inevitably leads to an \emph{asymmetric} interaction graph.
The simplest example of an E-I network is the WC model~\citep{WHR-CJD:73}, a firing-rate framework describing the dynamical interaction of two neuronal populations: one excitatory and one inhibitory.
In their pioneering work, Wilson and Cowan revealed that this minimal E-I circuit can exhibit rich dynamical behavior ranging from stable equilibria to limit cycles. 
Therefore, we leverage the dynamical richness of the WC model as a testbed to study the intertwined relationship of the E-I firing-rate dynamics and the corresponding game-theoretic setting. 

We adopt a simple two-dimensional firing-rate model in~\eqref{eq:EI-2dim} with a saturating activation function, capturing the interconnection between one excitatory and one inhibitory population to implement the WC model:

\vspace*{-1ex}
{\small
\begin{equation}\label{eq:EI-2dim}
    \begin{pmatrix}
        \dot x_E\\
        \dot x_I
    \end{pmatrix} = -
    \begin{pmatrix}
        x_E\\
        x_I
    \end{pmatrix} + 
    \left[
    \begin{pmatrix}
        w_{EE} & -w_{EI}\\
        w_{IE} & -w_{II}
    \end{pmatrix}
    \begin{pmatrix}
        x_E\\
        x_I
    \end{pmatrix} + 
    \begin{pmatrix}
        u_E\\
        u_I
    \end{pmatrix}
    \right]_0^1
\end{equation}
}
The state \(x = (x_E, x_I)^\top\) is defined on the unit square $[0,1]^{2}$ to represent allowable firing rates, with initial conditions picked therein.
The saturating activation function we use here is a linear-threshold activation function, also known as a ReLU (Rectified Linear Unit), that cuts off at \(0\) and saturates at \(1\). 
The synaptic matrix $W$ has column definite sign, and therefore adheres to Dale's law. 
In this section, we reveal that the WC dynamics lives on a special game that makes explicit the opposing objectives of the two populations under certain conditions, while revealing how their strategic tension structures their dynamical behavior. 
We identify the self-excitatory weight $w_{EE}$ as the principal switch governing transitions across these regimes such as the dissipation-excitation dominance and the inhibition-stabilized regime, which dictate the dominant dynamical behavior.
We connect these regimes to the zero-sum game properties and finally discuss the emergence of a limit cycle, when the populations perpetually fail to arrive at consensus.

\begin{figure}
    \centering
    \includegraphics[width=.75\linewidth]{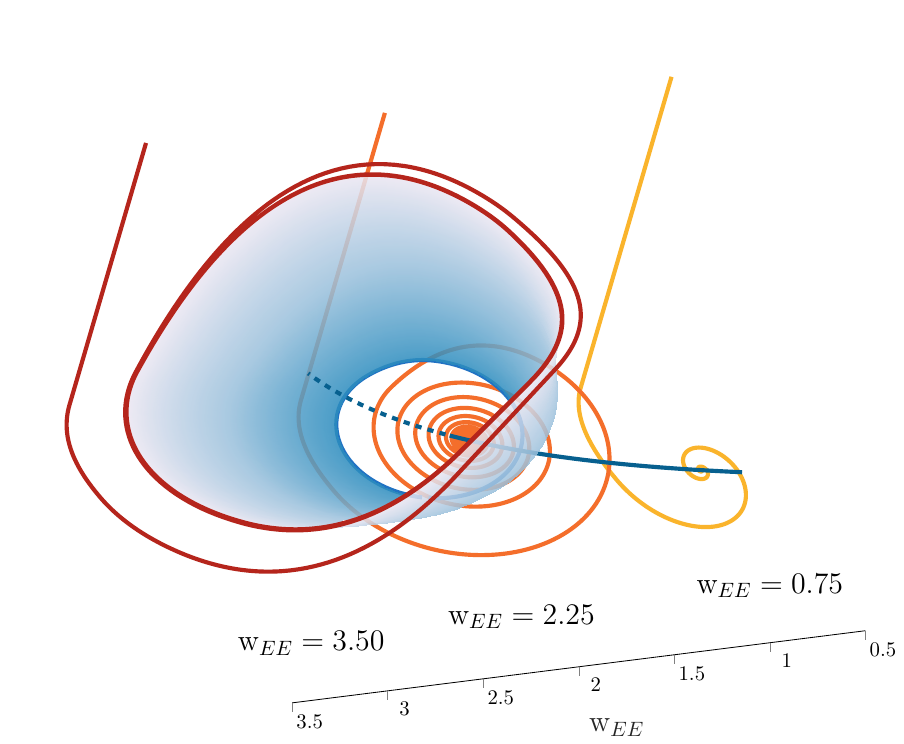}
    \caption{\textbf{Self-excitatory weight $w_{EE}$ drives the emergence of limit cycles through Hopf-like bifurcation.} Hopf-like bifurcation observed as the E-I firing-rate network transitions from the antagonistic weak decision regime to the antagonistic indecision regime. Plot shows the limit set for each \(w_{EE}\) as it is varied starting in the consensual regime (\(w_{EE} = 0.5 < 1\)) and capturing the point of bifurcation at \(w_{EE} = w_{II} + 2 = 2.5\). Parameters used: \(w_{EI} = 4, w_{IE} = 2.3, w_{II} = 0.5, u_E = 1, u_I = -0.2\). Solid blue curves represent stable limit sets, while the dashed blue line represents the locus of unstable equilibria.}
    \label{fig: Hopf}
\end{figure}

\begin{figure}[t!]
 \centering
 \includegraphics[height=13cm, width=13cm]{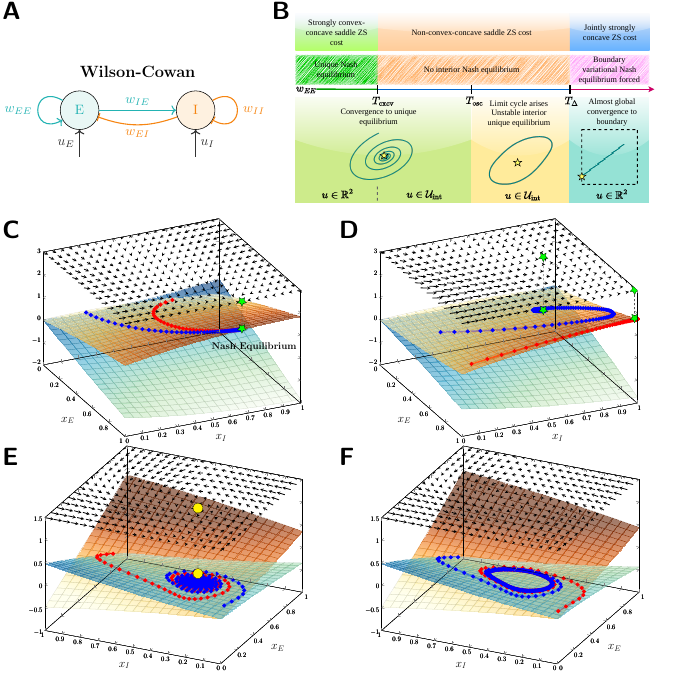}
 \caption{\textbf{Schematic of the Wilson-Cowan (WC) model and complete dynamical and game characterization of the exhibited regimes} 
 (A) Schematic of the WC model, with one excitatory neuron ($E$) and one inhibitory neuron ($I$). 
 Excitatory synapses ($w_{EE}$, $w_{IE}$, light blue) are positive, while inhibitory synapses ($w_{II}$, $w_{EI}$, light orange) are negative. Each neuron receives an external input ($u_E$, $u_I$). 
 (B) Summary of the dynamical and game-theoretic regimes of the two-population WC model. Define \(T_{\mathrm{cxcv}}=d_E\), \(T_{\mathrm{osc}}=d_E+d_I+w_{II}\), and \(T_{\Delta}=d_E+\tfrac{w_{EI}w_{IE}}{d_I+w_{II}}\). Assume \(T_{\mathrm{cxcv}}<T_{\mathrm{osc}}<T_{\Delta}\). Then the system transitions from a unique globally stable Nash equilibrium (\(w_{EE}<T_{\mathrm{cxcv}}\)), to a unique stable interior non-Nash equilibrium (\(T_{\mathrm{cxcv}}<w_{EE}<T_{\mathrm{osc}}\), \(u\in\mathcal U_{\mathrm{int}}\)), to a stable limit-cycle regime (\(T_{\mathrm{osc}}<w_{EE}<T_{\Delta}\), \(u\in\mathcal U_{\mathrm{int}}\)). For \(u\notin\mathcal U_{\mathrm{int}}\), stable boundary/saturated equilibria and multistability may occur; for \(w_{EE}>T_{\Delta}\), at least one stable boundary variational Nash equilibrium is forced. (C--F) Energy surfaces associated with the excitatory (blue) and inhibitory (orange) neurons, together with the corresponding dynamical regimes.
 (C) Unique Nash equilibrium: the landscape is convex in the excitatory variable and concave in the inhibitory variable, yielding a single equilibrium that is globally asymptotically stable.
 (D) Multiple locally stable equilibria: outside \(\mathcal U_{\mathrm{int}}\), the geometry admits multiple attracting states, and the dynamics settle to one of them depending on the initial condition.
 (E) Non-Nash spiral convergence: for \(u\in\mathcal U_{\mathrm{int}}\) and \(w_{EE}>d_E\), the competing curvatures steer trajectories in a spiral toward a unique stable interior equilibrium.
 (F) Self-sustained oscillations: in the strong self-excitation regime, the same interaction destabilizes the interior equilibrium and drives the dynamics onto a stable limit cycle rather than a fixed point.}
 \label{fig:wilson-cowan}
\end{figure}

\subsection*{The Wilson-Cowan conceals a zero-sum saddle game}
We demonstrate that the proximal gradient play induced by the WC model admits a transformation that reveals a single cost function that drives a zero-sum game (ZSG). 
A ZSG is a special case of non-cooperative games, usually played between two agents, where the gain of one player is exactly the loss of the other.
Thus, if the first player in a ZSG is seeking to minimize a scalar cost \(\En\), the other player must necessarily be seeking to maximize the same cost \(\En\).
Consider the following ZSG on \([0, 1]^2\) played between an excitatory population and an inhibitory one, with firing rates \(x_E\) and \(x_I\) respectively:
\begin{align}
    \label{eq:EI-2Dim-ZS-game}
    \min_{x_E \in [0, 1]}\max_{x_I \in [0, 1]} \En_\text{tot}&(x_I, x_E, u_E, u_I)\\
    \En_\text{tot}(x_E, x_I,u) =&  \left(\frac{1-w_{EE}}{2w_{EI}}\right)x_E^2 - \left(\frac{1 + w_{II}}{2w_{IE}}\right) x_I^2 \nonumber\\&+ x_Ex_I -\frac{u_E}{w_{EI}}x_E + \frac{u_I}{w_{IE}}x_I.
    \label{eq:EI-2Dim-ZS-cost}
\end{align}

We also define the strategy, a weighted gradient play, which we call the internal dynamics of the E-I populations:
\begin{subequations}\label{eq:strategy_desc/asc}
    \begin{align}
        \dot x_E &= -w_{EI} \cdot \nabla_{x_E} \En_\text{tot}(x_E, x_I, u)\\
        \dot x_I &= ~~~w_{IE} \cdot \nabla_{x_I} \En_\text{tot}(x_E, x_I, u).
    \end{align}
\end{subequations}
This is essentially the proximal gradient play driving a state-space model.
The opposing signs reflect the opposing (minimizing vs maximizing) objectives of each player, indicating that firing patterns are driven solely by a self interest. 
Apart from pursuing self-interest, there is also a constraining action of projecting the firing rates, at each step, onto physically feasible values represented by the set \(\mathcal X\). 
Specifically, the projection step encodes physical constraints on firing rates, such as nonnegativity, and an upper limit due to refractory factors like the ion gate switching and diffusion delays.
This two-step \textit{gradient descent} followed by \textit{projection} resembles the proximal point algorithm, often used as a generalized Nash equilibrium seeking iterator in continuous games~\citep{QXM-JWL:25, NL:13}.

\subsection*{A steep saddle leads to monostability}
For a minimax game such as~\eqref{eq:EI-2Dim-ZS-game}, if the cost function has a saddle point within the action space \([0,1]^2\), the saddle point is necessarily a Nash equilibrium.
One can show that the presented E-I zero-sum scalar cost \(\En_\mathrm{tot}\) is either a concave down surface or a saddle surface.  
The zero-sum saddle structure arises in \eqref{eq:EI-2Dim-ZS-cost} if and only if
\begin{equation}\label{eq:saddle_iff} 
    w_{EE} < 1 + \tfrac{w_{EI}w_{IE}}{1+w_{II}}.\quad\boxed{\text{Existence of Zero-Sum Saddle}}
\end{equation}
The existence of the interior saddle point is just one piece of the puzzle. The frame of the puzzle is given by the \textit{steepness} of the saddle, studied using the notions of strong or weak convex-concavity. 
Strong convex-concavity is a powerful concept, as it extends the notion of strong convexity in optimization to saddle cost functions in minimax games.

The saddle surface is strongly convex-concave if and only if \(w_{EE} < 1\), which we henceforth call the dissipation-excitation dominance condition. The strong convex-concavity of the saddle surface ensures the uniqueness of the Nash equilibrium, and the zero-sum saddle existence inequality~\eqref{eq:saddle_iff} is automatically satisfied under \(w_{EE} < 1\). Therefore, under dissipation-excitation dominance, both the strategy~\eqref{eq:strategy_desc/asc} and the WC dynamics~\eqref{eq:EI-2dim} converge to a unique global equilibrium for every input, and the point of convergence is exactly the saddle as long as it lies within \([0,1]^2\).

\subsection*{A shallow saddle leads to multistability and limit cycles}
Switching from strong to weak convex-concavity, we lose guarantees on the uniqueness of the equilibrium for both the game and the dynamics. 
In particular, under weak convex-concavity we observe the emergence of either multiple attractors or a stable limit cycle under~\eqref{eq:strategy_desc/asc}.
Under certain inputs, the game still admits a unique Nash equilibrium at the saddle point.
The saddle location can be obtained by solving \(\nabla\En_\text{tot}(x_E^*,x_I^*,u_{E},u_{I})=0\), for given \(u = (u_E, u_I)^\top\) where \(x_E^*\) and \(x_I^*\) are the saddle coordinates. 
If the input is chosen outside a set \(\mathcal U_{\mathrm{int}}\), the saddle point still exists, but it would be outside of the physically feasible region $[0,1]^{2}$ and thus inaccessible, with the system converging to multiple equilibria on the boundary instead. 

The weakly convex-concave geometric condition on the saddle is by itself insufficient to determine the dynamical behavior of the system, for which we look to~\eqref{eq:strategy_desc/asc}.
As a matter of fact, the eigenvalues associated with the modes of the internal dynamics undergo a Hopf-like crossover at \(w_{EE} = w_{II} + 2\). For \(w_{EE} > w_{II} + 2\), the system converges to the limit cycle. 
For smaller values, trajectories spiral toward the interior saddle point when it exists.
The latter is the well-known inhibitory-stabilized regime, where the excitatory dynamics would be unstable were it not for the inhibitory stabilization that leads to convergence to the interior saddle instead of the limit cycle.

Finally, when~\eqref{eq:saddle_iff} is violated, the ZSG cost surface is concave down. The minimization of the cost by \(x_E\) implies that the \(x_E\) axis is always either fully activated or inactivated depending on initial conditions.
This regime captures excitation-dominated multistability for all inputs. 

\section*{Lateral inhibition and neural input sensitivity}
Monostability of excitatory-inhibitory circuits is the cornerstone of robust and reliable contrast enhancement. Contrast enhancement in E-I networks exploits a phenomenon known as lateral inhibition, in which a single inhibitory neuron integrates inputs from multiple excitatory neurons and, through recurrent feedback, suppresses their activity~\citep{IJS-SM:11}.  When a set of excitatory neurons receives external inputs that differ only slightly in intensity, the resulting excitation recruits the central inhibitory neuron, which in turn feeds inhibition back onto the entire excitatory ensemble. The recurrent interaction effectively suppresses all but the strongest excitatory response, thereby amplifying small input differences into a discrete, winner-take-all outcome. Thus, the contrast-enhancing dynamics provide a substrate for reliable categorical computation, with the overall suppressive action of lateral inhibition being consistent with the sparse activity patterns observed in neural circuits~\citep{GW-KWM-NR:14}. The key ingredient enabling robust discrimination is monostability, which we define as the global asymptotic stability of the underlying E-I dynamics. In this regime, activity converges to a unique equilibrium for any admissible input, ensuring reliable and reproducible responses. Monostability arises from Lyapunov diagonal stability, a structural condition in network theory that enforces component-wise dissipation and suppresses unstable modes in recurrent systems. In the remainder of this section, we analyze a minimal E-I motif that achieves robust input discrimination under this condition. We then extend the construction to winner-take-all dynamics and, ultimately, to a columnar organization that scales the same stabilizing principle to layered architectures.

\subsection*{E$\mathbf{^{2}}$-I lateral inhibition model}
The simplest realization of a laterally inhibitive E-I network is a firing-rate model comprised of two excitatory neurons receiving external input and a central inhibitory neuron receiving only circuital input, which we will name the E$^{2}$-I-network:
\begin{equation}\label{eq: EIE}
    \begin{pmatrix}
        \dot x_{E_1}\\
        \dot x_{I}\\
        \dot x_{E_2}
    \end{pmatrix} = 
    -\begin{pmatrix}
        x_{E_1}\\
        x_{I}\\
        x_{E_2}
    \end{pmatrix} +
    \left[W
    \begin{pmatrix}
        x_{E_1}\\
        x_{I}\\
        x_{E_2}
    \end{pmatrix} + 
    \begin{pmatrix}
        u_{E_1}\\
        0\\
        u_{E_2}
    \end{pmatrix}
    \right]_0^1
\end{equation}
where the synaptic matrix $W$ is given by
\begin{equation}\label{eq: EIEsyn}
    W = \begin{pmatrix}
        w_{EE} & -w_{EI} & 0\\
        w_{IE} & -w_{II} & w_{IE}\\
        0 & -w_{EI} & w_{EE}
    \end{pmatrix}.
\end{equation}
Notice that, in our setting, the excitatory neurons do not interact directly, but rather in an ancillary manner, through their common inhibitory node. In the simple E$^2$-I motif, monostability is ensured by the simple dissipation-excitation dominance condition
\begin{equation}
    w_{EE}<1\qquad \text{monostability}.
\end{equation}
Neural correlates of successful laterally inhibitive activity are equilibrium states for the system~\eqref{eq: EIE} where either one of the excitatory neurons and the inhibitory neurons are active - namely $x^{\star}_{1}=(1,1,0)$ or $x^{\star}_{2}=(0,1,1)$. The functional condition ensuring the desired neural output is
\begin{equation}
    w_{IE}\geq 1+w_{II} \qquad \text{functionality},
\end{equation}
hence a dominance of inhibitory dissipation and self-inhibition by the excitatory-to-inhibitory input. Under the novel dissipation-excitation dominance and functional conditions, we define the \emph{separation parameter} $\updelta=(1-w_{EE}+w_{EI})>0$. The system~\eqref{eq: EIE} then acts as a $\updelta$-precision lateral inhibition network, a biologically plausible winner-take-all (WTA) mechanism that correctly categorizes all inputs satisfying
\begin{equation}\label{eq: input-sep}
    u_{E_1}\geq\updelta,\quad u_{E_2}\leq-\updelta\qquad \text{(or viceversa)}
\end{equation}
Neural inputs satisfying condition~\eqref{eq: input-sep} might be relayed by specialized pre-processing circuits that center biased stimuli around a zero mean signal for fast and efficient discrimination. Together, these conditions identify a regime in which dissipation controls excitation, while $w_{IE}$ actively recruits inhibition to stabilize selective amplification. Further details, including timescale separation of excitatory and inhibitory dynamics, and energy visualizations, are provided in the SI, section IV.

\subsection*{E$^{k}$-I winner-take-all model}
The E$^{2}$-I architecture can be straightforwardly generalized to a generic E$^{k}$-I network. The E$^{k}$-I model, with $k\geq 2$ excitatory neurons and one inhibitory neuron, implements biologically plausible winner-take-all~\citep{GS:78} dynamics via the firing-rate model. The E$^k$-I inherits the same dissipation-excitation and excitatory-to-inhibitory dominance conditions from the E$^2$-I, hence
\begin{align}
    w_{EE}&<1\qquad\qquad\quad \text{monostability},\\
    w_{EI}&\geq 1+w_{II}\qquad \text{functionality}.
\end{align}
In the general setting of a E$^k$-I circuit, an excitatory input $u\in\real^k$ is $\updelta$-precision laterally inhibitive whenever $u_{E_i}\geq\updelta$ for $i=1,\dots,k$ and $u_{E_j}\leq-\updelta$ for all $j\neq i$. The generalization of WTA to an arbitrary number of excitatory neurons provides solid ground to understand the sparsification of activity mediated by inhibitory interneurons in columnar strata.

In general, the ability of the E$^{k}$-I circuit to discriminate between competing inputs is intrinsically tied to its $\updelta$-precision threshold. At first glance, such finite precision might appear as a severe limitation; yet it is unlikely that biological nervous systems rely on rigid mechanisms with little capacity to adapt to weak or noisy signals. Instead, this apparent constraint invites a higher-level organizational solution. In the next section, we show that a columnar arrangement of multiple E$^{k}$-I microcircuits overcomes finite-precision limits by progressively amplifying small input differences. This hierarchical architecture enables reliable discrimination of arbitrarily weakly separated signals and allows us to explicitly estimate the number of layers required for categorization as a function of circuit precision and input separation. 
  
\begin{figure}[!tbph]
    \centering
    \subfloat{\includegraphics[width=\linewidth]{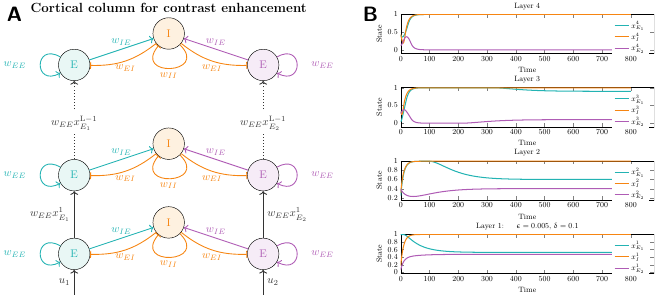}}

    \caption{\textbf{Schematic of a column of E$^{2}$-I networks and layer-wise forward contrast enhancement.} (A) Vertical arrangement of E$^{2}$-I modules, each forwarding its excitatory activity to the layer above. In the first layer, the two excitatory neurons receive distinct external inputs, $u_{E_1}$ and $u_{E_2}$. Within the $\updelta$-FLI subthreshold regime, the layer converges to a Nash equilibrium $(x_{E_1}^{1,*},x_{E_2}^{1,*})$ that lies between the categorical states $(1,0)$ and $(0,1)$. These equilibrium activities are then propagated to the next layer. If the weighted activity difference satisfies $w_{EE}|x_{E_1}^{1,}-x_{E_2}^{1,}|<\updelta$, the subsequent layer also remains in a subthreshold regime, and the process continues. Once the accumulated difference exceeds the threshold ($w_{EE}|x_{E_1}^{\mathrm{L}-1,*}-x_{E_2}^{\mathrm{L}-1,*}|\geq\updelta$), the next layer sharply discriminates the two input channels. (B) Evolution of excitatory activities across four successive E$^{2}$-I layers (from bottom to top). Each layer operates with $\updelta=0.1$, has self-excitatory weights $w_{EE}=0.8$, and the first layer receives external inputs separated by $\upepsilon=u_{E_2}-u_{E_1}=0.005$. The difference between the equilibrium activities of the two excitatory neurons progressively amplifies across layers, eventually surpassing the $\updelta$-FLI threshold in the third layer and reaching categorical separation in the fourth, demonstrating hierarchical contrast enhancement. In particular, notice how for the chosen parameters the layer depth estimate is $\mathrm{L} = \lceil 3.16\rceil = 4$.}
    \label{fig: Cort}
\end{figure}

\subsection*{Lateral inhibitory columns for contrast enhancement}
The E$^{k}$-I lateral inhibition network is a simple biologically plausible model of contrast enhancement in excitatory-inhibitory motifs, but relies on the assumption of the incoming inputs being separated enough to allow categorical discrimination. Yet, earlier nervous streams may receive inputs that are subthreshold with respect to the precision of the network, and therefore fail in the categorical detection of contrasts. Furthermore, a unique E$^{k}$-I network able to detect very subtle differences in the inputs would require a careful fine-tuning of the synaptic parameters, thereby imposing stringent structural conditions. Instead, real neuronal networks are characterized by notable structural variability~\citep{ME-GJM:06}, and theoretical studies suggests that different combinations of synaptic parameters result in the same network phenomenology~\citep{PAA-BD-ME:04}. A plausible solution to the conflict between neural biology and mathematical modeling for contrast enhancement comes in the form of hierarchical (columnar) organization of neural circuits.

The hierarchical organization of neural circuits suggests that contrast enhancement is not a single-stage operation, but a computation that recurs across layers and columns. Each lamina in a column can be abstracted as an E$^{k}$-I microcircuit, in which a central inhibitory neuron mediates lateral inhibition among neighboring excitatory cells. This motif closely matches the organization of local parvalbumin-expressing interneurons, which receive dense excitation and provide broad, fast inhibition that normalizes population activity within a layer~\citep{ABV-BW-CM:12}. When arranged in a vertical cascade, the E$^{k}$-I layers form a hierarchical contrast-enhancement network: each layer receives the excitatory activity of the one below, adjusts the inhibitory bias to the average incoming input, and transmits a sharpened, sparser representation upward. This architecture aligns with canonical microcircuit descriptions~\citep{DRJ-MKAC:04} and explains the progressive sparsification and selectivity observed across cortical layers~\citep{RM-PT:99, VWE-GJL:00}. Moreover, it provides a biological interpretation for normalization modules in artificial neural networks, which similarly implement recurrent inhibitory feedback via divisive normalization to stabilize dynamics and improve feature discrimination~\citep{BMF-CSA-DGH:21}. In this view, columnar neuronal circuits become a stack of interacting E$^{k}$-I games (see Fig.~\ref{fig: Cort}(A)), each layer refining the contrast and selectivity of the one beneath, thereby enabling robust hierarchical computation under biological constraints.

Consider a column composed of $\mathrm{L} \in \mathbb{N}$ layers, each representing for notational simplicity an E$^{2}$-I network with state $\mathrm{x}^{l} = (x_{E_{1}}^{l}, x_{I}^{l}, x_{E_{2}}^{l})$, for $l = 1, \dots, \mathrm{L}$. The dynamics of the $l$-th lamina can be written compactly as
\begin{equation}
    \dot{\mathrm{x}}^{l}=-\mathrm{x}^{l}+[W\mathrm{x}^{l}+U\mathrm{x}^{l-1}]_{0}^{1},
\end{equation}
where $W$ denotes the intra-laminar synaptic matrix~\eqref{eq: EIEsyn}, $U = \mathrm{diag}(w_{EE}, 0, w_{EE})$ the inter-laminar input matrix, and $\mathrm{x}^{0} = (u_{E_{1}}/w_{EE}, 0, u_{E_{2}}/w_{EE})$ encodes the external drive to the bottom layer. As implied by the structure of $U$, feedforward propagation occurs exclusively among excitatory neurons and exploits compensatory inhibitory biases to center the input between layers; instead, the inhibitory feedback remains local to each layer. Each lamina of the neuronal column functions as a $\updelta$-precision lateral inhibitory network under the synaptic parameter constraints
\begin{align}
    w_{EE}&<1/2\qquad\qquad \text{monostability},\nonumber\\
    w_{EI}&\geq 1+w_{II}\qquad \text{functionality}.\nonumber
\end{align}
When the incoming external inputs are subthreshold, i.e., $\upepsilon\leq u_{E_1}<\updelta$ and $-\updelta<u_{E_2}\leq-\upepsilon$ (or viceversa) for $\upepsilon>0$, the layered architecture acts as an amplifier of discriminability: feedforward excitation progressively sharpens the difference between the two input streams until, at a sufficient depth, one layer reaches the WTA regime and performs categorical selection (see Fig.~\ref{fig: Cort}(B)). The minimal number of layers required for successful discrimination can be computed analytically as  the ceiling
\begin{equation}\label{eq: lay}
    \mathrm{L}=  1+\Bigg\lceil\frac{\ln\left(\frac{\upepsilon}{\updelta}\right)}{\ln\left(\frac{1-w_{EE}}{w_{EE}}\right)}\Bigg\rceil \quad{\text{layers depth}}
\end{equation}
This estimate is positive whenever the microcircuits satisfy the dissipation–excitation dominance conditions. Importantly, the number of required layers increases as the input difference $\upepsilon$ decreases, indicating that deeper columns achieve finer categorical precision. Conversely, we directly recover the single-layer $\updelta$-precision WTA E$^{2}$-I network in the limit
\begin{equation}
    \mathrm{L}\xrightarrow[]{\upepsilon\to\updelta}1.
\end{equation}
Therefore, the mismatch between the desired precision $\upepsilon$ and the per-layer achievable precision $\updelta$ effectively links column depth to the computational resolution of neuronal input discrimination.

\section*{Conclusions}

\subsection*{A novel framework for \emph{asymmetric} neural dynamics}
In this manuscript, we have proposed a new biologically plausible framework that unveils the competitive nature of neural interactions in \emph{asymmetric} neural networks. Starting from well-known examples of \emph{symmetric} neural networks in the literature, we showcased a powerful connection from canonical energy-based methods in physics and Lyapunov stability in dynamical systems to proximal gradient methods in optimization and machine learning.
In particular, we have proved how the energy function of \emph{symmetric} firing-rate networks coincides with a composite cost function, arising from the associated proximal gradient dynamics, that includes both the global cost of interaction with other neurons and the aggregate cost of single-neuronal activation. We then extend proximal gradient dynamics to \emph{asymmetric} networks, associating to each neuron its own energy decomposed into the cost of interaction with the other neurons and the cost of self-activation. The joint optimization of the cost functions associated to all the neurons in the network instantiate a game where the neurons act as selfish agents trying to minimize their own energy.

The novel energetic game-theoretic framework has enabled us to provide a new perspective on canonical models from the theoretical neuroscience literature, with a focus on the strategic value of single neuron dynamics. In particular, we showed how the known results about the dynamical regimes of the WC model admit a direct energetic interpretation: global asymptotic convergence towards a unique state regardless of the ambient input is a result of both the excitatory and inhibitory populations converging to a steep zero-sum game saddle. On the other hand, as the system parameters lower the saddle steepness, indecision in the form of limit cycle s arises on the zero-sum game, showing weak competition between the neurons. 

In the same spirit, we study lateral inhibition dynamics for networks of many excitatory neurons and one inhibitory neuron producing enhancement of contrasts in the external input. Specifically, we showed how the dynamics of E$^{2}$-I networks converge to a state representing a discrete categorization of the contrasts in the external input, which also coincides with the unique Nash equilibrium of the game played by the neurons. The finite precision of a single E$^{2}$-I network can be enhanced by combining many of them in a lateral inhibition column organization to progressively sharpen subtle contrasts in the input, reaching a state of discrete categorization in a finite number of layers, for which we provide a closed-form expression. Finally, in the SI we provide a full characterization of stability for generic E-I systems with homogeneous weights and arbitrary network topology, establishing a tight relationship between control-theoretic, graph-theoretic, and neuroscientific concepts.

\subsection*{Mechanism design and randomness}
Starting from known \emph{asymmetric} firing-rate dynamics, we have derived a latent mathematical structure that frames neural activity as a game among interacting neurons, each pursuing its own strategy. Yet, this is only one side of the coin. The other side reveals the complementary field of mechanism design~\citep{GA:73, HL-RS:06} - often described as reverse game theory - which focuses on shaping costs and payoffs to elicit desired collective behaviors. Mechanism design has already transformed domains involving large populations of strategic agents, from online auctions~\citep{DF-GM-PA:20} to voting systems~\citep{NN-RT-TE:07}. Translating these principles into neuroscience could open a powerful new line of inquiry: engineering neural “games” by tailoring neuronal energy functions to yield specific dynamic or computational outcomes. Such a framework would provide a systematic route to synthesize neural architectures that exhibit, by design, desired regimes such as selective amplification, oscillation, or competition.

In addition, further works could leverage statistical tools to study excitatory-inhibitory networks for contrast enhancement in more general contexts. Our results on the dynamic regimes of excitatory-inhibitory networks hinged on the simplifying assumption of weights homogeneity - all the synapses of the same type (e.g., excitatory-to-inhibitory) have the same strength. Allowing instead heterogeneous but sign-constrained synapses would enable a richer description of how distributed parameter variability shapes emergent phenomena. By generalizing the synaptic structure of E$^{k}$-I networks, one could map how sets of synaptic weight configurations govern transitions between winner-take-all, oscillatory, and chaotic regimes - much as phase transitions describe order and disorder in physical systems~\citep{KJ-SH:15, MF-OS:17}. In doing so, statistical methods could bridge microscopic variability with macroscopic computational function, bringing the theory of neural dynamics closer to both biological realism and designable complexity.

\subsection*{Engineering your way into neuroscience}
Neuroscientists observe, physicists explain, and engineers build. In recent years, these three modes of inquiry have increasingly converged to study neural systems. The nervous system is not only an object of observation or explanation - it is a dynamical machine, whose balance, adaptability, and computation can be engineered. In this work, Lyapunov diagonal stability provides a powerful example of how an engineering concept can illuminate a biological principle. It formalizes the delicate equilibrium between excitation and inhibition, revealing how networks maintain stability while amplifying and categorizing information.

In excitatory-inhibitory circuits, stability is more than a static property: it shapes how neurons interact, compete, and cooperate to extract structure from the environment. Expressing this principle in mathematical terms explains existing observations and shows how such dynamics could, in principle, be engineered. The convergence of scientific modes suggests a deeper partnership between disciplines: control theory and dynamical systems offer the rigorous tools to describe robustness and adaptability, while neuroscience provides the architectures and constraints that make those tools meaningful in living systems. And where engineering seeks structure, physics generalizes it to randomness - showing how even in the presence of noise and heterogeneity, order and computation can emerge through collective dynamics.

To fully understand the brain, we may need to learn to build with its principles. Stability, feedback, and balance become the grammar of computation itself. Engineering our way into neuroscience means recognizing that the same laws that make artificial systems stable can also help us understand how biological systems \emph{think}.

\paragraph{Acknowledgments:} This work was supported in part by ARO MURI grant W911NF-24-1-0228 and by Next Generation EU grant C96E22000350007.

\bibliographystyle{unsrtnat}
\bibliography{SB-main,FB,SI}  

\newpage
\section*{\Huge Supporting information}

\tableofcontents
\startcontents[sections]

\section{Preliminaries}

\subsection{Notation}
Denote \(\mathbb R^n\) as the \(n\)-dimensional real Euclidean space, \(\mathbb R^{m\times n}\) as the space of all real valued \(m\times n\) matrices, \(\mathbb S^n \subset \mathbb R^{n\times n}\) as the space of all real symmetric matrices, \(\mathbb S^n_{\succ0} \subset \mathbb S^n\) and \(\mathbb S^n_{\succeq0} \subset \mathbb S^n\) as positive definite and positive semidefinite matrices respectively, with \(\mathbb S^n_{\prec 0} \subset \mathbb S^n\) and \(\mathbb S^n_{\preceq 0} \subset \mathbb S^n\) being their negative definite and negative semidefinite matrices respectively.
Let \(C \subseteq \mathbb R^n\) be a set.
Denote the interior of \(C\) by \(\operatorname{int} C\).
For a real function \(f\) denote \(\operatorname{dom} f\) as the domain of \(f\).
We denote by $C^{k}(\mathcal{X};\mathcal{Y})$ the class of $k$-differentiable functions from $\mathcal{X}$ to $\mathcal{Y}$. 
We denote with $l$-Lipschitz the class of real vector functions that are Lipschitz with constant $l>0$. 
Let $f\in C^{1}(\real^{n};\real)$ and denote by $\nabla f(x)\in\real^{n}$ its gradient.
We denote the partial derivative of $f$ with respect to $x_{i}$ as $\nabla_{x_{i}}f(x)$. 
Let $g\in C^{1}(\real^{n};\real^{n})$ and denote with $\mathcal D_{x}g(x)\in\real^{n\times  n}$ its Jacobian.
We denote the identity matrix of dimension $n\in\mathbb N$ as $\mathcal{I}_{n}\in\real^{n\times n}$. 
We use  $A\succ 0$ to denote that $A\in\real^{n\times n}$ is a positive definite matrix. 
Let $B\in\real^{n\times n}$ and we denote its trace operator as $\tr(B)$. 
We denote the weighted inner product as \(\langle\cdot ,\cdot \rangle_{A}:\real^{n} \times \real^{n}\to \real\), for $A\succ 0$; if $A=\mathcal{I}_n$ we recover the standard Euclidean inner product, denoted by $\langle \cdot , \cdot \rangle_{2}$ or simply $\langle \cdot , \cdot \rangle$.
The weighted Euclidean 2-norm (for vectors, or induced, for matrices) is given by \(\lVert \cdot \rVert_A\) for $A\succ 0$; if \(A = \mathcal I_n\), we recover the standard Euclidean 2-norm which we denote by \(\lVert \cdot \rVert_2\) or simply \(\lVert \cdot \rVert\).
The Frobenius norm of a matrix is denoted by \(\lVert\cdot \rVert_F\).
Denote by \([n]\) the set of positive integers \(\{1, \dots , n\}\) for some \(n>0\).
Let \(\vectorzeros[n]\) and \(\vectorones[n]\) be the \(n\)-dimensional all-zeros and all-ones vectors respectively.
Denote the action of the linear-threshold operator on \(z \in \mathbb R^n\) by \([z]_0^1 \coloneqq \min(1, \max(z))\), where the \(\min\) and \(\max\) act elementwise.
For a vector \(v \in \mathbb R^n\), \(\operatorname{diag}(v)\) represents the diagonal matrix with \(v\) as the diagonal, while \(\operatorname{antidiag}(v)\) represents an antidiagonal matrix with \(v\) as the antidiagonal.
When operating on a square matrix \(A\), \(\operatorname{diag}(A)\) should be understood to output the diagonal vector.

\subsection{Game theory}
In contrast to optimization theory, where one tries to minimize a single scalar-valued function, game theory considers the problem of ``simultaneously minimizing several coupled scalar-valued functions.'' 
Given $N$ agents or players, the \(i\)th agent has a cost function $J_i$ which it tries to minimize.
However, each $J_i$ also depends on the actions of the other agents, so, in general, it is not possible to find a set of actions that globally minimizes all $J_i$ simultaneously. 
Instead, the strategy used by each agent plays a key role in deciding what conditions constitute a solution. 
In this work, we use the notion of Nash equilibrium as the ``solution'' to the game. 
We introduce the definition of Nash equilibrium and draw connections to variational inequalities. 
The interested reader is referred to the following two books~\citep{EAA-MGF:23, SI:25} for game-theoretic reference and to~\citep{RTR:17} for the variational inequality perspective of Nash solutions, for a deeper treatment of each subject.

\subsubsection{Best response and Nash Equilibria}
An agent is said to play `best response' when it chooses the action that minimizes its cost in response to each possible action played by the other players.
We make clear the notion of best response in the following definitions under the setup: in a game of \(n\) players let the \(i\)th agent with choice \(x_i \in \mathcal U_i\) and cost \(J_i(x_i, x_{-i})\) playing against a set of agents \(X_{-i}\) with choices \(x_{-i} \in \mathcal U_{-i}\) for \(i \in [n]\).
Here, the subscript \(-i\) on a quantity indicates that the quantity pertains to all players except the \(i\)th player.
Further, let \(\mathcal U_i\) and \(\mathcal U_{-i}\) be compact sets and let \(J_i(x_i, x_{-i})\) be a locally bounded smooth function.

\begin{definition}[Nash equilibrium]
    Nash equilibrium is a game solution such that no \(i\)th player can decrease their cost \(J_i\) by deviating from it unilaterally, i.e., if the Nash equilibrium is given by \(x^*\), it holds that
    \[J_i(x_i^*, x_{-i}^*) \leq J_i(x_i, x_{-i}^*)\]
    for every \(i\)th player, for all \(x_i \in \mathcal U_i\).
\end{definition}

\begin{definition}[Local Nash equilibrium]
A strategy profile $x^*=(x_1^*,\dots,x_N^*)\in \mathcal U_1\times\cdots\times \mathcal U_N$ is a \emph{local} Nash equilibrium if, for each $i$th player, there exists a neighborhood $U_i\subseteq \mathcal U_i$ of $x_i^*$ such that
\[J_i(x_i^*,x_{-i}^*) \le J_i(x_i,x_{-i}^*) \qquad \forall\, x_i\in U_i .\]
\end{definition}

\subsubsection{Zero-sum games and saddle points}
A two-player zero-sum game (ZSG) is a non-cooperative game in which the sum of the cost functions of the two players is identically constant (without loss of generality, zero), which implies that the benefit of one player is strictly the loss of the other player.
\begin{definition}[Saddle point equilibrium]
    A saddle point equilibrium (SPE) for a two-player ZSG is given by the actions that minimize the objective \(J\) for one player (say player 1), and maximize it (or minimize \(-J\)) for the other player, ie. if the SPE is given by \((\hat x_1, \hat x_2)\), it holds that
    \[J(\hat x_1, x_2) \leq J(\hat x_1, \hat x_2) \leq J(x_1, \hat x_2),\]
    for all \(x_i \in \mathcal U_i\), \(i \in \{1, 2\}\).
\end{definition}

\subsection{Graph theory}

Graph theory is the study of graphs, which are mathematical structures used to model pairwise relations between objects. It has wide applications in all applied science, among which we list neuroscience, computer science, biology, social sciences, and many other fields. We hereby provide a concise account of the basic concepts from graph theory. The interested reader is referred to the following references~\citep{WDB:01, DR:25} for a more in-depth treatment.

\subsubsection{Basic Definitions}

\begin{definition}[Graph]
A \textbf{graph} is an ordered pair $(\mathcal{V}, \mathcal{E})$, where:
\begin{itemize}
    \item $\mathcal{V}$ is a non-empty set of \textbf{vertices} (also called \textbf{nodes}),
    \item $ \mathcal{E} \subseteq \{ (i,j) \mid i, j \in \mathcal{V}, i \neq j \} $ is a set of \textbf{edges}, where each edge connects a pair of distinct vertices.
\end{itemize}
\end{definition}

\begin{definition}[Adjacent Vertices and Degree]
Two vertices $ i $ and $ j $ are \textbf{adjacent} if $(i, j) \in \mathcal{E} $.
\end{definition}

\begin{definition}[Simple Graph]
A graph with no loops (edges from a vertex to itself) and no multiple edges (more than one edge connecting two vertices) is called a \textbf{simple graph}.
\end{definition}

We now provide a general definition that bridges a generic dynamical system to its underlying graph, by means of the adjacency matrix.

\begin{definition}[Adjacency matrix of a dynamical system]
Let $x(t) \in \mathbb{R}^N$ evolve according to
\[
\dot{x} = F(x),
\]
where each component $F_i(x)$ may depend on a subset of variables $\{x_j\}_{j=1}^N$.  
The \emph{adjacency matrix} $A \in \mathbb{R}^{N \times N}$ of this system is defined by
\[
A_{ij} =
\begin{cases}
1, & \text{if } F_i \text{ explicitly depends on } x_j,\\[4pt]
0, & \text{otherwise.}
\end{cases}
\]
Thus, $A_{ij}$ encodes the presence, direction, and (when evaluated at a reference state) local strength of the direct influence of unit $j$ on unit $i$.  
The adjacency matrix therefore provides the formal link between the vector field $F$ and the underlying directed graph of dynamical interactions. Specifically
\begin{equation}
    A_{ij}=1\iff (j,i)\in\mathcal{E}\quad\land\quad A_{ij}=0\iff(j,i)\notin\mathcal{E}.
\end{equation}
\end{definition}

\subsection*{Example}
A simple undirected graph~\ref{fig: und-graph}:
\begin{figure}[!tbh]
    \begin{center}
        \begin{tikzpicture}[scale=1.2]
            \node[draw=black, fill=LimeGreen, circle] (A) at (0,0) {A};
            \node[draw=black, fill=LimeGreen, circle] (B) at (2,0) {B};
            \node[draw=black, fill=LimeGreen, circle] (C) at (1,1.5) {C};
            \draw (A) -- (B);
            \draw (A) -- (C);
            \draw (B) -- (C);
        \end{tikzpicture}
    \caption{\textbf{Example of undirected graph.} The edges connecting nodes $A,\ B,\ C$ have no directionality, and consequently each edge can be traversed in both directions.}
    \label{fig: und-graph}
    \end{center}
\end{figure}

We now provide a few definitions of graphs that usually underlie different models in theoretical neuroscience and machine learning, and contextualize them with examples.

\begin{itemize}
    \item \textbf{Directed Graph (Digraph)}: Edges have a direction, meaning that $(i,j)\in\mathcal{E}$ does not imply that $(j,i)\in\mathcal{E}$. An excitatory-inhibitory neural network constitutes an example of directed graph.
    \item \textbf{Undirected Graph}: Edges are unordered pairs \( \{i,j\} \) such that $(i,j)\in\mathcal{E}$ if and only if $(j,i)\in\mathcal{E}$. The famous Hopfield network is an example of undirected graph, since the associated adjacency matrix $A$ is symmetric.
    \item \textbf{Weighted Graph}: Each edge $(i,j)$ has an associated numerical value (weight) $\mathcal{W}(i,j)=w_{ij}$, with $\mathcal{W}:\mathcal{E}\to\real$.
    \item \textbf{Complete Graph}: Every pair of distinct vertices $i,j\in\mathcal{V}$ is connected by an edge.
    \item \textbf{Bipartite Graph}: Vertices can be divided into two sets such that no edge connects vertices of the same set. A known example is Restricted Boltzmann Machines in generative artificial intelligence.
\end{itemize}

\subsubsection{Directed Graphs and Degree Concepts}

\begin{definition}[In-Degree and Out-Degree]
Let $ (\mathcal{V}, \mathcal{E}) $ be a directed graph and $ i \in \mathcal{V} $.
\begin{itemize}
    \item The \textbf{in-degree} of $ i $, denoted $ \deg^{-}(i) $, is the number of edges of the form $ (j, i) \in \mathcal{E} $.
    \item The \textbf{out-degree} of $ i $, denoted $ \deg^{+}(i) $, is the number of edges of the form $ (i, j) \in \mathcal{E} $.
\end{itemize}
\end{definition}

\begin{theorem}
In any directed graph,
\[
\sum_{i \in \mathcal{V}} \deg^+(i) = \sum_{i \in \mathcal{V}} \deg^-(i) = |\mathcal{E}|
\]
That is, the total number of incoming edges equals the total number of outgoing edges, which equals the number of edges in the graph.
\end{theorem}

\subsection*{Example}
\begin{figure}[!tbh]
    \begin{center}
        \begin{tikzpicture}[->, >=stealth, node distance=2.5cm]
        \node[circle, draw=black, fill=LimeGreen] (A) at (0,0) {A};
        \node[circle, draw=black, fill=LimeGreen] (B) at (3,0) {B};
        \node[circle, draw=black, fill=LimeGreen] (C) at (1.5,2) {C};
        \draw (A) edge[bend left] (B);
        \draw (B) edge[bend left] (A);
        \draw (A) -- (C);
        \draw (C) -- (B);
        \end{tikzpicture}

    \caption{\textbf{Example of a directed graph.} The edges connecting nodes $A,\ B,\ C$ have specified directionality, and can be traversed only according to the paths specified by the arrows.}
    \label{fig: dir-graph}
    \end{center}
\end{figure}

In the digraph above~\ref{fig: dir-graph}:
\begin{itemize}
    \item \( \deg^+(A) = 2 \), \( \deg^-(A) = 1 \)
    \item \( \deg^+(B) = 1 \), \( \deg^-(B) = 2 \)
    \item \( \deg^+(C) = 1 \), \( \deg^-(C) = 1 \)
\end{itemize}
Total in-degrees and out-degrees are both equal to 4, matching the number of edges.

\section{Game-theoretic interpretation of firing rate (FR) dynamics}
\label{sec:game-theory}
In this section, we will focus on a recent interpretation of the firing rate dynamics for a symmetric synaptic matrix, and we will then extend it to the asymmetric case. We start by considering a generic firing rate system of the form
\begin{equation}\label{eq: fr-prx}
    \begin{cases}
        \dot{x}=-Dx+\Phi(Wx+u)\\
        x(0)=x_{0}\in\real^{N}
    \end{cases}
\end{equation}
where $D\in\real^{N\times N}$ is the dissipation matrix, which is diagonal and positive definite $D\succ 0$, $W\in \real^{N\times N}$ is the synaptic matrix, $\Phi:\real^{N}\to\real^{N}$ is the activation function and $u\in\mathcal{U}\subset\real^{N}$ is the external input. 

\begin{remark}
    In the following treatment, for notational simplicity we will resort to the classic notation of optimization theory. Namely, the component of the minimization cost will be
    \begin{itemize}
        \item $J(x,u)\leftarrow \En_{\text{int}}(x,u)$ the new notation for the interaction cost;
        \item $\Gamma(x)\leftarrow \En_{\text{act}}(x)$ the new notation for the activation cost.
    \end{itemize}
\end{remark}

\subsection{Proximal gradient dynamics for symmetric FR networks}
It was proved in~\citep{BS-BG-BF-ZS:25fr} that for a symmetric synaptic matrix $W\in\real^{N\times N}$ and constant input $u\in\real^{N}$, the FR dynamic~\eqref{eq: fr-prox_si} admits a global energy (Lyapunov) function
\begin{equation}\label{eq: En-gen}
    \En(x,u)=-\frac{1}{2}x^{\top}Wx-x^{\top}u+\sum_{i=1}^{N}\int_{0}^{x_{i}}\upphi_{i}^{-1}(D_{ii}z)\ dz
\end{equation}
guaranteeing convergence to the equilibria of the system from any point in state-space. It was also shown in~\citep{VC-AG-AD-GR-FB:23a, AG-AD-FB:24d} that the firing rate dynamics~\eqref{eq: fr-prx} for a neural network with symmetric synaptic matrix $W\in\real^{N\times N}$ can be rewritten as proximal gradient dynamics 
\begin{equation}\label{eq: fr-prox_si}
    \begin{cases}
        \dot x = -Dx+\prox_{\Gamma}(Dx-\nabla J(x,u))\\
        x(0)=x_{0}\in\real^{N}
    \end{cases}
\end{equation}
for a cost function $J:\real^{N}\times\mathcal{U}\to\real$ quantifying the interaction cost among neurons in the network and a scalar function
\begin{align}\label{eq: Gamma}
    \Gamma(x)&=\sum_{i=1}^{N}\upgamma_{i}(x_{i})\\
    \upphi_{i}(x_{i})&=\prox_{\upgamma_{i}}(x_{i})
\end{align}
quantifying the activation cost associated to the neurons of the network. Furthermore, it has been proven in~\citep{VC-AG-AD-GR-FB:23a,AG-AD-FB:24d} that the dynamics~\eqref{eq: fr-prox_si} minimize, for any fixed $u\in\mathcal{U}$, the composite cost function
\begin{equation}\label{eq: composite-cost}
    J(x,u)+\Gamma(Dx)
\end{equation}
Consequently, any $x^{\star}\in\real^{N}$ such that $-Dx^{\star}+\Phi(Wx^{\star}+u)=\vectorzeros[N]$ and such that $x(t)\to x^{\star}$ satisfies $x^{\star}=\arg \min_{x\in\real^{N}} J(x,u)+\Gamma(Dx)$. 

We next exploit the equivalence of the two formulations~\eqref{eq: fr-prx} and~\eqref{eq: fr-prox_si} of the firing rate dynamics to prove that $\En(x,u)=J(x,u)+\Gamma(Dx)$ in the special case when $D=\mathcal{I}_{N}$.
\begin{proposition}[\textbf{Uniqueness of the cost for proximal gradient FR dynamics}]\label{prop: Ener-uni}
    Consider the firing rate dynamics~\eqref{eq: fr-prx} for a symmetric synaptic matrix $W=W^{\top}\in\real^{N\times N}$, and let $\En:\real^{N}\times\mathcal{U}\to\real$ be the associated energy function defined in~\eqref{eq: En-gen}. Consider the equivalent firing rate dynamics~\eqref{eq: fr-prox_si} and the associated cost~\eqref{eq: composite-cost} that is minimized by the proximal gradient dynamics. Then for $D=\mathcal{I}_{N}$ and fixed $u\in\mathcal{U}$
    \begin{equation}
        \En(x,u)= J(x,u)+\Gamma(x)
    \end{equation}
\end{proposition}
\begin{proof}
    We already have an explicit expression for the firing rate energy~\eqref{eq: En-gen}, while both the interaction cost $J(x,u)$ and the activation cost $\Gamma(x)$ only have an abstract definition. Exploiting the equivalence between the two formulation of the models, since $\prox_{\Gamma}(x)=\Phi(x)$ it must hold that
    \begin{equation}
        x-\nabla J(x,u)=Wx+u
    \end{equation}
    or equivalently
    \begin{equation}\label{eq: inter-cost-cond}
        \nabla_{x} J(x,u) = x-Wx-u
    \end{equation}
    Since $W=W^{\top}$, then $\nabla J(x,u)$ is integrable, and we obtain that
    \begin{equation}
        J(x,u) = \frac{1}{2}x^{\top}x-\frac{1}{2}x^{\top}Wx-x^{\top}u.
    \end{equation}
    Resorting now to the definition of proximal operator, we have that it must hold
    \begin{equation}
        \Phi(x)=\prox_{\Gamma}(x)=\arg \min_{z\in\real^{N}} \Gamma(z)+\frac{1}{2}\|x-z\|_{2}^{2}
    \end{equation}
    Considering the ansatz $\Gamma(z)=\sum_{i=1}^{N}\int_{0}^{z_{i}}\upphi_{i}^{-1}(s)\  ds-\frac{1}{2}\|z\|^{2}_{2}$ and differentiating w.r.t. $z\in\real^{N}$, we obtain
    \begin{align}
        \vectorzeros[N] &= \nabla_{z}\left[\sum_{i=1}^{N}\int_{0}^{z_{i}}\upphi_{i}^{-1}(s)\  ds-\frac{1}{2}\|z\|^{2}_{2}+\frac{1}{2}\|x-z\|^{2}_{2}\right]\nonumber\\
        &=\Phi^{-1}(z)-z-(x-z)\nonumber\\
        &=\Phi^{-1}(z)-x
    \end{align}
    which gives $z=\Phi(x)$, as desired. Thus, we finally observe that
    \begin{align}
        J(x,u)+\Gamma(x)&=\frac{1}{2}x^{\top}(\mathcal{I}_{N}-W)x-x^{\top}u+\sum_{i=1}^{N}\int_{0}^{x_{i}}\upphi^{-1}_{i}(z)\ dz-\frac{1}{2}\|x\|_{2}^{2}\nonumber\\
        &=-\frac{1}{2}x^{\top}Wx-x^{\top}u+\sum_{i=1}^{N}\int_{0}^{x_{i}}\upphi^{-1}_{i}(z)\ dz=\En(x,u) 
    \end{align}
\end{proof}
In the general case in which $D\neq \mathcal{I}_{N}$, the energy~\eqref{eq: En-gen} and the cost function associated to the proximal gradient dynamics do not coincide. Indeed, it is easy to observe that in the general case
\begin{equation}
    J(x,u)+\Gamma(Dx)=\frac{1}{2}x^{\top}(D-D^{2})x-\frac{1}{2}x^{\top}Wx-x^{\top}u+\sum_{i=1}^{N}\int_{0}^{D_{ii}x_{i}}\upphi_{i}^{-1}(z)\ dz .
\end{equation}

\subsection{Proximal gradient-play dynamics for asymmetric FR networks}\label{sec:PxGPD}
We now extend the previous treatment to the case of an asymmetric synaptic matrix $W\neq W^{\top}$, where in general we lose the energy characterization of the firing rate dynamics. The single neuron dynamics~\eqref{eq: fr-prx} can then be rewritten, analogously to the symmetric case, as
\begin{equation}
    \dot{x}_{i}=-D_{ii}x_{i}+\prox_{\upgamma_{i}}(D_{ii}x_{i}-\nabla_{x_{i}} J_{i}(x_{i},x_{-i},u_{i}))
\end{equation}
where $J_{i}:\real\times \real^{N-1}\times\real\to\real$ is the interaction cost of neuron $i$, is smooth and bounded on some compact set \(\mathcal X\) for each \(i \in [n]\), and $\upgamma_{i}:\real\to\real$ is the activation cost of the single neuron, and is closed, convex and proper, satisfying
\begin{equation}
    \prox_{\upgamma_{i}}(x)=\upphi_{i}(x)
\end{equation}
We assume the dynamics solve a joint optimization problem, an \(n\)-player game given by
\begin{equation}\label{eq:optim_prob}
    x_i^* = \argmin{x_i \in \mathbb R}{J_i(x_i, x_{-i},u_{i}) + \upgamma_i(D_{ii}x_i)},
\end{equation} 
The solution to~\eqref{eq:optim_prob} will satisfy first-order necessary conditions given by the inclusion: find \(x^* = [x_1^*, \dots, x_n^*]^\top\) such that
\begin{equation}\label{eq:inclusion_prob}
    \vectorzeros[N] \in \tilde\nabla_{x} J(x^*,u) + \partial \Gamma (Dx^*),
\end{equation}
Notice that \(\tilde \nabla J = [\nabla_{x_{1}} J_1, \dots, \nabla_{x_{N}} J_N]^\top\) is not necessarily the gradient of a single scalar function and is therefore referred to as the proximal gradient-play (PxGP) operator.
Following similar steps to~\citep{SHM-MRJ:21pg}, we add and subtract \(Dx^*\) in \eqref{eq:inclusion_prob}, and since \(\operatorname{prox}_\Gamma = (\operatorname{Id} + \partial \Gamma)^{-1}\) is single-valued~\citep{NP-SB:13pa, RTR:76}, we get
\begin{subequations}
\begin{align}\label{eq:inclusion_rearr}
    \vectorzeros[N] &\in -Dx^* + \tilde\nabla J(x^*,u) + (\partial \Gamma + \operatorname{Id})(Dx^*)\\
    Dx^* &= \operatorname{prox}_\Gamma(Dx^* - \tilde\nabla J(x^*,u)). \label{eq:prox_eq}
\end{align}
\end{subequations}
In order to match \eqref{eq:prox_eq} to the firing rate equilibrium that satisfies \(Dx^* = \Phi(Wx^* + u)\), we choose
\begin{equation}\label{eq:pseudograd}
    \tilde\nabla J (x,u) = (D - W)x - u,
\end{equation}
and it follows that
\begin{equation}
    Wx^* + u = (D - \tilde\nabla J) (x^*).
\end{equation}
We then consider the following ODE dynamics, whose equilibria satisfy the fixed point equation~\eqref{eq:prox_eq}:
\begin{equation}\label{eq:pseudograd_dyn}
    \dot x = - Dx + \operatorname{prox}_\Gamma\left(Dx - \tilde\nabla J(x,u)\right)
\end{equation}
We term these dynamics the proximal gradient-play, abbreviated as PxGP dynamics (or PxGPD).
The dynamics in~\eqref{eq:pseudograd_dyn} have a similar form to the dynamics in~\eqref{eq: fr-prox_si}, and as we illustrate next, they enjoy a similar and meaningful energetic interpretation. The energetic interpretation of the proximal gradient-play reveals a balancing structure among neurons, where each neuron competes with the others to minimize its own objective. The individual objectives of the neurons, or energies, can have matching or opposing convexity, thus revealing regimes of global cooperation or competition.

\begin{remark}[\textbf{Vector valued energy}]
    Following an approach similar to the proof of Proposition~\ref{prop: Ener-uni}, we can explicitly compute the two terms $J_{i}(x_{i},x_{-i},u_{i})$ and $\upgamma_{i}$ involved in the single neuron dynamics. We know that
    \begin{equation}
        D_{ii}x_{i}-\nabla_{x_{i}}J_{i}(x_{i},x_{-i},u_{i})=\sum_{j=1}^{N}W_{ij}x_{j}+u_{i}
    \end{equation}
    from which it  follows that
    \begin{equation}
        J_{i}(x_{i},x_{-i},u_{i})=\frac{1}{2}(D_{ii}-W_{ii})x_{i}^{2}-x_{i} \Big(\sum_{j\neq i} W_{ij}x_{j} + u_{i}\Big)
    \end{equation}
    Notice that this is the minimal function that yields the desired partial derivative, but not the unique one. For any other function $K:\real^{N-1}\to\real$ evaluated as $K(x_{-i})$, we also have that
    \begin{equation}
        \nabla_{x_{i}}J(x_{i},x_{-i},u_{i})=\nabla_{x_{i}}[J(x_{i},x_{-i},u_{i})+K(x_{-i})]
    \end{equation}
    Thus, it is also be possible to use $K(x_{-i})$ to complete the square and gain deeper understanding on the role of each single interaction cost. The activation cost is, analogously to Proposition~\ref{prop: Ener-uni}
    \begin{equation}
        \upgamma_{i}(x_{i})=\int_{0}^{x_{i}}\upphi^{-1}_{i}(z) \ dz -\frac{1}{2}x_{i}^{2}
    \end{equation}
    Thus, defining the function $\bar{\Gamma}:\real^{N}\to\real^{N}$ such that $\bar{\Gamma}_{i}(x)=\upgamma_{i}(x_{i})$, we can write the vectorial energy
    \begin{equation}
        \bar{\En}(x,u)=-x\odot \Tilde{W}x-x\odot u+\bar{\Gamma}(x)
    \end{equation}
    where $\Tilde{W}_{ij}=W_{ij}$ for $i\neq j$ and $\Tilde{W}_{ii}=\frac{1}{2}(-D_{ii}+1+W_{ii})$.
\end{remark}

The above treatment is notably different from that in \citep{AG-AD-FB:24d} and \citep{SHM-MRJ:21pg} that discuss the notion of proximal gradient dynamics. 
The latter can be recovered by setting \(D = I\) and assuming that \(\tilde \nabla J\) is a true gradient arising from a scalar \(J\).
Note that \eqref{eq:pseudograd_dyn} is exactly identical to the FR dynamics when \eqref{eq:pseudograd} is plugged in. 
This formalism gives us a bridge to characterize the asymmetric FR network using a game with multiple cost functions rather than a single scalar cost.

Given the energy characterization of the neural cost functions, the neural network dynamics engage in the following proximal gradient play
\begin{equation}
    x^{\star}_{i}=\arg \min_{x_{i}} \En^{i}(x_{i},x_{-i}^{\star},u_{i})
\end{equation}
From the energy characterization, we observe that given an external input $u_{i}$, the neuron strives to minimize an individual cost $\En^{i}$ that jointly accounts for the individual cost of interaction with other neurons $J_{i}(x,u_{i})$ and the individual cost associated to the physical constraints of the system $\upgamma_{i}(x_{i})$. Thus, we have successfully extended the known energy characterization for symmetric FR networks to the general case of asymmetric FR networks.

\section{Stability of E-I networks}
\label{sec:EI}
In this section, we study general excitatory-inhibitory (E-I) networks and, in particular, the necessary and sufficient conditions for the uniqueness of their equilibria and their stability. We start with the useful definition of matrices of class $\mathcal{P}$.
\begin{definition}[\textbf{$\mathcal{P}$-matrix}]
    Let $A\in\real^{N\times N}$. We say $A\in\mathcal{P}$ or $A$ is a $\mathcal{P}$-matrix if all its principal minors are positive.
\end{definition}
We now introduce the concept of Lyapunov diagonal stability.
\begin{definition}[\textbf{Lyapunov diagonal stability}]\label{prp: LDS}
    Let $A\in\real^{N\times N}$ be a square matrix. $A$ is Lyapunov diagonally Stable ($A\in\mathcal{LDS}$) if there exists $P\in\real^{N\times N}$ diagonal, with $P\succ 0$, such that
    \begin{equation}
        AP+PA^{\top}\prec 0
    \end{equation}
\end{definition}
We now relate the previous two classes of matrices by means of the result presented in~\citep{HD:92}.
\begin{proposition}[\textbf{$\mathcal{LDS}\ \Rightarrow \mathcal{P}$}]\label{prop: LDS-to-P}
    Let $A\in\real^{N \times N}$ be $\mathcal{LDS}$. Then $-A\in\mathcal{P}$.
\end{proposition}
 
The sufficient condition ensuring $\mathcal{LDS}$ is more restrictive than the one obtained by simply proving that the matrix in the class $\mathcal{P}$. We now begin with a general characterization of E-I networks. We rely on the following crucial assumption.
\begin{assumption}[Excitatory-inhibitory reciprocity]\label{ass: balance}
    Let $(\mathcal{V},\mathcal{E})$ be the graph associated to the E-I network, with $\mathcal{V}=\{1,\dots,N\}$
    nodes and $\mathcal{E}$ edges. Assume that $\mathcal{V}$ contains at least one excitatory and one inhibitory node, and that we can divide the set of nodes in $\mathcal{G}_{E}=\{1,\dots,N_{E}\}$ set of excitatory neurons and $\mathcal{G}_{I}=\{N_E+1,\dots,N\}$ set of inhibitory neurons, with $N_E+N_I=N$. Then $\forall i \in\mathcal{G}_E$ and $\forall j\in\mathcal{G}_I$
    \begin{equation}
        (i,j)\in\mathcal{E}\quad\iff\quad (j,i)\in\mathcal{E}
    \end{equation}
\end{assumption}
Namely, our treatment will rely on the assumption of a reciprocal E-I network, where every excitatory neuron that excites an inhibitory neuron is inhibited back, and vice-versa.

\subsection{Dynamics and equilibria}
We first define the firing rate dynamics that form the basis of E-I networks.
\begin{definition}[\textbf{Firing-rate Network}]\label{def:FR}
    Let $0 \prec D\in\real^{N\times N}$ be the dissipation matrix and $W\in\real^{N\times N}$ be the synaptic matrix. 
    The dynamics of the firing-rate network are defined as in~\eqref{eq: fr-prx}
    \begin{equation}\label{eq:FR}
        \dot x = -Dx+\Phi(Wx+u)
    \end{equation}
    where $u\in\real^N$ is an external input and $\Phi:\real^{N}\to\real^{N}$ is the activation function.
\end{definition}

\subsubsection{Dynamics of E-I network}

We now define the dynamics of a generic E-I network satisfying the excitatory-inhibitory reciprocity in Assumption~\ref{ass: balance}. Let us define the components of the dynamics. 
\begin{itemize}
    \item Dissipation matrix $D\in\real^{N \times N}$
    \begin{align}\label{eq: dis}
        D&=\diag(\underbrace{d_E,\dots,d_E}_{N_{E}},\underbrace{d_I,\dots,d_I}_{N_{I}})\nonumber\\
        &=
        \begin{pmatrix}
            D_E & \vectorzeros[N_{E}\times N_I]\\
            \vectorzeros[N_{I}\times N_E] & D_I
        \end{pmatrix}
    \end{align}
    \item Synaptic matrix $W\in\real^{N \times N}$
    \begin{equation}\label{eq: syn}
        W=
        \begin{pmatrix}
            W_E^E & W_I^E\\
            W_E^I & W_I^I
        \end{pmatrix}
    \end{equation}
    where $W_X^Y$ contains the edges going from $\mathcal{G}_X$ to $\mathcal{G}_Y$, for $X,Y\in\{E,I\}$.
\end{itemize}
\begin{figure}[!ht]
    \centering
    \includegraphics[width=0.5\linewidth]{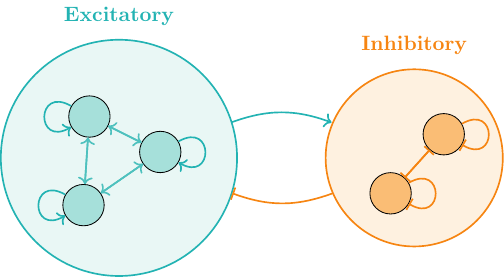}
    \caption{Schematic representation of an excitatory-inhibitory network. The network is comprised of $N_{E}$ excitatory neurons and $N_{I}$ inhibitory neurons with arbitrary same-type connectivity and directed different-type connectivity.}
    \label{fig:NetworkEI}
\end{figure}

We now make a common simplifying assumption in the study of excitatory-inhibitory networks.
\begin{assumption}[\textbf{Synaptic homogeneity}]\label{ass: hom}
    The synaptic weights in the E-I network are homogeneous and depend only on the sending and receiving population of neurons. That is, the synaptic weights for each type of connection are identical,
    \begin{subequations}
        \begin{align}
            W_{ij}&=w_{EE}>0\qquad (i,j)\in \mathcal{E}_{EE}\quad \text{ excitatory-to-excitatory}\\
            W_{ij}&=-w_{EI}<0\qquad (i,j)\in \mathcal{E}_{EI}\quad \text{inhibitory-to-excitatory}\\
            W_{ij}&=w_{IE}>0\qquad (i,j)\in \mathcal{E}_{IE}\quad \text{ \: excitatory-to-inhibitory}\\
            W_{ij}&=-w_{II}<0\qquad (i,j)\in \mathcal{E}_{II}\quad \text{ inhibitory-to-inhibitory}
        \end{align}
    \end{subequations}
    where $\mathcal{E}_{XY}$, for $X,Y\in\{E,I\}$, denotes the subset of edges from population $X$ to population $Y$.
\end{assumption}
\begin{remark}[\textbf{Random connectivity}]
    Notice that from the previous assumption we are not requiring the graph $(\mathcal{V},\mathcal{E})$ to be completely connected, but we are allowing for sparsity. Indeed, for any indices $i\in\mathcal{G}_{X}$ and $j\in\mathcal{G}_{Y}$, for $X,Y\in\{E,I\}$, such that $(i,j)\notin \mathcal{E}$ it follows that $W_{ij}=0$.
\end{remark}
The dynamics of an excitatory-inhibitory network can thus be defined.
\begin{definition}[\textbf{E-I Network}]\label{def: EIdyn}
    For the firing rate dynamics defined in \eqref{eq:FR}, let the dissipation matrix \(D\in\real^{N\times N}\) have the structure described in~\eqref{eq: dis} and the synaptic matrix $W\in\real^{N\times N}$ be as in~\eqref{eq: syn}. 
    The dynamics of the E-I network is then given by
    \begin{equation}\label{eq: EIdyn}
        \dot x = -Dx+\Phi(Wx+b+u)
    \end{equation}
    where $b\in\real^N$ is called the internal bias and $u\in\real^N$ is the external input.
\end{definition}

The choice of the activation function is allowed to be fairly generic, but we will restrict ourselves to a specific class. 
Results henceforth presented will consider E-I networks with activation function satisfying the following assumption.
\begin{assumption}[\textbf{Slope-restriction and diagonality}]\label{ass: sl-d}
    The activation function $\Phi$ is continuous. In addition, the activation function $\Phi$ is
    \begin{enumerate}[label=A{\arabic*})]
        \item\label{it: diaghom} diagonal and homogeneous: $\Phi_{i}(x)=\upphi(x_{i})$;
        \item\label{it: posbou} positive and bounded: $\upphi: \real \to [0, M]$, for $M>0$;
        \item\label{it: weak} weakly increasing: for all scalars $y\geq z$, we have $\upphi(y)\geq\upphi(z)$;
        \item\label{it: nonexp} non-expansive: for all scalars $y\geq z$, we have $\upphi(y)-\upphi(z)\leq y-z$.
    \end{enumerate}
\end{assumption}

\begin{remark}\label{rem: Lip}
    Notice that by using the facts that $\upphi$ is continuous, ($A3$) weakly increasing, and ($A4$) non-expansive, we obtain
    \begin{equation}
        \frac{|\upphi(x)-\upphi(z)|}{|x-z|}\leq 1\qquad x,z\in\real,\; x\neq z
    \end{equation}
    Therefore, we additionally have that $\upphi$ is \(1\)-Lipschitz.
\end{remark}
We extend a result proved under Lemma 3 of the Appendix of~\citep{SB-GB-FB-SZ:24m} to the case of \(D \neq I\).
\begin{lemma}\label{lem:invariant}
    Let Assumption~\ref{ass: sl-d} hold. Then the set $\mathcal{X}:=\{x\in\real^{N}:\ Dx\in[0,M]^{N}\}$ is forward-invariant for the firing rate dynamics~\eqref{eq: EIdyn}.
\end{lemma}
\begin{proof}
    Let $y=Dx$, so that $\dot{y}=D\dot{x}$. The $y$-dynamics are
    \begin{equation}
        \dot{y}=-Dy+D\Phi(WD^{-1}y+b+u).
    \end{equation}
    From Assumptions~\ref{it: diaghom} and~\ref{it: posbou} we have that
    \begin{align}
        \dot{y}_i&=D_{ii}\bigl(-0+\upphi([WD^{-1}y]_{i}+b_i+u_{i})\bigr)\geq 0\qquad \forall y\in[0,M]^{N},\ y_{i}=0\nonumber\\
        \dot{y}_i&=D_{ii}\bigl(-M+\upphi([WD^{-1}y]_{i}+b_i+u_{i})\bigr)\leq 0\qquad \forall y\in[0,M]^{N},\ y_{i}=M
    \end{align}
    since $D_{ii}>0$. Thus $[0,M]^N$ is forward-invariant for the $y$-dynamics, and taking $x=D^{-1}y$ we conclude.
\end{proof}
\begin{remark}[Absorbing invariance]
    Notice that for all $y\in\real^{N}$ such that $y_i<0$ then we have that $f_{FR}(y)_i>0$. Conversely, for all $y\in\real^{N}$ such that $y_{i}>M$ we have $f_{FR}(y)_i<0$. Therefore, $\mathcal{X}$ is also absorbing for the dynamics~\eqref{eq: EIdyn}, in the sense that for all $x_{0}\in\real^{N}$ we have $\lim_{t\to +\infty} x^{x_0}(t)\in \mathcal{X}$, where $x^{x_0}()$ denotes the flow of the solution of~\eqref{eq: EIdyn} for $x_0\in\real^N$ initial condition. 
\end{remark}

\subsubsection{Uniqueness of equilibria}
First, we arrive at a necessary condition for unique equilibrium of E-I networks. 

\begin{proposition}[\textbf{Necessary $\mathcal{P}$-condition}]\label{prop: EIuniq}
    $D-W\in\mathcal{P}$ only if $w_{EE}<d_{E}$.
\end{proposition}
\begin{proof}
    Suppose that $w_{EE}\geq d_{E}$. Among the principal minors of order 1 of the matrix $D-W$ we have 
    \begin{equation}
        d_{E}-w_{EE}\leq 0
    \end{equation}
    and this would contradict the definition of matrix in class $\mathcal{P}$.
\end{proof}

We now extend results on the homeomorphism of the Hopfield network dynamics in~\citep[Theorem 3]{FM-TA:95} to the firing rate model.

\begin{proposition}[\textbf{Bijectivity of the firing rate map}]\label{prop: surjFR}
    Let the firing rate map be $x\mapsto F(x)=-Dx+\Phi(Wx+b+u)$, for fixed bias $b\in\real^{N}$ and external input $u\in\mathcal{U}$. Assume $D-W\in\mathcal{P}$ and the activation function $\Phi$ satisfies assumption~\ref{ass: sl-d} (in particular, it is continuous,~\ref{it: diaghom} diagonal,~\ref{it: posbou} positive and bounded,~\ref{it: weak} weakly increasing, and~\ref{it: nonexp} slope-restricted). Then, $F:\real^{N}\to\real^{N}$ is a homeomorphism.
\end{proposition}
\begin{proof}
    From Assumptions~\ref{it: diaghom}-\ref{it: nonexp}, the generalized derivative of $\upphi$ at $z\in\real$ is the convex hull of the set of limits
    \begin{equation}
        \partial \upphi(z):=\lim_{k\to +\infty} \frac{\partial \upphi}{\partial z}(z+h_{k})
    \end{equation}
    where $\upphi$ is differentiable at each point $z+h_{k}\in\real$ of the subsequence $h_{k}\xrightarrow[]{k\to +\infty} 0$. In particular, from Remark~\ref{rem: Lip} and Assumption~\ref{it: weak} for all $z\in\real$ we have that $v\in[0,1]$ for all $v\in\partial\upphi(z)$. Exploiting now Assumption~\ref{it: diaghom} we have that the generalized Jacobian $\partial\Phi$ is a set of diagonal matrices, and in particular all the elements $\Sigma(x)\in\partial\Phi(Wx+b+u)$ of the generalized Jacobian satisfy
    \begin{equation}\label{eq: gen-jac}
        \vectorzeros[N\times N]\preceq \Sigma(x)\preceq \mathcal{I}_{N}\qquad \forall x\in\real^{N}.
    \end{equation}
    Consider the generalized Jacobian of the firing rate map evaluated at $x\in\real^{N}$
    \begin{equation}\label{eq: jacF}
        \partial F(x):=\{-D+\Sigma(x)W:\ \Sigma(x)\in\partial\Phi(Wx+b+u)\}.
    \end{equation}
    Since $D-W\in\mathcal{P}$, we know that $\det(-D+W)\neq0$. Using now~\citep[Lemma 2]{FM-TA:95}, we have that for any diagonal matrix $K\in\real^{N\times N}$ such that $\vectorzeros[N\times N]\preceq K \preceq \mathcal{I}_{N}$, $\det(-D+KW)\neq 0$. Exploiting now the bound in~\eqref{eq: gen-jac}, we get for all $x\in\real^N$
    \begin{equation}
        \det(-D+\Sigma(x)W)\neq0\qquad \forall \Sigma(x)\in\partial\Phi(Wx+b+u).
    \end{equation}
    Therefore, at all points of differentiability the Jacobian of the map $F$ is full rank, and so is the generalized Jacobian on the zero measure set of points of non-differentiability for $F$. Consequently, the map $F$ is locally invertible everywhere, by \citep[Theorem 15]{CFH:83} on the local invertibility of locally Lipschitz functions.

    Finally, observe that the map $x\mapsto F(x)$ is radially unbounded, i.e., for any sequence $\{x_{k}\}_{k\in\mathbb{N}}$ such that $\|x_{k}\|\xrightarrow[]{k\to +\infty}+\infty$ 
    \begin{align}
        \lim_{k\to +\infty} \|F(x_k)\| &= \lim_{k\to +\infty} \|-Dx_k+\Phi(Wx_{k}+b+u)\|\nonumber\\
        &\geq \lim_{k\to +\infty} |\|\Phi(Wx_{k}+b+u)\|-\|Dx_{k}\||=+\infty.
    \end{align}
    since $D$ is full rank and $\Phi$ is bounded. Therefore, applying~\citep[Theorem 1]{FM-TA:95} on homeomorphisms induced by $C^0$ functions, the map $x\mapsto F(x)$ induces a homeomorphism from $\real^{N}$ to $\real^{N}$.
\end{proof}

An immediate consequence of the this result is that there exists a unique $x^{\star}\in\real^{N}$ such that $F(x^{\star})=\vectorzeros[N]$, and therefore a unique equilibrium point for the E-I firing rate dynamics.

\subsection{Stability of the equilibrium points}
\subsubsection{Lyapunov diagonal stability}
To characterize the stability of the unique equilibrium point, we find it convenient to study the condition known as Lyapunov diagonal stability ($\mathcal{LDS}$) of the system matrix $-D+W$. The choice is inspired by the approach followed by Forti \& Tesi~\citep{FM-TA:95}.

\begin{theorem}[\textbf{General E-I network - $\mathcal{LDS}$}]\label{thm: LDS}
    Given the E-I network dynamics in~\eqref{eq: EIdyn}, suppose that the excitatory-inhibitory balance assumption~\ref{ass: balance} and the synaptic homogeneity assumption~\ref{ass: hom} hold. Let $0\leq k^{E}_{in}\leq N_E$ be the maximal in-degree of $W_E^E$ and $0\leq k^{E}_{out}\leq N_E$ be the maximal out-degree of $W_E^E$. Analogously, let $0\leq k^{I}_{in}\leq N_I$ be the maximal in-degree of $W_I^I$ and $0\leq k^{I}_{out}\leq N_I$ be the maximal out-degree of $W_I^I$. If 
    \begin{subequations}
        \begin{align}
        d_{I}&>w_{II}[-4+(k^{I}_{in}+k^{I}_{out})]/2 \label{eq:LDS-II}
        \\
        d_E&>w_{EE}(k^{E}_{in}+k^{E}_{out})/2  \label{eq:LDS-EE}
    \end{align}   
    \end{subequations}    
    then $W-D\in\mathcal{LDS}$. 
\end{theorem}

\begin{proof}
    Define a positive definite diagonal matrix $P\in\real^{N\times N}$ as follows.
    \begin{align}
        P&=\diag(\underbrace{p_E,\dots,p_E}_{N_E},\underbrace{p_I,\dots,p_I}_{N_I})
        \notag\\
        &=
        \begin{pmatrix}
            P_E & \vectorzeros[N_{E}\times N_I]\\
            \vectorzeros[N_{I}\times N_E] & P_I
        \end{pmatrix} \succ 0
    \end{align}
    where $p_E, p_I>0$. Following the approach by 
    Forti~\&~Tesi~\citep{FM-TA:95}, we impose the $\mathcal{LDS}$ condition by requiring $(D-W)^{\top}P+P(D-W)\succ 0$, where we have multiplied \(-1\) on both sides. The matrix in this linear-matrix inequality (LMI) is symmetric and results in the following coefficients for the different blocks:
    \begin{itemize}
        \item Excitatory-to-excitatory block $(D_E-W_E^E)^{\top}P_E+P_E(D_E-W_E^E)$:
        \begin{subequations}            
            \begin{align}
                \{(D_E-W_E^E)^{\top}P_E+P_E(D_E-W_E^E)\}_{ii}&=2p_{E}(d_{E}-w_{EE})\\
                \{(D_E-W_E^E)^{\top}P_E+P_E(D_E-W_E^E)\}_{ij}&=-k_{ij}p_E w_{EE}
            \end{align}
    \end{subequations}
            where $k_{ij}=2$ if both $\{W_E^E\}_{ij}$ and $\{W_E^E\}_{ji}$ are non-zero, $k_{ij}=1$ if either one of the two is zero and the other one non-zero, and $k_{ij}=0$ if they are both zero.
        \item Inhibitory-to-inhibitory $(D_I-W_I^I)^{\top}P_I+P_I(D_I-W_I^I)$:
        \begin{subequations} 
            \begin{align}
                \{(D_I-W_I^I)^{\top}P_I+P_I(D_I-W_I^I)\}_{ii}&=2p_{I}(d_{I}+w_{II})\\
                \{(D_I-W_I^I)^{\top}P_I+P_I(D_I-W_I^I)\}_{ij}&=k_{ij}p_I w_{II}
            \end{align}
            \end{subequations} 
            where $k_{ij}=2$ if both $\{W_I^I\}_{ij}$ and $\{W_I^I\}_{ji}$ are non-zero, $k_{ij}=1$ if either one of the two is zero and the other one non-zero, and $k_{ij}=0$ if they are both zero.
        \item Excitatory-to-inhibitory $-(W_E^I)^{\top}P_I-P_E(W_I^E)$:\\

        Whenever the elements of the blocks are non-zero, they are always of the form
        \begin{equation}
           \{-(W_E^I)^{\top}P_I-P_E(W_I^E)\}_{ij}=p_E w_{EI} - p_I w_{IE}
        \end{equation}
    \end{itemize}
   It is immediate to observe that by choosing $p_E=w_{EI}^{-1}$ and $p_I = w_{IE}^{-1}$ the extra-diagonal blocks are zero, and therefore we are left with evaluating the positive definiteness of a block-diagonal matrix. Precisely, by Gershgorin circle theorem~\citep{HRA-JCR:85},

   we have that the inhibitory-to-inhibitory block $(D_I-W_I^I)^{\top}P_I+P_I(D_I-W_I^I)$ has eigenvalues in the interval with center $2p_{I}(d_{I}+w_{II})$ and radius $p_{I}w_{II}[-2+(k^{I}_{in}+k^{I}_{out})]$.
   Thus, the condition for all the eigenvalues to be positive is
    \begin{equation}
        d_{I}+w_{II}>w_{II}[-2+(k^{I}_{in}+k^{I}_{out})]/2
    \end{equation}

    Considering then the excitatory-to-excitatory block $(D_E-W_E^E)^{\top}P_E+P_E(D_E-W_E^E)$, by using our hypothesis on the maximal in-degree and out-degree on the graph and exploiting Gershgorin circle theorem again, the maximal interval containing the eigenvalues of the matrix  has center $2p_{E}(d_{E}-w_{EE})$ and radius $p_{E}w_{EE}[-2+(k^{E}_{in}+k^{E}_{out})]$.    

   Thus, the condition for all the eigenvalues to be positive is
    \begin{equation}
        d_E-w_{EE}(k^{E}_{in}+k^{E}_{out})/2>0
    \end{equation}
\end{proof}

Notice in particular that, from Proposition~\ref{prop: LDS-to-P} we have that $W-D\in\mathcal{LDS}$ implies $D-W\in\mathcal{P}$, and using Proposition~\ref{prop: surjFR} we have that the dynamics in~\eqref{eq: EIdyn} have a unique equilibrium point.

\subsubsection{Global stability}

We aim to study the global stability of a generic E-I network. Global stability can be either exponential or just asymptotic, with the former being a stronger notion than the latter.  The standard way of proving global asymptotic stability (GAS) is through Lyapunov theory. Specifically, proving GAS requires the identification of a suitable Lyapunov function that has negative total time derivative along the trajectories of the dynamical system. Unfortunately, finding suitable Lyapunov functions is more of an art than a science, and most of the results for recurrent neural networks focus on the Hopfield model.  The recent Lyapunov function proposed in~\citep{SB-GB-FB-SZ:24m} works for the symmetric firing rate model, but the approach fails for E-I networks with their asymmetric synaptic matrix~$W$.  Next, we show how to establish convergence to equilibria for a specific class of asymmetric firing-rate networks, which include networks where the dissipation matrix is a scalar matrix, i.e., $D=d\mathcal{I}_{N}$ for some $d>0$.

\begin{theorem}[\textbf{Global convergence for $\mathcal{LDS}$ networks}]\label{thm: EIglob}
    Consider the firing rate model~\eqref{eq:FR} with dissipation matrix $D$, synaptic matrix $W$ and Lipschitz, diagonal activation function $\Phi$. If
    \begin{itemize}
        \item $W-D\in\mathcal{LDS}$ and $D\succ 0$.
        \item $W$ and $D$ commute, i.e., $WD=DW$.
        \item each $\upphi_{i}$ is weakly increasing~\ref{it: weak} and non-expansive~\ref{it: nonexp}.

    \end{itemize}
    Then, for each constant $u\in\real^{N}$, there is a unique equilibrium point $x^{\star}$ that is globally asymptotically stable.
\end{theorem}
\begin{proof}
    Since $W-D\in\mathcal{LDS}$, then $D-W\in\mathcal{P}$ and the map $x\mapsto -Dx+\Phi(Wx+u)$ is bijective from $\real^{N}$ to itself (see Proposition~\ref{prop: surjFR}). Therefore, there exists a unique $x^{\star}\in\real^{N}$ such that $-Dx^{\star}+\Phi(Wx^{\star}+u)=\vectorzeros[N]$, establishing existence and uniqueness of the equilibrium.

    Given the pair $(x^{\star},u)$, perform the change of variables $y=x-x^{\star}$, so that the firing rate dynamics can be rewritten as
    \begin{align}
        \dot y &= \dot x\nonumber\\
        &= -D(y+x^{\star})+\Phi(W(y+x^{\star})+u)\nonumber\\
        &= -Dy+\Phi(W(y+x^{\star})+u)-\Phi(Wx^{\star}+u)=-Dy+\Tilde\Phi(Wy)
    \end{align}
    where from the second to the third line we have used the equivalence $Dx^{\star}=\Phi(Wx^{\star}+u)$ and where $\Tilde\Phi(z)=\Phi(z+Wx^{\star}+u)-\Phi(Wx^{\star}+u)$. The new function $\Tilde\Phi$ is diagonal, with each entry defined as $\Tilde\Phi_{i}(z)=\Tilde\upphi(z_{i})=\upphi(z_{i}+(Wx^{\star}+u)_{i})-\upphi((Wx^{\star}+u)_{i})$. In particular, notice that the assumptions on $\upphi$ imply that $\Tilde\upphi$ is Lipschitz, weakly increasing~\ref{it: weak}, and non-expansive~\ref{it: nonexp}.

    In addition, it satisfies $\Tilde\upphi(0)=0$.

    Consider now the Lyapunov function
    \begin{align}
        V(y) &=\frac{1}{2}\|y\|^{2}+V_{W}^{\upphi}(y) 
        \notag
        \\
        &=\frac{1}{2}\|y\|^{2}+\sum_{i=1}^{N}P_{ii}\int_{0}^{[Wy]_{i}}\Tilde\upphi(z)\ dz
    \end{align}

    where $P_{ii}>0$ are the diagonal coefficient of a matrix $P\succ 0$. Then computing the total time derivative of $V_{W}^{\upphi}$
    along the trajectories of the \textit{shifted} firing rate system we have
    \begin{align}\label{eq: ineqV1}
        \frac{d}{dt}V_{W}^{\upphi}(y)_{|y=y(t)}&=\sum_{i=1}^{N}P_{ii}\Tilde\upphi([Wy]_{i})[W\dot y]_{i}\nonumber\\
        &=\Tilde\Phi(Wy)^{\top}PW\dot y\nonumber\\
        &=\Tilde\Phi(Wy)^{\top}PW\left[-Dy+\Tilde\Phi(Wy)\right]\nonumber\\
        &=-\Tilde\Phi(Wy)^{\top}PWDy+\Tilde\Phi(Wy)^{\top}PW\Tilde\Phi(Wy)\nonumber\\
        &\stackrel{\mathmakebox[\widthof{=}]{\scriptstyle \text{comm.}}}{=}-\Tilde\Phi(Wy)^{\top}PDWy+\Tilde\Phi(Wy)^{\top}PW\Tilde\Phi(Wy)\nonumber\\
        &\stackrel{\mathmakebox[\widthof{=}]{\scriptstyle\text{non-exp.}}}{\leq}-\Tilde\Phi(Wy)^{\top}PD\Tilde\Phi(Wy)+\Tilde\Phi(Wy)^{\top}PW\Tilde\Phi(Wy)\nonumber\\
        &=\frac{1}{2}\Tilde\Phi(Wy)^{\top}\left[P(-D+W)+(-D+W)^{\top}P\right]\Tilde\Phi(Wy)
    \end{align}

    The non-expansiveness property has been used by noticing that for each of the components of $\Tilde\Phi$, we have $\Tilde\upphi(z_{i})\leq z_{i}$ for $z_{i}>0$, and $\Tilde\upphi(z_{i})\geq z_{i}$ for $z_{i}<0$. Exploiting the diagonality of $PD$, we  obtain that $P_{ii}D_{ii}\Tilde\upphi([Wy]_{i})[Wy]_{i}\geq P_{ii}D_{ii}\Tilde\upphi([Wy]_{i})^{2}$ for all $i=1,\dots,N$ and all $y\in\real^{N}$.
    We now focus on the quadratic term - the contribution of which makes $V(y)$ radially unbounded. We start by defining $d_{\min}=\min_{i=1,\dots,N}\ d_{i}/2$ and observe that $D-d_{\min}\mathcal{I}_{N}\succ 0$. We set $D_{0}=(D-d_{\min}\mathcal{I}_{N})^{1/2}$. Then
    \begin{align}
        \frac{d}{dt}\frac{1}{2}{\|y\|^{2}_{2}}_{|y=y(t)}&=y^{\top}\dot y
        \notag
        \\
        &=-y^{\top}Dy+y^{\top}\Tilde{\Phi}(Wy)
        \notag
        \\
        &=-d_{\min}\|y\|_{2}^{2}-y^{\top}D_{0}^{2}y+y^{\top}\Tilde\Phi(Wy).
    \end{align}
    Consider now $y^{\top}\Tilde{\Phi}(Wy)$ as the cross term in the expansion of the square of the sum of two terms: specifically, from
    \begin{equation}
        \|D_{0}y-\frac{1}{2}D_{0}^{-1}\Tilde\Phi(Wy)\|^{2}_{2}=y^{\top}D_{0}^{2}y-y^{\top}\Tilde\Phi(Wy)+\frac{1}{4}\Tilde\Phi(Wy)^{\top}D_{0}^{-2}\Tilde\Phi(Wy)
    \end{equation}
    we obtain the following bound
    \begin{align}\label{eq: ineqV2}
        \frac{d}{dt}\frac{1}{2}{\|y(t)\|_{2}^{2}}_{|y=y(t)}&=-d_{\min}\|y\|^{2}_{2}-\|D_{0}y-\frac{1}{2}D_{0}^{-1}\Tilde\Phi(Wy)\|^{2}_{2}+\frac{1}{4}\Tilde\Phi(Wy)^{\top}D_{0}^{-2}\Tilde\Phi(Wy)\nonumber\\
        &\leq -d_{\min}\|y\|_{2}^{2}+\frac{1}{4}\Tilde\Phi(Wy)^{\top}D_{0}^{-2}\Tilde\Phi(Wy)
    \end{align}

    Combining now the inequalities~\eqref{eq: ineqV1} and~\eqref{eq: ineqV2} we obtain
    \begin{align}
        \frac{d}{dt}V(y)_{|y=y(t)}\leq -d_{\min}\|y\|_{2}^{2}+\frac{1}{4}\Tilde\Phi(Wy)^{\top}D_{0}^{-2}\Tilde\Phi(Wy)+\frac{1}{2}\Tilde\Phi(Wy)^{\top}\left[P(-D+W)+(-D+W)^{\top}P\right]\Tilde\Phi(Wy)
    \end{align}
    and since $W-D\in\mathcal{LDS}$, we can find $P\succ 0$ such that 
    \begin{equation}\label{eq: ineqM}
        \left[P(-D+W)+(-D+W)^{\top}P\right]\prec -\frac{1}{2}D_{0}^{-2}.
    \end{equation}

    Finally, using the inequality~\eqref{eq: ineqM} we get
    \begin{equation}
        \frac{d}{dt}V(y)_{y=y(t)}\leq -d_{\min}\|y\|_{2}^{2}
    \end{equation}
    Therefore, the point $y=\vectorzeros[N]$ is GAS and, in the original coordinates, $x^{\star}$ is GAS.
\end{proof}

The same Lyapunov diagonal stability condition for the system matrix $-D+W$ is used by Forti \& Tesi~\citep{FM-TA:95} to ensure the existence of a Lyapunov function for \emph{asymmetric} systems of the Hopfield type. In that work, the authors rely on a similar Lur'e-Postnikov Lyapunov function, and are able to prove global asymptotic stability for general dissipation matrices $D$, not restricting to the class that commutes with $W$. The same could not be replicated in our case due to the different position of the non-linearity $\Phi$ in the dynamics, which makes the general treatment more complicated. For these reasons, we propose a conjecture and subsequent numerical validation on the global asymptotic stability of general \emph{asymmetric} firing-rate networks with $W-D\in\mathcal{LDS}$, thereby dropping the commutativity requirement $DW=WD$.

\begin{conjecture}[\(\mathcal{LDS}\) implies global asymptotic stability of the firing rate]\label{conj: GAS}
    Consider the firing-rate network in \eqref{eq:FR} satisfying Assumption~\ref{ass: sl-d} and let the dissipation matrix \(D\) and synaptic matrix \(W\) satisfy \(W-D \in \mathcal{LDS}\). Then, for each input \(u\), the unique equilibrium point of the network is globally asymptotically stable on \(\mathcal X\).
\end{conjecture}

The numerical evidence for the above conjecture is presented in the form of probability estimation via Monte Carlo simulation. To do this, we need to sample from the set of matrix pairs \((D, W)\) such that \(W-D \in \mathcal{LDS}\). The next result provides useful information about this set and how to sample it uniformly.

\begin{proposition}[Topological properties of sets of matrix pairs yielding \(\mathcal{LDS}\)]\label{prop: LDS topology}
    Consider the following:
    \begin{enumerate}[label=\roman*)]
        \item 
        The set of matrix pairs \((D, W)\) given by
        \begin{equation}
            \mathscr P = \{(D, W) \in \operatorname{diag}\left(\mathbb R^{n}_{>0}\right)\times \mathbb R^{n\times n} : W-D \in \mathcal{LDS}\}
        \end{equation}
        is nonconvex (for \(n\ge2\)), not closed on \(\operatorname{diag}\left(\mathbb R^{n}_{>0}\right)\times \mathbb R^{n\times n}\), and a blunt cone.
        \item
        For any given diagonal matrices \(D\succ0\) and \(\Lambda\succ 0\) and margin \(\upgamma>0\) such that the set 
        \begin{equation}
            \mathscr P_{\upgamma, \Lambda, D} = \{W \in \mathbb R^{n\times n} : (W-D)^\top \Lambda + \Lambda(W-D) \preceq -\upgamma I\}
        \end{equation}
        is nonempty, \(\mathscr P_{\upgamma, \Lambda, D}\) is closed, unbounded, and convex.

        \item
        For any given diagonal matrices \(D\succ0\) and \(\Lambda\succ 0\), margin \(\upgamma>0\), and \(R>0\), the set 
        \begin{equation}
            \mathscr P_{\upgamma, \Lambda, D} \cap \{W \in \mathbb R^{n\times n} : \lVert W \rVert_F \leq R \}
        \end{equation}
        is convex and compact.
    \end{enumerate}

\end{proposition}
\begin{proof}
    \begin{enumerate}[label=\roman*)]
        \item \textit{Nonconvexity} (for \(n\ge2\)). We work in dimension \(n=2\); the general case follows by taking direct sums with stable diagonal blocks \(-I_{n-2}\).
        Take \(D_1=D_2=\mathcal I_2\) and
        \[
        A_1=\begin{pmatrix}-1 & a\\ 0 & -2\end{pmatrix},\qquad
        A_2=\begin{pmatrix}-1 & 0\\ a & -2\end{pmatrix}.
        \]
        Both \(A_1\) and \(A_2\) are triangular with eigenvalues \(-1\) and \(-2\), hence Hurwitz. An explicit diagonal certificate for \(A_1\) is \(\Lambda_1=\operatorname{diag}(1,\,a^2+1)\), for which
        \[A_1^\top\Lambda_1+\Lambda_1 A_1 = \begin{pmatrix}-2 & a\\ a & -4(a^2+1)\end{pmatrix}.\]
        This is negative definite since the \((1,1)\)-entry is \(-2<0\) and the Schur complement is \(-4(a^2+1)+\tfrac{a^2}{2}=-\tfrac{7a^2}{2}-4<0\) for all \(a\). By the same construction with \(\Lambda_2=\operatorname{diag}(a^2+1,\,1)\), we get \(A_2\in\mathcal{LDS}\) for any \(a\in\mathbb R\).

        Now consider the midpoint \(\bar A = \tfrac{1}{2}(A_1+A_2) = \begin{pmatrix}-1 & a/2\\ a/2 & -2\end{pmatrix}\).
        Its determinant is
        \[
        \det(\bar A) = 2 - \frac{a^2}{4}.
        \]
        For \(|a|>2\sqrt{2}\) we have \(\det(\bar A)<0\), so \(\bar A\) has a positive eigenvalue and cannot be Hurwitz, hence \(\bar A\notin\mathcal{LDS}\).
        Therefore \((\mathcal I_2,\,\mathcal I_2+\bar A)\notin\mathscr P\) even though \((D_1,W_1),(D_2,W_2)\in\mathscr P\), contradicting convexity.
        Thus \(\mathscr P\) must be nonconvex for \(n\ge2\).
        
        To show that \(\mathscr P\) is not closed, consider the sequence within the domain: let \(D = \mathcal{I}_n\) and \(W_k = (1-1/k)\mathcal{I}_n\) for \(k \in \mathbb N\).
        Then \(W_k - D = -(1/k)\mathcal{I}_n \in \mathcal{LDS}\) for all \(k\) (with Lyapunov certificate \(\Lambda = \mathcal{I}_n\)), so \((D, W_k) \in \mathscr P\) for all \(k\).
        However, \(W_k \to \mathcal{I}_n\) as \(k\to\infty\), and the limit point \((D, \mathcal{I}_n)\) satisfies \(W - D = 0 \notin \mathcal{LDS}\), so \((\mathcal{I}_n, \mathcal{I}_n) \notin \mathscr P\).
        Thus \(\mathscr P\) is not closed in \(\operatorname{diag}\left(\mathbb R^{n}_{>0}\right)\times \mathbb R^{n\times n}\).

        \(\mathscr P\) is a cone due to the fact that for \(\upalpha >0\) and any \((D, W) \in \mathscr P\), say with diagonal Lyapunov solution \(\Lambda\), \((\upalpha D, \upalpha W) \in \mathscr P\) with diagonal Lyapunov solution \(\frac{1}{\upalpha}\Lambda\), which can be verified by the Lyapunov inequality.

        It is a blunt cone due to \(\upalpha = 0\) not preserving the \(\mathcal{LDS}\) property.
        
        \item \textit{Convexity.} Given $W_1$, $W_2 \in \mathscr P_{\upgamma, \Lambda, D}$ and \(\uplambda\in[0,1]\), linearity of the Lyapunov map gives
        \[(-D + \uplambda W_1 + (1-\uplambda)W_2)^\top\Lambda + \Lambda(-D + \uplambda W_1 + (1-\uplambda)W_2) \preceq -\upgamma I,\]
        so \(\mathscr P_{\upgamma, \Lambda, D}\) is convex.

        \textit{Unboundedness.} Let \(W_0\in\mathscr P_{\upgamma,\Lambda,D}\) be any feasible point.
        For \(n\ge2\), pick any nonzero skew-symmetric matrix \(S\) (\(S^\top=-S\)) and define \(W_t = W_0 + t\Lambda^{-1}S\) for \(t\in\mathbb R\).
        Then
        \[\Lambda(W_t-D)+(W_t-D)^\top\Lambda = \Lambda(W_0-D)+(W_0-D)^\top\Lambda + tS + tS^\top = \Lambda(W_0-D)+(W_0-D)^\top\Lambda\preceq -\upgamma I,\]
        so every \(W_t\in\mathscr P_{\upgamma,\Lambda,D}\), while \(\|W_t\|_F\to\infty\) as \(|t|\to\infty\).
        For \(n=1\), if \(w_0\) is feasible then \(w_0-t\) is feasible for all \(t\ge0\), which is also unbounded.

        \textit{Closedness.} Consider the map
        \(\Psi_{\upgamma, \Lambda, D}: \mathbb R^{n\times n} \rightarrow \mathbb S^{n}\) given by
        \[\Psi_{\upgamma, \Lambda, D}(W) = (W-D)^\top\Lambda + \Lambda(W-D) + \upgamma I.\]
        \(\Psi_{\upgamma, \Lambda, D}\) is continuous and \(\mathscr P_{\upgamma, \Lambda, D} = \Psi_{\upgamma, \Lambda, D}^{-1}(\mathbb S^n_{\preceq 0})\). Since \(\mathbb S^n_{\preceq 0}\) is closed and the preimage of a closed set under a continuous map is closed, \(\mathscr P_{\upgamma, \Lambda, D}\) is closed.
        
        \item Convexity follows from (ii) and the fact that the intersection of convex sets is convex (the Frobenius ball \(\{\|W\|_F\le R\}\) is convex). Compactness follows because \(\mathscr P_{\upgamma,\Lambda,D}\) is closed by (ii), the Frobenius ball is compact, and a closed subset of a compact set is compact.

    \end{enumerate}
\end{proof}

The noncompactness of the set \(\mathscr P\) means that it cannot be sampled uniformly. Even \(\mathscr P_{\upgamma, \Lambda, D}\), the slice of \(\mathscr P\) corresponding to a fixed $D$ and $\Lambda$ with a given margin \(\upgamma>0\), is not compact. Therefore, 

based on Proposition~\ref{prop: LDS topology}(iii), we restrict our sample space to a reasonable and well-controlled compact subset, and include other parameters for sampling inputs.

\begin{definition}[Sample space for Monte Carlo]\label{def:sample_space}
    Given integer \(n > 0\),  $0<\updelta_0 \ll 1$, and \(u_0>0\), let \(Q_\Lambda =  Q_D = [\updelta_0, 1]^n\) and \(Q_U = [-u_0, u_0]^n\). Let \(D = \operatorname{diag}(\bf d)\), with \({\bf d} = [d_1, \dots, d_n] \in Q_D\), be the dissipation matrix, and let 
 \(\mathcal X = \{x \in \mathbb R^n | Dx \in [0, 1]^n \}\). For \(\upgamma>0\) and \(R > 0\), let the sample space be
    \begin{align}
        \mathcal{M}_\upgamma(\updelta_0, u_0) = &\{(u, x_0, D, \Lambda, W) \in (\mathbb R^n)^2\times (\mathbb S^n_{\succ 0})^2 \times \mathbb R^{n\times n}:
      \nonumber  \\ 
         &u\in Q_U, x_0 \in \mathcal X, D \in \operatorname{diag}(Q_D), \Lambda \in \operatorname{diag}(Q_\Lambda),\nonumber\\
        ~ &W\in \mathbb R^{n\times n}~\text{s.t.}~(W-D)^\top \Lambda + \Lambda (W-D) \preceq -\upgamma I, 
        \lVert W\rVert_F \leq R\} 
    \end{align}
\end{definition}

The set \(\mathcal{M}_\upgamma(\updelta_0, u_0)\) is appropriate to be sampled from uniformly due to its compactness, as we state next.

\begin{lemma}[Compactness of sampling set for Monte Carlo]
    Let \(0<\updelta_0 \ll 1,  u_0>0, \upgamma>0, \text{ and } R>0\). 
    Then, \(\mathcal{M}_\upgamma(\updelta_0, u_0)\) is compact.
\end{lemma}

\begin{proof}
    The result follows from Proposition~\ref{prop: LDS topology} and the application of Tychonoff's theorem~\citep{DGW:94}.
\end{proof}

Let \(\Delta_1, \dots, \Delta_{N}\) be \(N\) uniform samples of 5-tuples \((u_k, x_{0,k}, D_k, \Lambda_k, W_k)\) from \(\mathcal{M}(\updelta_0, u_0)\).
For each \(k\), let \(\upgamma_k:\mathcal{M}(\updelta_0, u_0) \rightarrow \{0, 1\}\) be the ground truth indicator given by
\begin{equation}
    \upgamma(\Delta_k) = \begin{cases}
        1, \quad \text{if GAS under } \Delta_k\\
        0, \quad \text{otherwise}
    \end{cases}
\end{equation}
Let \(\operatorname{FR}(\Delta)\) be the firing-rate network with dissipative and synaptic weights given by \(\Delta = (D, W)\), respectively.
The following result gives us the minimum number of samples needed for a given statement to hold probabilistically under a demanded guarantee.

\begin{lemma}[Probability estimation via Monte Carlo]\label{lem:prob-guarantee}
    Let \(A\) be the event `For every initial condition on \(\mathcal X\) and every input with entries between \(\pm u_0\), the trajectory of every firing-rate network with diagonal dissipativity \(D\succ0\) and synaptic matrix \(W\) satisfying the \(\mathcal{LDS}\) condition converges to a unique equilibrium point.'
    We define
    \begin{equation}
        p = \mathbb P\Big[A\Big]    
    \end{equation}
    and the estimate thereof
    \begin{equation}
        \hat p = \frac{1}{N} \sum_{k=1}^{N} \hat\upgamma(\Delta_k)
    \end{equation}
    where \(\hat \upgamma(\Delta_k)\) is the outcome of a sufficiently long numerical simulation of the trajectory of \(\operatorname{FR}(\Delta_k)\) that serves as an estimate of \(\upgamma(\Delta_k)\).
    Given \(\upepsilon>0\) and \(\updelta>0\), it holds that
    \begin{equation}
        \mathbb P[|\hat p - p| < \upepsilon] \geq 1-\updelta
    \end{equation} if
    \begin{equation}
        N \geq \ceil*{\frac{1}{2\upepsilon^2} \ln\frac{2}{\updelta}}.
    \end{equation}
\end{lemma}
Lemma~\ref{lem:prob-guarantee} follows from a straightforward application of Hoeffding's inequality, with a similar application derived in Section 8.3 of~\citep{RT-GC-FD:13}.

We now discuss the setup and findings of the Monte Carlo simulation.
\begin{table}[htb]
    \centering
   \begin{tikzpicture}[
  font=\sffamily\small,
  table/.style={
    matrix of nodes,
    nodes in empty cells,
    nodes={
      minimum height=0.7cm,
      text depth=0pt,
      text width=2.5cm,
      align=center,
      anchor=center
    },
    column sep=-\pgflinewidth,
    row sep=-\pgflinewidth,
    draw=gnblue6!70,
    thick
  },
  blocktitle/.style={
    font=\bfseries,
    text width=1.8cm,
    align=center,
    fill=gray!10,
    draw=gnblue6!70,
    thick
  },
  blockline/.style={line width=0.8pt},
  innerframe/.style={
    rounded corners=3pt,
    thick,
    draw=gnblue6!70!gnblue6,
    fill=gnblue6!5
  }
]

\matrix (tbl) [table] {
  & \textbf{Activation function} & \textbf{Success rate} \\  
  & ReLU & $1.0$ \\
  & Saturation$_{0}^{1}$ & $1.0$ \\
  & Sigmoid & $1.0$ \\
  & Logsumexp & $1.0$ \\
  & ReLU & $1.0$ \\
  & Saturation$_{0}^{1}$ & $1.0$ \\
  & Sigmoid & $1.0$ \\
  & Logsumexp & $1.0$ \\
  & ReLU & $1.0$ \\
  & Saturation$_{0}^{1}$ & $1.0$ \\
  & Sigmoid & $1.0$ \\
  & Logsumexp & $1.0$ \\
};

\begin{scope}[on background layer]
  \filldraw[innerframe]
    ([shift={(-0.2cm,0.2cm)}]tbl.north west)
    rectangle
    ([shift={(0.2cm,-0.2cm)}]tbl.south east);
\end{scope}

\node[blocktitle, minimum height=4*0.4cm, anchor=center]
  at (-3,2.35) {$n=3$};
\node[blocktitle, minimum height=4*0.4cm, anchor=center]
  at (-3,-0.35) {$n=4$};
\node[blocktitle, minimum height=4*0.4cm, anchor=center]
  at (-3,-3.05) {$n=5$};

\draw[blockline, draw=gnblue6!70] (tbl-1-1.south west) -- (tbl-1-3.south east); 
\draw[blockline, draw=gnblue6!70] (tbl-5-2.south west) -- (tbl-5-3.south east); 
\draw[blockline, draw=gnblue6!70] (tbl-9-2.south west) -- (tbl-9-3.south east); 

\end{tikzpicture}
    \caption{Monte Carlo results for the conjecture on firing-rate global asymptotic stability}
    \label{tab:MC}
\end{table}
\paragraph{Numerical validation of Conjecture~\ref{conj: GAS}.}
We have numerically tested Conjecture~\ref{conj: GAS} using a Python implementation available in the provided Zenodo code repository. 
Following the requirements derived in Lemma~\ref{lem:prob-guarantee}, we performed \(N = 27{,}000\) independent simulations, whose results we summarize in Table~\ref{tab:MC}, based on which we claim that the conjecture is true with observed probability \(\hat p = 1\) within error \(\epsilon = 0.01\) with high confidence \(1 - \updelta = 0.99\).
The parameters used for all the samples were \(\updelta_0 = 0.01\) and \(u_0 = 10\). 
The Frobenius bound on \(W\) was chosen to be \(R = 100\).
The sample drawn from $\mathcal{M}_\upgamma(\updelta_0, u_0)$ was retained only if the obtained \(W-D\) satisfies the Lyapunov Diagonal Stability (\(\mathcal{LDS}\)) condition for the solution \(\Lambda\), until a total of \(N\) \(\mathcal{LDS}\) samples were obtained.
For each combination of network dimension \(n \in \{3,4,5\}\) and activation function (ReLU, Saturation, Sigmoid, Logsumexp) we drew \(N = 27{,}000\) samples.

For each admitted sample, we evaluated GAS through a two-step procedure. First, we computed the equilibrium point \(x^\star\) using the nonlinear solver \texttt{scipy.fsolve} with tolerance \(10^{-5}\). Second, we integrated the dynamical system using the LSODA integrator (\texttt{scipy.integrate.LSODA}) over the time window \([0,\,T]\) with \(T = 5000\), producing an asymptotic state \(x^{\infty}\). 
The system was classified as GAS whenever
\[
\|x^\star - x^{\infty}\|_{\infty} < 10^{-4}.
\]
This threshold was chosen to ensure reliable numerical discrimination between true convergence and long-transient behavior.
In conclusion, the Monte Carlo simulation probabilistically establishes the nontrivial claim that under the \(\mathcal LDS\) condition, the firing-rate network is GAS for a wide range of activation functions satisfying Assumption~\ref{ass: sl-d}, at least for network dimension \(n<6\), though it likely holds for all \(n\).

\subsection{Wilson-Cowan model}
The Wilson-Cowan (WC) model~\citep{HRW-JDW:72} is a famous dynamical system capturing the interaction between an excitatory and an inhibitory neuron. 
Despite its simplicity, the model displays rich parameter and input dependent behavior that can result in a variety of dynamical outcomes.
\begin{figure}[!ht]
    \centering
    \includegraphics[width=0.6\linewidth]{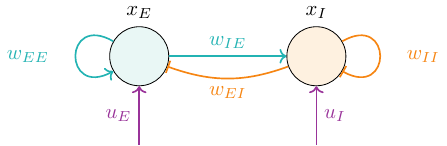}
    \caption{Schematic depiction of the Wilson-Cowan model, with one excitatory and one inhibitory neuron receiving external inputs.}
    \label{fig:WilsonCowan}
\end{figure}

The mathematical expression of the WC model for the saturating activation function is
\begin{equation}\label{eq: W-C}\tag{W-C}
    \begin{pmatrix}
        \dot x_{E}\\
        \dot x_{I}
    \end{pmatrix}
    =
    -D
    \begin{pmatrix}
        x_{E}\\
        x_{I}
    \end{pmatrix}
    +
    \Bigg[\begin{matrix}\underbrace{
        \begin{pmatrix}
            w_{EE} & -w_{EI}\\
            w_{IE} & -w_{II}
        \end{pmatrix}}_{\coloneqq W}
        \begin{pmatrix}
        x_{E}\\
        x_{I}
    \end{pmatrix}\vspace{-1em}
    +
    \underbrace{\begin{pmatrix}
        u_{E}\\
        u_{I}
    \end{pmatrix}}_{\coloneqq u}
    \end{matrix}\Bigg]_{0}^{1}
    \vspace{1em}
\end{equation}
where $D=\diag(d_{E},d_{I})$ with $d_{E},d_{I}>0$ and \(w_{XY}>0\) for all \(X, Y \in \{E, I\}\).
This WC model being a firing-rate model naturally inherits all of its equilibrium properties, in particular, Proposition~\ref{prop: surjFR} applies.
The linear-threshold model is given by
\begin{equation}\label{eq:FR_LTN} \tag{LTN}
    \dot x = -Dx + [Wx + u]_0^1,
\end{equation}
where specifically \(\Phi(\cdot) = [\cdot]_0^1\), with the clipping operator applied componentwise.
We henceforth call this system the Linear-Threshold Network (LTN).
We now discuss the dynamical behavior of this system, which is known to be able to show sustained oscillations and stable fixed-point behavior, and the conditions under which they arise.

\subsubsection{Conditions for limit cycle}
This section extends the results of~\citep{EN-JC:19,EN-RP-JC:22} to the \(D \neq I\) case.
For convenience, we define
\begin{equation}\label{eq:Delta_def}
\Delta\coloneqq\det(D-W)=w_{EI}w_{IE}-(w_{EE}-d_E)(w_{II}+d_I),
\end{equation}
and the input set
\begin{equation}\label{eq:Uint_def}
\mathcal U_\mathrm{int} = \Big\{(u_E, u_I)\ \Big|\ 0 < u_E < \tfrac{w_{EI}}{d_I} - \left(\tfrac{w_{EE}}{d_E} - 1\right),\quad 0 < \left(d_I + w_{II}\right) u_E - w_{EI}u_I < \tfrac{\Delta}{d_E}\Big\}.
\end{equation}
The following result provides the necessary and sufficient conditions for the emergence of a limit cycle. The results of~\citep{EN-JC:19,EN-RP-JC:22} are recovered by setting \(D = I\). We will provide later a novel game-theoretic interpretation of these conditions.

\begin{theorem}[Necessary and sufficient conditions for a stable limit cycle]\label{th:limcyc}
Consider the system~\eqref{eq: W-C} on \(\mathcal X\).
It admits a (nontrivial) stable limit cycle in \(\mathcal X\) if and only if
\begin{enumerate}[label=\roman*)]
\item \(U_\mathrm{int} \neq \emptyset\) and \(u \in \mathcal U_\mathrm{int}\), and
\item \(w_{EE} > w_{II} + d_E + d_I\).
\end{enumerate}
\end{theorem}
\begin{proof}
    The system~\eqref{eq: W-C} is a continuous planar switched-affine system on the positively invariant compact box \(\mathcal X\). 
    As in the planar E-I result of~\citep{EN-RP-JC:22}, the problem therefore reduces to characterizing when the system has no asymptotically stable equilibrium.
    We first determine over which switching regions the dynamics are stable. If \(x_E\) is inactive or saturated (call those regions \(\Omega_{(0, 0)}\), \(\Omega_{(0, l)}\), \(\Omega_{(0, s)}\), \(\Omega_{(s, 0)}\), \(\Omega_{(s, l)}\), \(\Omega_{(s, s)}\); the subscripts are \(0\), for inactive, \(l\) for active, and \(s\) for saturated), the Jacobian is one of
    \begin{equation}\label{eq:jacs1}
        \begin{pmatrix}
        -d_E & 0\\
        0 & -d_I
        \end{pmatrix},
        \qquad
        \begin{pmatrix}
        -d_E & 0\\
        w_{IE} & -(w_{II}+d_I)
        \end{pmatrix},
    \end{equation}
    both of which are Hurwitz for all admissible parameters.
    If \(x_E\) is active and \(x_I\) is inactive or saturated (call those regions \(\Omega_{(l, 0)}\), \(\Omega_{(l, s)}\)), then the Jacobian is
    \begin{equation}
        \begin{pmatrix}
        w_{EE}-d_E & -w_{EI}\\
        0 & -d_I
        \end{pmatrix},
    \end{equation}
    which is Hurwitz if and only if \(w_{EE}<d_E\).
    On the other hand, when both are active (\(\Omega_{(l, l)}\)), we have the Jacobian
    \begin{equation}
        \begin{pmatrix}
        w_{EE}-d_E & -w_{EI}\\
        w_{IE} & -(w_{II} + d_I)
        \end{pmatrix},
    \end{equation}
    whose trace is \(w_{EE}-w_{II}-d_E-d_I\) and determinant is \(\Delta\).
    
    So, \(w_{EE} \geq d_E\) is necessary for a stable limit cycle, as otherwise all switching regions are stable, as there is at least one fixed point (and in this case, a unique one).
    Next, to exclude the equilibria that arise in naturally stable regions (with Jacobians~\eqref{eq:jacs1}), we write out the candidate equilibria and exclude them case-by-case from appearing within their corresponding regions.
    Their equilibrium candidates are
    \begin{equation}
        x^*_{(0,0)}=\binom00,\qquad
        x^*_{(0,\ell)}=\binom{0}{\frac{u_I}{d_I+w_{II}}},\qquad
        x^*_{(0,s)}=\binom{0}{\frac1{d_I}},
    \end{equation}
    \begin{equation}
        x^*_{(s,0)}=\binom{\frac1{d_E}}0,\qquad
        x^*_{(s,\ell)}=\binom{\frac1{d_E}}{\frac{u_I+w_{IE}/d_E}{d_I+w_{II}}},\qquad
        x^*_{(s,s)}=\binom{\frac1{d_E}}{\frac1{d_I}}.
    \end{equation}
    
    Substituting the candidates into their expected equilibrium conditions, the candidates do not belong to their corresponding region if and only if the following hold:
    \begin{equation}\label{eq:aux1}
        u_E>0,
    \end{equation}
    \begin{equation}\label{eq:aux2}
        u_E<\frac{w_{EI}}{d_I}-\frac{w_{EE}-d_E}{d_E},
    \end{equation}
    \begin{equation}\label{eq:aux3}
        u_E>\min\Bigl\{\frac{w_{EI}}{d_I},\frac{w_{EI}}{d_I+w_{II}}u_I\Bigr\},
    \end{equation}
    \begin{equation}\label{eq:aux4}
        u_E<-\frac{w_{EE}-d_E}{d_E}+\max\Bigl\{0,\frac{w_{EI}}{d_I+w_{II}}\Bigl(u_I+\frac{w_{IE}}{d_E}\Bigr)\Bigr\}.
    \end{equation}
    Under~\eqref{eq:aux2}, one has \(u_E<\frac{w_{EI}}{d_I}\), so~\eqref{eq:aux3} is equivalent to
    \begin{equation}\label{eq:aux3b}
        (d_I+w_{II})u_E-w_{EI}u_I>0,
    \end{equation}
    which 
    we violate exactly when \(u \in \mathcal U_\mathrm{int}\).
    Also, because~\eqref{eq:aux1} holds and \(w_{EE}\ge d_E\), the zero branch of the maximum in~\eqref{eq:aux4} is impossible. Hence~\eqref{eq:aux4} is equivalent to
    \begin{equation}\label{eq:aux4b}
        (d_I+w_{II})u_E-w_{EI}u_I<\frac{\Delta}{d_E}.
    \end{equation}
    Combining~\eqref{eq:aux1}-\eqref{eq:aux4b}, the six stable regions contain no equilibrium candidate if and only if \(u\in\mathcal U_{\mathrm{int}}\).

    It remains to check the all-active region \(\Omega_{(l,l)}\) contains an equilibrium under \(u \in \mathcal U_\mathrm{int}\).
    The all-active equilibrium, if it exists, will be given by 
    \[x_\mathrm{int} = (D-W)^{-1}u = \frac{1}{\Delta} \binom{(d_I+w_{II})u_E-w_{EI}u_I}{w_{IE}u_E+(d_E-w_{EE})u_I}.\]
    Substituting this into the equilibrium condition, we get that \(u\in \mathcal U_\mathrm{int}\) implies that \(x_\mathrm{int} \in \Omega_{(l,l)}\), so that the Jacobian at \(x_\mathrm{int}\), namely \(-D + W\), cannot be Hurwitz.
    For \(-D+W\) to be Hurwitz it is necessary and sufficient that its trace be negative and its determinant be positive, i.e.,
    \[
    \operatorname{tr}(-D+W) = w_{EE}-w_{II}-d_E-d_I < 0
    \quad\text{and}\quad
    \operatorname{det}(-D+W) = \Delta = (d_E-w_{EE})(d_I+w_{II})+w_{EI}w_{IE} > 0.
    \]
    Under \(w_{EE}\geq d_E\), the first condition requires \(w_{EE}-w_{II}-d_E-d_I<0\) while the second, \(\Delta>0\), must hold for the interior equilibrium \(x_\mathrm{int}\) to exist (since \((D-W)^{-1}\) is used above). Hence \(-D+W\) is not Hurwitz if and only if the trace condition \(w_{EE}-w_{II}-d_E-d_I>0\) is violated.
    This gives us (ii).
\end{proof}

Now that we have the necessary and sufficient conditions for the existence of limit cycles for the 2-dimensional WC model, we can prove the following result.

\begin{proposition}[\textbf{Uniqueness of the equilibrium of the Wilson-Cowan}]
\label{prop: uniq/stab WC}
    Let $w_{EE}<d_{E}$. 
    Then for any input $u\in\real^2$ and bias $b\in\real^{2}$, system~\eqref{eq: W-C} admits a globally asymptotically stable equilibrium point.
\end{proposition}
\begin{proof}
    Under the hypothesis $w_{EE}<d_{E}$ we have that the matrix $D-W$ has principal minors:
    \textbf{Order 1:}
    \begin{align}
            d_{E}-w_{EE}&>0\\
            d_{I}+w_{II}&>0
    \end{align}
    \textbf{Order 2:}
    \begin{equation}
            (d_{E}-w_{EE})(d_{I}+w_{II})+w_{EI}w_{IE}>0
    \end{equation}
    and consequently, $D-W\in\mathcal{P}$. 
    Thus, Proposition~\ref{prop: surjFR} applies and we have a unique equilibrium.

    The unique equilibrium is locally stable, as seen by noticing that all the possible Jacobians, which are of the form \(\mathcal D_{x}(\Sigma)=-D + \Sigma W\) where \(\Sigma\in \operatorname{diag}\{0, 1\}^2\), are Hurwitz.
    Using Theorem~\ref{th:limcyc}, we rule out limit cycles as we have assumed \(w_{EE} < d_E\), contrary to what is necessary for limit cycles to appear.
    Since \eqref{eq: W-C} is forward invariant on the compact set \(\mathcal X = [0, 1/d_E]\times [0, 1/d_I]\), all its trajectories, by application of the Poincar\`e-Bendixson theorem, must go to either a limit cycle or an equilibrium point.
    Combined with local stability at the unique equilibrium and having ruled out limit  cycles, we conclude GAS.
\end{proof}

We now discuss a special two-player game associated with the two-population WC model.

\subsubsection{Zero-sum game for the Wilson-Cowan}
Once again recall~\eqref{eq:pseudograd}, which we rewrite as follows:
\begin{align}\label{eq:pseudogradient_WC}
    \tilde\nabla J (x,u) = (D - W)x - u
    = M\Lambda^{-1}\underbrace{\left(\Lambda M \left((D - W)x - u\right)\right)}_{=:N(x)},
\end{align}
where \(M = \operatorname{diag}(1, -1)\) and \(\Lambda = \operatorname{diag}(1/w_{EI}, 1/w_{IE})\). 
Then specifically for the PxGP corresponding to the system~\eqref{eq: W-C},
\begin{align}
    N(x)
    &=
    \begin{pmatrix}
        \tfrac{1}{w_{EI}} & 0\\
        0 & -\tfrac{1}{w_{IE}}
    \end{pmatrix}
    \begin{pmatrix}
        d_E - w_{EE} & w_{EI}\\
        -w_{IE} & d_I + w_{II}
    \end{pmatrix}
    \binom{x_E}{x_I}
    -
    \begin{pmatrix}
        \tfrac{1}{w_{EI}} & 0\\
        0 & -\tfrac{1}{w_{IE}}
    \end{pmatrix}
    \binom{u_E}{u_I}
    \notag
    \\
    &=
    \underbrace{\begin{pmatrix}
        \tfrac{d_E- w_{EE}}{w_{EI}} & 1\\[2pt]
        1 & -\tfrac{d_I + w_{II}}{w_{IE}}
    \end{pmatrix}}_{=:H}
    \binom{x_E}{x_I}
    -
    \binom{\tfrac{u_E}{w_{EI}}}{-\tfrac{u_I}{w_{IE}}}.
    \label{eq:sym_zsg_grad}
\end{align}
The benefit is that~\eqref{eq:sym_zsg_grad} is an affine map with symmetric linear part \(H\), hence \(N(x)\) arises as the true gradient of a scalar function. 
Define the quadratic cost
\begin{equation}\label{eq:ZSG_cost}
    J_\mathrm{ZSG}(x)
    =
    \frac{1}{2}\,x^\top H x
    -
    \binom{\tfrac{u_E}{w_{EI}}}{-\tfrac{u_I}{w_{IE}}}^\top x,
    \qquad x=\binom{x_E}{x_I}.
\end{equation}
Using the scalar \(J_\mathrm{ZSG}\), define the constrained zero-sum game
\begin{equation}\label{eq:ZSG_game}
    \min_{x_E\in \mathcal{X}_E}\max_{x_I \in \mathcal X_I} \ J_\mathrm{ZSG}(x_E, x_I)
\end{equation}
on \(\mathcal X = \mathcal X_E \times \mathcal X_I \coloneqq [0, 1/d_E]\times [0, 1/d_I]\).
We define the (unconstrained) descent/ascent strategy
\begin{align}\label{eq:strategy_desc/asc}
    \dot x = -\, M\Lambda^{-1}\nabla J_\mathrm{ZSG}(x) \eqqcolon -F(x).
\end{align}
The dynamics~\eqref{eq:strategy_desc/asc} is a minimax optimization, a weighted version of the typical gradient descent/ascent (see~\citep{AC-BG-JC:17}).
When the graph of~\eqref{eq:ZSG_cost} is a saddle surface, a minimax optimizer ideally should converge to the saddle point of the game.
We interpret all the findings concerning the WC dynamics in terms of the ZSG, to postulate what game behavior induces a particular dynamical behavior, and also show the connection between the proposed weighted strategy~\eqref{eq:strategy_desc/asc} and the system dynamics~\eqref{eq: W-C}.

\subsubsection{Cost functions that lead to global convergence}
We now interpret the LDS condition with regards to the zero-sum cost surface~\eqref{eq:ZSG_cost} and the WC model achieving convergence to the saddle point equilibrium of the game~\eqref{eq:ZSG_game}
regardless of the initial state.
This corresponds to monostable behavior: the network disregards any former activity and settles at a unique equilibrium decided by the inputs.
Recall that trajectories of~\eqref{eq: W-C} approach the unique equilibrium asymptotically, as proved in Proposition~\ref{prop: uniq/stab WC}, under \(W-D \in \mathcal LDS\).
For the WC model, the \(\mathcal{LDS}\) condition is exactly enforced by \(w_{EE} < d_E\).
An interesting observation is that the geometric condition of the graph of the ZSG cost in~\eqref{eq:ZSG_cost}

being a saddle surface is necessary but not sufficient for convergence of WC trajectories to the saddle. In fact, for the two-dimensional \(H\) matrix, to ensure indefiniteness of the Hessian of the cost, we only need to ensure \(\operatorname{det} H< 0\), which gives us
\begin{equation}
    (d_E - w_{EE})(d_I + w_{II}) + w_{EI}w_{IE} > 0
\end{equation}
which in terms of \(w_{EE}\) is
\begin{equation}\label{eq:saddle_cond_wEE}
    w_{EE} < d_E + \frac{w_{EI}w_{IE}}{d_I + w_{II}}.
\end{equation}
Satisfying~\eqref{eq:saddle_cond_wEE} ensures that the graph of the ZSG cost is a saddle surface, and since it is a weaker condition than \(w_{EE} < d_E\), it is necessary but not sufficient for the graph of the cost function to be a saddle surface for convergence to a unique equilibrium.
The condition \(w_{EE} < d_E\) itself in the game perspective is that the ZSG cost in~\eqref{eq:ZSG_cost} 

is strictly convex in \(x_E\) and strictly concave in \(x_I\), as the second derivative with respect to \(x_E\) is strictly positive (\(d_E - w_{EE} > 0\)) and with respect to \(x_I\) is strictly negative (\(-(w_{II} + d_I) < 0\)).
Next, we look at another regime of dynamical behavior, which takes us into the inhibitory-stabilized regime.

\subsubsection{Cost functions that induce a limit cycle}
We next interpret the oscillatory regime of the WC dynamics through the geometry of the zero-sum cost~\eqref{eq:ZSG_cost}. 
Depending on the application, this regime may represent a useful rhythmic mechanism, as in a central pattern generator, or an undesirable lack of convergence in a feedforward computational circuit. 
We do not re-derive the limit cycle conditions from the game formulation itself, but rather show that the same parameter regime identified in Theorem~\ref{th:limcyc} has a natural interpretation in terms of the associated quadratic cost, as shown in the convergent regime above.

From Theorem~\ref{th:limcyc}, the limit-cycle regime of~\eqref{eq: W-C} is characterized by the conditions
\begin{equation}
u\in\mathcal U_{\mathrm{int}},
\qquad
w_{EE}>w_{II}+d_E+d_I.
\end{equation}
The first condition ensures that the affine equilibrium
\begin{equation}\label{eq:xint}
x_{\mathrm{int}}=(D-W)^{-1}u=\frac1{\Delta}
\binom{(d_I+w_{II})u_E-w_{EI}u_I}{w_{IE}u_E+(d_E-w_{EE})u_I}
\end{equation}
lies in the interior of the active--active region. In particular, \(u\in\mathcal U_{\mathrm{int}}\) implies \(\Delta>0\). Using~\eqref{eq:sym_zsg_grad}, we then obtain \(\operatorname{det}H < 0\), as earlier, so the zero-sum cost \(J_{\mathrm{ZSG}}\) still has a saddle geometry in the oscillatory regime.
However, unlike the globally convergent regime, the cost is no longer convex in the minimizing variable \(x_E\). Indeed,
\begin{equation}
\frac{\partial^2 J_{\mathrm{ZSG}}}{\partial x_E^2}
=\frac{d_E-w_{EE}}{w_{EI}}<0
\end{equation}
whenever \(w_{EE}>d_E\), while
\begin{equation}
\frac{\partial^2 J_{\mathrm{ZSG}}}{\partial x_I^2}
=-\frac{d_I+w_{II}}{w_{IE}}<0
\end{equation}
always holds. 
Thus, although \(J_{\mathrm{ZSG}}\) remains saddle-shaped as a quadratic form, it loses the convex-concave structure that characterized the monostable regime. 
In particular, the excitatory coordinate no longer behaves as a minimizing direction.
This loss of convexity is consistent with the WC dynamics entering the oscillatory regime. In the active--active mode, the Jacobian is
\begin{equation}
\begin{pmatrix}
w_{EE}-d_E & -w_{EI}\\
w_{IE} & -(w_{II}+d_I)
\end{pmatrix},
\end{equation}
whose trace is \(w_{EE}-w_{II}-d_E-d_I\).
Therefore, under the condition \(w_{EE}>w_{II}+d_E+d_I\), the interior active-active equilibrium is unstable. 
The game-theoretic interpretation is then the following: the associated ZSG cost continues to possess saddle geometry, but its curvature is no longer aligned with the minimization/maximization roles of the two coordinates. 
This mirrors the loss of monostable convergence in the WC dynamics and is precisely the regime in which the planar system exhibits the limit-cycle behavior of Theorem~\ref{th:limcyc}.
\section{Functionality of E-I networks}
We now leverage the results of the previous section to study specific excitatory-inhibitory architectures beyond the WC model. 
In particular, we study a new model of lateral inhibition in the cortex resulting in winner-take-all (WTA) behavior, and eventually describe a more general columnar architecture that amplifies even small differences in input signals. 
We continue to consider only the subset of FR models that use the linear-threshold activation function.

\subsection{Lateral inhibition and winner-take-all dynamics in E$^{k}$-I models}
\subsubsection{The simple E$^{2}$-I network of binary decisions}
We start from the simplest setting of two excitatory and one inhibitory neuron implementing a biologically plausible WTA network. The two excitatory neurons are identical in their self-excitation and connectivity parameters. 
Furthermore, the two excitatory neurons are not interconnected and are given the same input except for a small difference that they are expected to detect and discriminate.
We make these notions concrete in this section, starting with defining the network.

\begin{definition}[\textbf{E$^{2}$-I network}]
    Let $D=\diag(d_{E},d_{I},d_{E})\succ 0$ be the dissipation matrix and  
    \begin{equation}\label{eq: SynW_SI_E2I}
        W = \begin{pmatrix}
            w_{EE} & -w_{EI} & 0\\
            w_{IE} & -w_{II} & w_{IE}\\
            0 & -w_{EI} & w_{EE}
        \end{pmatrix} \in \real^{3\times 3} 
    \end{equation}
    be the synaptic matrix.
    We define the E$^{2}$-I network dynamics as
    \begin{equation}\label{eq:E2I}\tag{E$^{2}$-I}
        \dot x =-Dx+[Wx+b+u]_{0}^{1}
    \end{equation}
    with state $x = (x_{E_{1}},x_{I},x_{E_{2}})$, internal bias $b\in\real^{3}$, and external input $u = (u_{E_1},\, u_I,\, u_{E_2})\in\real^{3}$.
\end{definition}

The first and third components of the neural state represent the activity of the two competing excitatory neurons, while the second component represents the activity of the inhibitory neuron. A schematic of this network is presented in Figure~\ref{fig:E$^{2}$I_net}.
\begin{figure}[!ht]
    \centering
    \includegraphics[width=.75\linewidth]{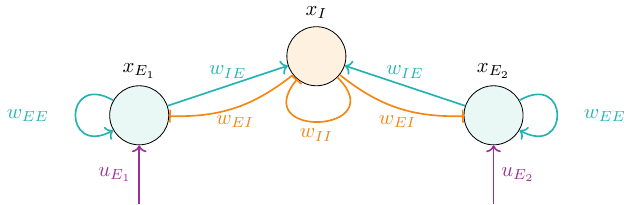}
    \caption{Schematic representation of the E$^{2}$-I network, with the two excitatory neurons on the sides and a central inhibitory neuron. 
    In the analyzed model, the only neurons that receive an external input are the excitatory neurons.}
    \label{fig:E$^{2}$I_net}
\end{figure}

We are interested in establishing WTA behavior, where the excitatory unit receiving the greater external input is the one that is ultimately fully activated, while  the one receiving the smaller external input is fully inhibited by its interaction with the inhibitory neuron. 
Maintaining the same ordering of the neural units, we are interested in equilibrium points proportional to
\begin{align}
    \upnu_{+}&=(1,\upalpha,0)\\
    \upnu_{-}&=(0,\upalpha,1)
\end{align}
or equivalently in terms of the true equilibria $\upxi_{+}=D^{-1}\upnu_{+}$ or $\upxi_{-}=D^{-1}\upnu_{-}$ for some \(\upalpha \in [0, 1]\).

\begin{assumption}[\textbf{E$^2$I inhibitory neuron insulation and normalized inputs}]\label{ass: insu}
    Let $u_{I}\in\real$ be the input to the inhibitory neuron and  $u_{E_{1}},u_{E_{2}}\in\real$ the inputs to the excitatory neurons. 
    We assume that
    \begin{enumerate}
        \item[(i)] the inhibitory neuron is insulated from events external to the E$^{2}$-I network, i.e., $u_{I}=0$;
    \item[(ii)] given actual external signals $\bar u_{E_{1}},\bar u_{E_{2}}$, we consider inputs to the excitatory neurons centered at zero namely, $u_{E_{1}}=\bar u_{E_{1}}-\bar u$ and $u_{E_2}=\bar u_{E_{2}}-\bar u$, with $\bar u=(\bar  u_{E_{1}}+\bar u_{E_{2}})/2$. 

    \end{enumerate}
\end{assumption}
Assumption~\ref{ass: insu} (ii) may be realized by absorbing the \textit{known} fixed mean into the internal bias \(b\).
We now formalize what we expect from the E$^2$I network for it to demonstrate lateral inhibition.

\begin{definition}[\textbf{Finite-Precision Lateral Inhibition (FLI)}]\label{def: fli}
    The network in~\eqref{eq:E2I} is laterally inhibitive with finite precision \(\updelta >0\) if, for excitatory inputs $(u_{E_{1}}, u_{E_{2}})\in\mathcal U \subset \real^{2}$ for an open set \(\mathcal U\) satisfying Assumption~\ref{ass: insu} and $u_{E_{1}}> \updelta,\ u_{E_{2}}< -\updelta$ (resp. $u_{E_{2}}>\updelta,\ u_{E_{1}}<- \updelta$), $\upxi_{+}=D^{-1}\upnu_{+}$ (resp. $\upxi_{-}=D^{-1}\upnu_{-}$) is the unique globally asymptotically stable equilibrium for the network activity.
\end{definition}

\subsubsection{Parameters of E$^{2}$-I networks leading to FLI}\label{sec:param_E2I}
Here we characterize the admissible parameter region that guarantees existence (and eventually under the \(\mathcal{LDS}\) condition, uniqueness and stability) of the desired FLI equilibrium, namely one proportional to $(1,\alpha,0)$ or $(0,\alpha,1)$ with $\alpha\in[0,1]$.
For simplicity, we set \(b = (b_E, b_I, b_E)^\top\).

\begin{lemma}[\textbf{Parameter conditions for existence of the FLI equilibrium}]\label{lem:param-existence}
Consider the system~\eqref{eq:E2I}.
For fixed $\updelta>0$, define input \(u\) satisfying Assumption~\ref{ass: insu} such that \(u_{E_1} \geq \updelta\), \(u_{E_2} \leq -\updelta\), and \(u_{I} = 0\), and take any bias \(b=(b_E,b_I,b_E)^\top\).
Let \(\upnu_{+}=(1,\alpha,0)^\top\) and \(\upxi_{+}=D^{-1}\upnu_{+}\).  
Then \(\upxi_{+}\) is an equilibrium of~\eqref{eq:E2I} for some \(\alpha\in[0,1]\) if and only if
\begin{equation}\label{eq:bounds-cond-lemma}
d_E^{-1}w_{EE}-d_I^{-1}\alpha w_{EI}+b_E+\updelta\ge 1,\qquad -d_I^{-1}\alpha w_{EI}+b_E-\updelta\le 0,
\end{equation}
and \(\alpha\) satisfies the inhibitory fixed-point condition
\begin{equation}\label{eq:I-balance-clamped}
\alpha=\big[d_E^{-1}w_{IE}-d_I^{-1}\alpha w_{II}+b_I\big]_0^1.
\end{equation}
Equivalently,~\eqref{eq:I-balance-clamped} holds if and only if one of the following cases holds:
\begin{equation}\label{eq:I-balance-cases}
\begin{cases}
d_E^{-1}w_{IE}-d_I^{-1}\alpha w_{II}+b_I=\alpha, \quad&0<\alpha<1,\\
d_E^{-1}w_{IE}+b_I\le 0, &\alpha=0,\\
d_E^{-1}w_{IE}-d_I^{-1}w_{II}+b_I\ge 1, &\alpha=1.
\end{cases}
\end{equation}
The corresponding conditions for \(\upxi_{-}=D^{-1}(0,\alpha,1)^\top\) under input \(u\) such that \(u_{E_1} \leq -\updelta\), \(u_{E_2} \geq \updelta\) are identical by symmetry.
\end{lemma}
\begin{proof}
The equilibrium condition \(\dot x=0\) is equivalent to \(y=[WD^{-1}y+b+u]_0^1\) with \(y=Dx\).  
Imposing \(y=\upnu_{+}\) yields the two excitatory saturation requirements, which are exactly~\eqref{eq:bounds-cond-lemma}.
For the inhibitory coordinate, using the second row of \(W\) in~\eqref{eq: SynW_SI_E2I} gives \(\alpha=[w_{IE}x_{E_1}-w_{II}x_I+w_{IE}x_{E_2}+b_I]_0^1=[d_E^{-1}w_{IE}\cdot 1-d_I^{-1}w_{II}\alpha+d_E^{-1}w_{IE}\cdot 0+b_I]_0^1\), which is~\eqref{eq:I-balance-clamped}.  
Finally,~\eqref{eq:I-balance-cases} is the saturation case split of~\eqref{eq:I-balance-clamped}.
\end{proof}

We are ready to provide a criterion to establish finite-precision lateral inhibition.

\begin{proposition}[\textbf{Conditions for FLI}]\label{prop: FLI}
Assume the hypotheses of Lemma~\ref{lem:param-existence} holds for $\updelta>0$ and let $d_E=d_I=d$ with $w_{EE}<d$.  
Then the network~\eqref{eq:E2I} is FLI with finite precision $\updelta$.
\end{proposition}
\begin{proof}
    Notice that Lemma~\ref{lem:param-existence} guarantees that $\upxi_{+}$ as an equilibrium exists (resp. $\upxi_{-}$).
    Note that $k_{in}^{E}=1$, $k_{out}^{E}=1$, and therefore the condition~\eqref{eq:LDS-EE}  becomes $d-w_{EE}(k_{in}^E+k_{out}^{E})/2=d-w_{EE} > 0$, which holds by hypothesis $w_{EE}<d$. On the other hand, the condition~\eqref{eq:LDS-II} is trivially satisfied, as the network just has one inhibitory neuron. 
    By Theorem~\ref{thm: LDS}, we conclude that \(W-D \in \mathcal{LDS}\) and, by Theorem~\ref{thm: EIglob}, that \(\upxi_{+}\) (resp. \(\upxi_{-}\)) is GAS.
\end{proof}

Even though Proposition~\ref{prop: FLI} establishes finite-precision lateral inhibition only for the scalar dissipation case $D=d\mathcal{I}_{3}$, the empirical evidence about Conjecture~\ref{conj: GAS} strongly suggests that the same finite-precision lateral inhibition holds also for general, positive diagonal matrices $D\in\real^3$ as long as $W-D\in\mathcal{LDS}$.
We now study the admissible parameter ranges of $(d_E,w_{EE},d_I,w_{EI})$ so that it is easy to pick parameters and understand the effect of changing them.
\begin{lemma}[\textbf{Admissible excitatory self-coupling range}]\label{lem:wEE-range}
Assume~\eqref{eq:bounds-cond-lemma} holds 

and $w_{EE}<d_E$.  Then
\begin{equation}\label{eq:wEEsp}
    1-2\updelta \le d_{E}^{-1}w_{EE} < 1.
\end{equation}
\end{lemma}
\begin{proof}
From the first inequality in~\eqref{eq:bounds-cond-lemma},
$d_{E}^{-1}w_{EE}-d_{I}^{-1}\alpha w_{EI}\ge 1-\updelta-b_E$.  
From the second inequality,
$-d_{I}^{-1}\alpha w_{EI}\le \updelta-b_E$.  
Substituting yields $d_{E}^{-1}w_{EE}\ge 1-2\updelta$, and the result follows. 
\end{proof}

\begin{remark}[\textbf{Interdependence of synaptic weights and FLI precision}]\label{rem:crit_prec}
If $\updelta>1/2$, then $1-2\updelta<0$. Since $d_{E}^{-1}w_{EE}>0$,~\eqref{eq:wEEsp} reduces to $0<d_{E}^{-1}w_{EE}<1$, which 

is independent of~$\updelta$.  
Thus $\updelta=1/2$ is the critical precision beyond which the admissible excitatory self-coupling region is unaffected by the FLI requirement.  
In particular, when $|u_{E_1}-u_{E_2}|>1$, satisfying $\mathcal{LDS}$ is sufficient.
\end{remark}

We next bound the inhibitory-to-excitatory weight.
 
\begin{lemma}[\textbf{Bounds on inhibitory-to-excitatory coupling}]\label{lem:wEI-range}
Assume~\eqref{eq:bounds-cond-lemma} holds and $w_{EE}<d_{E}$.  
Then
\begin{equation}
    0 \le d_{I}^{-1}\alpha w_{EI} < b_E+\updelta.
\end{equation}
\end{lemma}
\begin{proof}
Since $d_{E}^{-1}w_{EE}<1$, the first inequality in~\eqref{eq:bounds-cond-lemma} gives $d_{E}^{-1}w_{EE}-d_{I}^{-1}\alpha w_{EI}\ge 1-\updelta-b_E$, hence $d_{I}^{-1}\alpha w_{EI}\le d_{E}^{-1}w_{EE}-(1-\updelta-b_E)<b_E+\updelta$, where the strict inequality uses $d_{E}^{-1}w_{EE}<1$.
Nonnegativity follows from the biological sign constraint $w_{EI}\ge 0$.
\end{proof}

\subsubsection{Fast I Slow E system}

Up to now we have dealt with a~\ref{eq:E2I} system where all the neural units evolve with the same unitary timescale. However, the dynamics of different neurons may be characterized by different timescales. Here, we are particularly interested in studying the asymptotics of a system where the timescale for the inhibitory neuron is much faster than that for excitatory neurons.

\begin{definition}[\textbf{Fast I Slow E (FISE) E$^{2}$-I}]\label{def: FISE-E2I}
    Let $D=\diag(d_{E},d_{I},d_{E})$ and define $0<\varepsilon\ll 1$. Then, for $b,u\in\real^{3}$, we define the FISE E$^{2}$-I system as
    \begin{equation}\label{eq: FISE-E2I}\tag{FISE}
        \begin{pmatrix}
            \dot x_{E_{1}}\\
            \varepsilon\dot x_{I}\\
            \dot x_{E_{2}}
        \end{pmatrix}
        =
        -D
        \begin{pmatrix}
            x_{E_{1}}\\
            x_{I}\\
            x_{E_{2}}
        \end{pmatrix}
        +
        \begin{bmatrix}
            \begin{pmatrix}
                w_{EE} & -w_{EI} & 0\\
                w_{IE} & -w_{II} & w_{IE}\\
                0 & -w_{EI} & w_{EE}
            \end{pmatrix}
            \begin{pmatrix}
            x_{E_{1}}\\
            x_{I}\\
            x_{E_{2}}
        \end{pmatrix}
        + b + u
        \end{bmatrix}_{0}^{1}
    \end{equation}
\end{definition}

\begin{remark}[Asymptotic behavior of the system~\ref{eq: FISE-E2I}]\label{rem:asymptotic-FISE}
    Let $\Lambda_{\varepsilon}=\diag(1,\varepsilon,1)$, with $\varepsilon>0$, be the timescale matrix associated to the system~\eqref{eq: FISE-E2I}, and define $f_{\text{E$^{2}$I}}$ and $f_{\text{FISE}}$ as the vector fields in~\eqref{eq:E2I} and in~\eqref{eq: FISE-E2I}, respectively. It follows that for every $x\in\real^3$
    \begin{equation}
        f_{\text{FISE}}(x)=\Lambda_{\varepsilon}^{-1}f_{\text{E$^{2}$I}}(x).
    \end{equation}
    It then easily follows that the systems~\ref{eq:E2I} and~\ref{eq: FISE-E2I} have the same equilibrium points. In particular, if $W-D\in\mathcal{LDS}$, then also system~\ref{eq: FISE-E2I} has a unique equilibrium point. However, knowledge of system~\ref{eq:E2I} being GAS does not transfer to system~\ref{eq: FISE-E2I} being GAS as well.
\end{remark}

Taking the limit $\varepsilon\to 0^+$ in the system~\eqref{eq: FISE-E2I}, we obtain the following reduced model
\begin{equation}\tag{R-E$^{2}$-I}
   \begin{aligned}\label{eq: red_E2I}
        \dot{\bar{x}}_{E_{1}}&=-d_{E}\bar x_{E_{1}}+[w_{EE}\bar x_{E_{1}}-w_{EI}x_{I}^{\star}(\bar x_{E_{1}},\bar x_{E_{2}})+b_{1}+u_{E_{1}}]_{0}^{1}\\
        \dot{\bar{x}}_{E_{2}}&=-d_{E}\bar x_{E_{2}}+[w_{EE}\bar x_{E_{2}}-w_{EI}x_{I}^{\star}(\bar x_{E_{1}},\bar x_{E_{2}})+b_{3}+u_{E_{2}}]_{0}^{1}
\end{aligned} 
\end{equation}
where $x_{I}^{\star}:\real^{2}\to\real$ is the solution to the equation $0=-d_{I}x_{I}^{\star}(\bar{x}_{E_{1}},\bar x_{E_{2}})+[w_{IE}(\bar x_{E_{1}}+\bar x_{E_{2}})-w_{II}x_{I}^{\star}(\bar{x}_{E_{1}},\bar x_{E_{2}})+b_{2}+u_{I}]_{0}^{1}$. Next, we show that such the function $x^{\star}_{I}:\real^{2}\to\real$ is well defined, and study the asymptotics of systems~\ref{eq: FISE-E2I} and~\ref{eq: red_E2I} under slightly stronger assumptions than $\mathcal{LDS}$.

\begin{proposition}[\textbf{Asymptotics of the reduced model}]\label{prop:asymptotics-reduced}
    Let $d_{E}>w_{EE}$ and $d_{I}>w_{II}$. Let $x_{E}:\real_{\geq 0}\to\real^{2}$ be the trajectory of the excitatory neurons of system~\ref{eq: FISE-E2I} and $\bar x_{E}:\real_{\geq 0}\to\real^{2}$ the trajectory of system~\ref{eq: red_E2I}. Then
    \begin{enumerate}
        \item[(i)] it holds in general that
        \begin{equation}
            \limsup_{t\to +\infty} \|\bar x_{E}(t)-x_{E}(t)\|_{2}\leq\varepsilon h
        \end{equation}
        where $h>0$ is a finite constant.

        \item[(ii)] if, in addition, the system~\ref{eq: FISE-E2I} is GAS with the excitatory components of the equilibrium given by $x_{E}^{\star} = (x^{\star}_{E_{1}},x^{\star}_{E_{2}}) \in\real^2$, then for any choice of initial conditions $\bar x_{E}(0),x_{E}(0)\in\real^2$,
    \end{enumerate}
    \begin{equation}
        \limsup_{t\to +\infty} \|\bar x_{E}(t)-x_{E}^{\star}\|_{2}=0.
    \end{equation}
\end{proposition}
\begin{proof}
    We first show (i). Notice that, when $d_{E}>w_{EE}$, we have $W-D\in\mathcal{LDS}$, and consequently the existence, location, and uniqueness of the equilibrium $x^{\star}$ of~\ref{eq:E2I} extends also to the time-scaled system~\ref{eq: FISE-E2I}, cf. Remark~\ref{rem:asymptotic-FISE}. Define now the scalar functions
    \begin{align}
        f(x,z) &= -d_{E}x + [w_{EE}x-w_{EI}z+b+u]_{0}^{1}\\
        g(y,z) &= -d_{I}z + [w_{IE}y-w_{II}z+b+u]_{0}^{1}
    \end{align}
    for $x\in[0,1]$, $z\in[0,1]$, and $y\in[0,2]$. It is then easy to check that 

    \begin{align}
        (f(x_{1},z)-f(x_{2},z))(x_{1}-x_{2})&=-d_{E}(x_{1}-x_{2})^{2}+([w_{EE}x_{1}-w_{EI}z+b+u]_{0}^{1}-[w_{EE}x_{2}-w_{EI}z+b+u]_{0}^{1})(x_{1}-x_{2})\nonumber\\
        &\leq -d_{E}(x_{1}-x_{2})^{2}+|(w_{EE}x_{1}-w_{EI}z+b+u-w_{EE}x_{2}+w_{EI}z-b-u)(x_{1}-x_{2})|\nonumber\\
        &\leq -(d_{E}-w_{EE})(x_{1}-x_{2})^{2}\qquad \forall z\in[0,1] ,
    \end{align}
    where in the first inequality we have used Assumption~\ref{it: weak} on the monotonicity and Assumption~\ref{it: nonexp} on the non-expansiveness of the activation function.
    Thus,  the scalar function $f$ associated to the excitatory neurons is partially contracting in the 2-norm w.r.t. the variable $x$ and with contraction rate $c_{E}=(d_{E}-w_{EE})>0$ for all $z\in[0,1]$. Analogously
    \begin{align}
        (g(y,z_{1})-g(y,z_{2}))(z_{1}-z_{2})&=-d_{I}(z_{1}-z_{2})^{2}+([w_{IE}y-w_{II}z_{1}+b+u]_{0}^{1}-[w_{IE}y-w_{II}z_{2}+b+u]_{0}^{1})(z_{1}-z_{2})\nonumber\\
        &\leq -d_{I}(z_{1}-z_{2})^{2}+|(w_{IE}y-w_{II}z_{1}+b+u-w_{IE}y+w_{II}z_{2}-b-u)(z_{1}-z_{2})|\nonumber\\
        &\leq -(d_{I}-w_{II})(z_{1}-z_{2})^{2}\qquad \forall z\in[0,1].
    \end{align}
    Thereby, the scalar function $\varepsilon^{-1}g$ associated to the inhibitory neuron is partially contracting in the 2-norm w.r.t. the variable $z$ and with contraction rate $c_{I}=\varepsilon^{-1}(d_{I}-w_{II})$ for all $y\in[0,2]$.
    Notice that in our case we have $y=x_{E_{1}}+x_{E_{2}}$. Since $g$ is partially contracting w.r.t. to $z$, then for any $y\in[0,2]$ there exists a unique $z(y)$ such that $0=g(y,z(y))$ \citep[Theorem 2]{DDV-JJS:13}.
    Identifying now with $x_{E}(t)=(x_{E_{1}}(t),x_{E_{2}}(t))$ the trajectory of the excitatory neurons in~\eqref{eq: FISE-E2I} and exploiting \citep[Theorem 3]{DDV-JJS:13}, we have
    \begin{align}
        \|\bar x_{E}(t)-x_{E}(t)\|_{2}&\leq \uplambda e^{-c_{E}t}\|\bar x_{E}(0)-x_{E}(0)\|_{2}+\varepsilon[k(c_{E},c_{I},\varepsilon,t)+h(c_{E},t)]\nonumber\\
        &\leq \uplambda e^{-c_{E}t}\|\bar x_{E}(0)-x_{E}(0)\|_{2}+\varepsilon[\kappa(c_{E},c_{I}) e^{-c_{E}t}+h(c_{E},t)]
    \end{align}
    where $\kappa(c_{E},c_{I})>0$ is finite, and $h(c_{E},t)=h(1-e^{-c_{E}t})>0$ is finite $\forall t>0$, with $h>0$ (see \citep[Theorem 3]{DDV-JJS:13} for additional details). Let $h=\sup_{t\geq 0} h(c_{E},t)$ and taking the limit, we have
    \begin{align}
        \limsup_{t\to+\infty}\|\bar x_{E}(t)-x_{E}(t)\|_{2} \leq\varepsilon h.
    \end{align}

    Next, we show (ii). If the system~\ref{eq: FISE-E2I} is GAS, its GAS equilibrium point will be the same as that of~\ref{eq:E2I}, and in particular the excitatory components of the equilibrium $x_{E}^{\star} = (x^{\star}_{E_{1}},x^{\star}_{E_{2}}) \in\real^2$ will be independent of $\varepsilon>0$. Since also the dynamics of~\eqref{eq: red_E2I} are independent of $\varepsilon$, then it holds that
    \begin{align}
        \limsup_{t\to+\infty}\|\bar x_{E}(t)-x_{E}^{\star}\|_{2}
        &\leq \limsup_{t\to+\infty}\|\bar x_{E}(t)-x_{E}(t)\|_{2}+\underbrace{\lim_{t\to+\infty}\|x_{E}(t)-x_{E}^{\star}\|_{2}}_{= 0} \leq \varepsilon h
    \end{align}
    and taking $\varepsilon\to 0^+$ we conclude the result.
\end{proof}

\paragraph{Biological underpinnings of the reduced E$^{2}$-I model and visualization of the reduced excitatory energies.}
\begin{figure}[!tb]
    \centering
    \subfloat{\includegraphics[width=\linewidth]{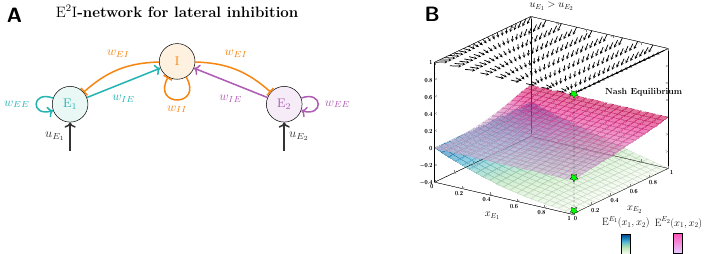}}     
    \caption{\textbf{Schematic of a E$^{2}$-I excitatory-inhibitory circuit and visualization of the interacting energies under $\mathcal{LDS}$ constraints.} (A) Schematic of a minimal E$^{2}$-I circuit composed of two excitatory neurons ($E_{1}$, $E_{2}$) interacting through a shared inhibitory interneuron. Each excitatory neuron receives an external input ($u_{E_{1}}$, $u_{E_{2}}$), which is forwarded to the inhibitory neuron. The interneuron in turn relays equal inhibitory feedback to both excitatory cells, effectively suppressing the activity of the neuron receiving the weaker excitatory drive. (B) Energies associated with the reduced excitatory subsystem obtained by eliminating the inhibitory variable through $x_{I}\equiv \bar{x}_{I}(x_{E_{1}},x_{E_{2}})$, under the condition $u_{E_{1}}>u_{E_{2}}$. Both excitatory neurons are characterized by convex individual energy landscapes, yet $\mathrm{E}^{E_{2}}$ exhibits a local maximum around $x_{E_{2}}=1$, driving its activity toward the inactive state ($x_{E_{2}}=0$). Conversely, $\mathrm{E}^{E_{1}}$ is minimized at $x_{E_{1}}=1$, promoting activation. The joint dynamics settle at the Nash equilibrium defined by the intersection of these tendencies—$x_{E_{1}}=1$, $x_{E_{2}}=0$—in full agreement with the theoretical predictions for E$^{2}$-I circuits under the $\mathcal{LDS}$ condition.}
    \label{fig: EIE}
\end{figure}
Excitatory neurons and inhibitory neurons are known for their different frequency of firing with respect to the same stimulus~\citep{CBW-GMJ:90}, with inhibitory neurons being capable of higher, sustained rates of firing. In mathematical terms, the different firing frequency between the two classes of neurons translates in the inhibitory time constant $\uptau_{I}$ being smaller than the excitatory time constant $\uptau_{E}$. \begin{equation}\label{eq: FISE-time}
        \begin{pmatrix}
            \uptau_{E}\dot x_{E_{1}}\\
            \uptau_{I}\dot x_{I}\\
            \uptau_{E}\dot x_{E_{2}}
        \end{pmatrix}
        =
        -D
        \begin{pmatrix}
            x_{E_{1}}\\
            x_{I}\\
            x_{E_{2}}
        \end{pmatrix}
        +
        \begin{bmatrix}
            \begin{pmatrix}
                w_{EE} & -w_{EI} & 0\\
                w_{IE} & -w_{II} & w_{IE}\\
                0 & -w_{EI} & w_{EE}
            \end{pmatrix}
            \begin{pmatrix}
            x_{E_{1}}\\
            x_{I}\\
            x_{E_{2}}
        \end{pmatrix}
        + b + u
        \end{bmatrix}_{0}^{1}
    \end{equation}
Under the time re-parametrization $\uptau_{E} s = t$ and setting $\varepsilon=\uptau_{I}/\uptau_{E}$, we can rewrite system~\eqref{eq: FISE-time} in $s$-time as presented in~\eqref{eq: FISE-E2I}. Exploiting time-separation for the dynamics, we can take the limit $\varepsilon\to 0$ and consider the reduced model where each excitatory neuron has dynamics
\begin{equation}
    \dot{x}_{E_{i}}=-d_Ex_{E_{i}}+[w_{EE}x_{E_{i}}-w_{EI}\bar{x}_{I}(x_{E_{1}},x_{E_{2}})+u_{E_i}]_{0}^{1}
\end{equation}
for $i=1,2$. In the reduced model, $\bar{x}_{I}(x_{E_{1}},x_{E_{2}})$ is the solution to the equation $0=-d_{I}\bar{x}_{I}(x_{E_{1}},x_{E_{2}})+[w_{IE}( x_{E_{1}}+x_{E_{2}})-w_{II}\bar{x}_{I}(x_{E_{1}},x_{E_{2}})+b_{2}+u_{I}]_{0}^{1}$. In the specific case where we have absence of self-inhibition for the inhibitory neuron $w_{II}=0$, we derive the explicit inhibitory contribution
\begin{equation}
    0=-d_I\bar{x}_{I}(x_{E_1},x_{E_2})+[w_{IE}(x_{E_{1}}+x_{E_2})]_{0}^{1}
\end{equation}
Under the hypothesis $w_{EE}<d_{E}$, we can apply Proposition~\ref{prop:asymptotics-reduced}, and the reduced system $\bar x_{E}(t)$ has the same asymptotic behavior of the two excitatory neurons in the full E$^{2}$-I network. Reducing the full E$^{2}$-I system to the reduced dynamics allows for a qualitative understanding of the interplay between the cost functions $\En^{E_1}(x_{E_{1}},\bar{x}_{I},x_{E_{2}},u_{E_1})$ associated to the first excitatory neuron and $\En^{E_2}(x_{E_{1}},\bar{x}_{I},x_{E_{2}},u_{E_2})$ associated to the second excitatory neuron. Specifically, the reduced energies computed at equilibrium $x_{I}\equiv \bar{x}_{I}(x_{E_1},x_{E_2})$ are
\begin{equation}
    \En^{E_i}=-x_{E_i}\left[\frac{1}{2}w_{EE}x_{E_i}-w_{EI}\bar{x}_{I}+u_{E_{i}}\right]+\frac{1}{2}x_{E_i}^{2}
\end{equation}
For a choice of external inputs $u_{E_1},u_{E_2}$ satisfying the $\updelta$-precision hypothesis and such that $u_{E_1}>u_{E_2}$, we can observe from Fig.~\ref{fig: EIE}(B) that the neural activity converges to the reduced-model equilibrium $(1,0)$, associated to the desired equilibrium $x^{\star}_{1}=(1,1,0)$ for the full model, from every point of state-space.
In addition, both energies $\En^{E_1}(x_{E_{1}},\bar{x}_{I},x_{E_{2}},u_{E_1})$ and $\En^{E_2}(x_{E_{1}},\bar{x}_{I},x_{E_{2}},u_{E_2})$ present a convex surface with non-aligned minima, and $(1,0)$ is the only point where both neurons cannot improve their action profile.

\subsubsection{A biologically plausible winner-take-all network}
\begin{figure}[ht!]
    \centering
    \includegraphics[width=0.7\linewidth]{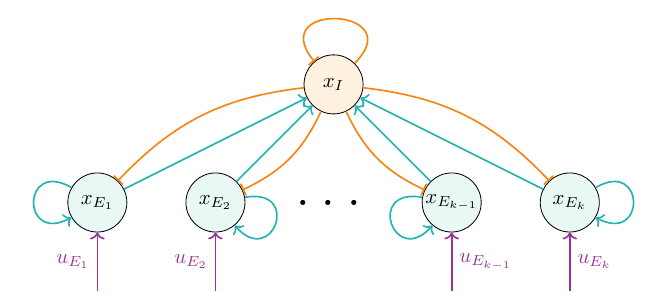}
    \caption{Star topology of a biologically plausible winner-take-all network.}
    \label{fig:WTA}
\end{figure}
Consider now a generic network having an arbitrary number $k\in\mathbb{N}$ of excitatory neurons connected to one central inhibitory neuron. 
We show how the associated dynamics~\eqref{eq: EIdyn} can potentially implement a biologically plausible WTA model~\citep[Chapter 3.2]{FB:25nn}. We define this network below.
\begin{definition}[\textbf{E$^{k}$-I network}]
    Let the dissipation matrix be $D=\diag(\underbrace{d_{E},\dots,d_{E}}_{k},d_{I})\in\real^{(k+1)\times (k+1)}$ and the synaptic matrix be
    \begin{equation}\label{eq:EkI synapse}
        W=
        \begin{pmatrix}
            w_{EE} & 0 & \dots & -w_{EI}\\
            0 & w_{EE} & \dots & -w_{EI}\\
            0 & \dots & \ddots & -w_{EI}\\
            w_{IE} & w_{IE} & \dots & -w_{II}
        \end{pmatrix} \in\real^{(k+1)\times (k+1)}.
    \end{equation}
    We define the E$^{k}$-I network dynamics as
    \begin{equation}\label{eq:EkI}\tag{E$^{k}$-I}
        \dot x =-Dx+[Wx+b+u]_{0}^{1}
    \end{equation}
    with state $x = (x_{E_{1}},\dots,x_{E_{k}}, x_{I})$, internal bias $b\in\real^{k+1}$, and external input $u\in\real^{k+1}$.
\end{definition}

We formalize next the generalization of an FLI-performing network to multiple signals $u\in\real^{k+1}$. We first make assumptions about the E$^k$I input vector, which extends the notion of Assumption~\ref{ass: insu}.

\begin{assumption}[E$^k$I inhibitory neuron insulation and normalized inputs]\label{ass: insu-k}
    Let $u_{I}\in\real$ be the input to the inhibitory neuron and  $(u_{E_{1}},\dots, u_{E_{k}})^\top\in\real^k$ be the inputs to the excitatory neurons. 
    We assume that
    \begin{enumerate}
        \item[(i)] the inhibitory neuron is insulated from events external to the E$^{k}$-I network, i.e., $u_{I}=0$;
        \item[(ii)] given actual external signals $\bar u_{E_{1}},\dots,\bar u_{E_{k}}$, we consider inputs to be such that $u_{E_{i}}=\bar u_{E_{i}}-\bar u$ for all \(i \in [k]\), with $\bar u=(\bar u_{E_{m}}+\max_{j \in [k] \setminus\{m\}}\bar u_{E_{j}})/2$, where \(m :=\argmax{i\in [k]}{\bar u_{E_i}}\), which we assume is the unique largest input.

    \end{enumerate}
\end{assumption}

This can be similarly realized by absorbing the mean of the highest and second highest inputs into the internal bias.
Notice that setting \(k=2\), we recover Assumption~\ref{ass: insu}.
We now make the notion of being able to distinguish the highest input among \(k\) inputs using the E$^{k}$-I network precise. The following generalizes Definition~\ref{def: fli}.

\begin{definition}[\(k\)-signal Finite-Precision Lateral Inhibition ($k$-FLI)]\label{def: fli-k}
    The network in~\eqref{eq:EkI} is laterally inhibitive with finite precision \(\updelta >0\) if, for any $u = (u_{E_{1}},\dots, u_{E_{k}}, u_{I})\in\real^{k+1}$ satisfying Assumption~\ref{ass: insu-k}
    and having unique $i\in\{1,\dots,k\}$ such that
    \begin{align}
        u_{E_{i}}&\geq\updelta\\
        u_{E_{j}}&\leq-\updelta\qquad \forall j\in\{1,\dots,k\},\ j\neq i
    \end{align} for some \(\alpha \in [0,1]\), the point
    \begin{equation}\label{eq: k_FLI_eq}
        x^{\star}= D^{-1}[\vect{e}_{i}\ \alpha]\in\real^{k+1}
    \end{equation} is the unique stable equilibrium for the network activity.
\end{definition}

Here $\vect{e}_{i}\in\real^{k}$ is the $i^{th}$ canonical basis vector in $\real^{k}$. 
The following result generalizes Lemma~\ref{lem:param-existence} and establishes conditions for the existence of  $k$-FLI equilibrium.

\begin{lemma}[\textbf{Parameter conditions for existence of $k$-FLI equilibrium}]\label{lm: k-bin}
    The network~\eqref{eq:EkI} admits the FLI equilibrium for finite precision \(\updelta>0\), i.e., \(x^*\) as defined in~\eqref{eq: k_FLI_eq} with \(\alpha \in (0,1)\), if
    \begin{subequations}\label{eq:EkI ineq}
    \begin{align}
        &d_{E}^{-1}w_{EE}-d_{I}^{-1}\alpha w_{EI}+ b_E + \updelta \geq 1\\
        &-d_{I}^{-1}\alpha w_{EI}+ b_E - \updelta \leq 0\qquad\\
        &d_{E}^{-1}w_{IE}-d_{I}^{-1}\alpha w_{II} + b_I = \alpha \label{eq:I_FLI-k}
    \end{align}
    \end{subequations}
    If \(\alpha = 1\) (resp. \(\alpha = 0\)), \eqref{eq:I_FLI-k} becomes \(d_E^{-1}w_{IE}-d_I^{-1}\alpha w_{II}+b_I \ge 1 ~(\text{resp. } \leq 0)\).
\end{lemma}

The proof follows a similar route as the one for Lemma~\ref{lem:param-existence} and is therefore omitted.  Proposition~\ref{prop: FLI} is extended through the following result.

\begin{proposition}[\textbf{Conditions for $k$-FLI}]\label{prop: k-FLI}
    Consider the network~\eqref{eq:EkI}.
    Then, 
    \begin{enumerate}[label=\roman*)]
        \item the E$^{k}$-I model has a unique equilibrium and $W-D \in \mathcal{LDS}$ iff $w_{EE}<d_{E}$;
        \item if the network satisfies the hypotheses of Lemma~\ref{lm: k-bin} for finite precision \(\updelta > 0\) and if \(d_E = d_I = d\) with \(w_{EE} < d_E\), the network is $k$-FLI with finite precision \(\updelta\).
    \end{enumerate}
\end{proposition}
\begin{proof}
    (i) $\Rightarrow$) This follows from Propositions~\ref{prop: LDS-to-P} and~\ref{prop: EIuniq}.

    $\Leftarrow$) Applying Theorem~\ref{thm: LDS},
    the inequalities reduce to \(d_{I}>-w_{II}\) and \(d_E>w_{EE}\), with the former being satisfied trivially, and the latter by assumption,
    giving us $W-D\in\mathcal{LDS}$. 
    Consequently, we also have that $D-W\in\mathcal{P}$ and therefore the E\textsuperscript{k}I dynamics admit a unique equilibrium point.

    (ii) Under the assumption of \(w_{EE} < d\), (i) gives us that E$^k$I has a unique equilibrium.
    When the conditions of Lemma~\ref{lm: k-bin} hold, that unique equilibrium is \(x^*\) \eqref{eq: k_FLI_eq}. 
    Since \(D = d\mathcal I_{k+1}\), Theorem~\ref{thm: EIglob} grants us the GAS of \(x^*\). 
    Thus by definition, E$^{k}$I is $k$-FLI.
\end{proof}

Recall, yet again, that while Proposition~\ref{prop: k-FLI} holds only for \(D=d\mathcal I_{k+1}\), Conjecture~\ref{conj: GAS} strongly suggests that under the LDS condition for arbitrary positive diagonal \(D\), the network is GAS at the established unique equilibrium.
Thus, under the suggested biologically-plausible synaptic structure and the established conditions of Proposition~\ref{prop: k-FLI}, this EI system is able to perform WTA behavior for an arbitrary number of inputs.

\subsection{Lateral inhibitory columns of E$^{2}$-I circuits}
We have shown that the E$^2$-I setup is capable of discriminating signals of a specified mean \(\bar u\) and minimum separation \(\updelta\), cf. Proposition~\ref{prop: FLI}. This capability is also enjoyed by the E$^k$-I setup, cf. Proposition~\ref{prop: k-FLI}, albeit the increase in the number of neurons does not translate into a higher precision. 
Recall also that in Assumption~\ref{ass: insu} (resp. Assumption~\ref{ass: insu-k}), a  normalization is applied to make it easier to control the characterization of FLI (resp. k-FLI).

However, if the input signals to the excitatory nodes are of an unknown or unexpected mean (say, for E$^2$I, \(u_{E_1} + u_{E_2} \neq 0\) still remains) or have a separation less than \(\updelta\), the firing rate model instead displays a \textit{soft} WTA behavior, settling short of complete activation or inhibition. 
Even worse, if the input mean is off by too much, it could wholly fail to discriminate between the inputs, with all E neurons simultaneously saturating or being inactivated.
Here, we address this problem by stacking multiple E$^k$I motifs in a columnar organization. We do it in two steps: in the first layer, we correct for the mean, as otherwise information might be lost at saturation or inactivation.

Secondly, we work on the resulting \textit{normalized} output of the excitatory neurons, to be fed into subsequent layers to ensure an ultimate output that exhibits \textit{full} WTA, as opposed to \textit{soft} WTA.

\subsubsection{Dynamic normalization of neural activity}
It is known in the neuroscience literature~\citep{SG:73ce} that cortical contour enhancement circuits evolve dynamically so as to normalize their total activity. 
Normalization helps in the improvement of signal-to-noise ratio in the presence of weak inputs and also aids in the prevention of overexcitation (and thus, a loss of information) due to widespread neural saturation. 
In line with the aforementioned goal of preventing runaway complete neural saturation or inactivation despite unexpected changes to the input mean, we adjust the equilibrium of the inhibitory neuron to compensate for changes in the mean of the input.
We formalize this notion in the following definition.

\begin{definition}[\textbf{FLI with inhibitory compensation}]
    An E$^2$-I unit has FLI with inhibitory compensation if the inhibitory neuron adjusts its equilibrium such that inputs of arbitrary mean can be differentiated according to FLI.
\end{definition}

The purpose of inhibitory compensation is to correct the input mean so that subsequent layers receive normalized inputs so as to then perform \textit{full} FLI.
The columnar structure will be such that subsequent layers beyond the first one do not have their own internal bias. The first layer is designed to provide the \textit{internal} bias so that the equilibrium conditions for FLI are met. The next result ensures this behavior.

\begin{theorem}[\textbf{Input-mean cancellation via inhibitory compensation}]\label{thm: norm-u}
    Consider the system~\ref{eq:E2I}
    having \(W\) as in~\eqref{eq: SynW_SI_E2I} and \(D=d\mathcal I_3\).
    Let $w_{EE}<d$, with excitatory bias \(b_E \in \mathbb R\).
    Let \(u_I=0\) and \(u_{E_1}, u_{E_2} \in \mathbb R\).

    Let the following law provide the inhibitory internal bias:
    \begin{equation}\label{eq:bI-law}
        b_I' = b_I + \left(\frac{w_{II}+d}{w_{EI}}\right)\bar u,
    \end{equation}
    where \(\bar u = (u_{E_1} + u_{E_2})/2\) and \(b_I\) is the inhibitory bias under which the excitatory equilibria are centered.
    Define \(y=Dx\).
    Assume that the original FLI equilibrium under inputs satisfying \(u_{E_1} \geq \updelta\), \(u_I = 0\), and \(u_{E_2} \leq -\updelta\) is given by 
    \(y^\mathrm{cent}= (y_{E_1}^{\mathrm{cent}}, \alpha, y_{E_2}^{\mathrm{cent}})^\top\)
    and that the compensated equilibrium under $(\bar u+\updelta,0,\bar u-\updelta)$ is 
    \(y^* = (y_{E_1}^*, \alpha', y_{E_2}^*)^\top\), with \(\alpha,\alpha'\in(0,1)\), i.e., both equilibria satisfy unsaturated conditions for the inhibitory neuron. 
    Then the excitatory components of the compensated equilibrium \(y^*\) satisfy the same equilibrium equations as those of \(y^\mathrm{cent}\), i.e., as when the inputs are \(u_{E_1}=+\updelta\), \(u_{E_2}=-\updelta\).
\end{theorem}

\begin{proof}
    By Proposition~\ref{prop: FLI}, for \(y(t) = Dx(t)\), \(y(t)\to y^*\). 

    Since \(0<\alpha'<1\),
    the inhibitory equilibrium is unsaturated and 
    \[(1+d^{-1}w_{II})y_I^*=d^{-1}w_{IE}(y_{E_1}^*+y_{E_2}^*)+b_I',\]
    hence \[-d^{-1}w_{EI}y_I^*=-d^{-1}w_{EI}\frac{d^{-1}w_{IE}(y_{E_1}^*+y_{E_2}^*)+b_I}{1+d^{-1}w_{II}}-\bar u,\] 
    where we have used~\eqref{eq:bI-law} and \(\frac{w_{II}+d}{1+d^{-1}w_{II}}=d\).
    Notice now how the term \(-d^{-1}w_{EI}y_I^*\) in the excitatory equilibrium now contains \(-\bar u\), which cancels out the nonzero mean \(\bar u\) when it enters with the input, as long as the assumption \(0<y_I^*<1\) holds.

    Let the original I equilibrium be \(\alpha = (w_{IE}d^{-1} + d_I)/(1 + w_{II}d^{-1}) \in (0,1)\).
    The excitatory equilibrium equation are consequently the following: for \(E_1\),
    \[y_{E_1}^*=[d^{-1}w_{EE}y_{E_1}^*-d^{-1}w_{EI}y_I^*+b_E+\bar u+\updelta]_0^1\]
    reduces to
    \begin{align}
        y_{E_1}^* &= [d^{-1}w_{EE}y_{E_1}^*-d^{-1}w_{EI}\frac{d^{-1}w_{IE}(y_{E_1}^*+y_{E_2}^*)+b_I}{1+d^{-1}w_{II}}+b_E+\updelta]_0^1\\ \intertext{and for $y^* = (1, y_I^*, 0)$,}
        1 & = [d^{-1}w_{EE}^*-d^{-1}w_{EI}\frac{d^{-1}w_{IE}+b_I}{1+d^{-1}w_{II}}+b_E+\updelta]_0^1\\
        1 & = [d^{-1}w_{EE}^*-d^{-1}w_{EI}\alpha+b_E+\updelta]_0^1
    \end{align}
    and for \(E_2\),
    \[y_{E_2}^*=[d^{-1}w_{EE}y_{E_2}^*-d^{-1}w_{EI}y_I^*+b_E+\bar u-\updelta]_0^1\]
    reduces to 
    \begin{align}
        y_{E_2}^* &=[d^{-1}w_{EE}y_{E_2}^*-d^{-1}w_{EI}\frac{d^{-1}w_{IE}(y_{E_1}^*+y_{E_2}^*)+b_I}{1+d^{-1}w_{II}}+b_E-\updelta]_0^1\\ \intertext{and for $y^* = (1, y_I^*, 0)$,}
        0 &=[-d^{-1}w_{EI}\frac{d^{-1}w_{IE}+b_I}{1+d^{-1}w_{II}}+b_E-\updelta]_0^1\\
        0 &= [-d^{-1}w_{EI}\alpha+b_E-\updelta]_0^1
    \end{align}
    which are exactly the original equilibrium equations with inputs \((+\updelta,0,-\updelta)\) and bias \(b_I\), proving the statement.
\end{proof}

As a result, we can characterize the  range of admissible input means that can be compensated with the first column layer.

\begin{corollary}[Admissible input-mean ranges that can be compensated with unsaturated inhibitory neuron]
\label{cor:umean_range}
Assume the hypotheses of Theorem~\ref{thm: norm-u} hold.  
The compensated inhibitory equilibrium \(d^{-1}\alpha'\) is unsaturated, i.e., \(0<\alpha'<1\), if 
\[-d^{-1}w_{EI}\alpha < \bar u < d^{-1}w_{EI}(1-\alpha),\]
where \(\alpha\) is the original centered-input inhibitory equilibrium. In particular, any input mean \(\bar u\) in this interval can be canceled via~\eqref{eq:bI-law} without inducing saturation of the inhibitory neuron.
\end{corollary}
\begin{proof}
By Theorem~\ref{thm: norm-u}, in the unsaturated regime the inhibitory equilibrium equation under the compensated bias~\eqref{eq:bI-law} yields
\[\alpha'=\alpha+\frac{d}{w_{EI}}\bar u,\]
or equivalently, $\bar u=\frac{w_{EI}}{d}(\alpha'-\alpha)$. Therefore, \(0<\alpha'<1\) holds if and only if \(0<\alpha+\frac{d}{w_{EI}}\bar u<1\), which is equivalent to
\[-\frac{w_{EI}}{d}\alpha<\bar u<\frac{w_{EI}}{d}(1-\alpha),\] 
as stated.
\end{proof}

Now that we have E$^{2}$-I units that accept arbitrary input means within a given range,

we are now ready to achieve contrast enhancement by stacking multiple layers of such E$^{2}$-I units.

\subsubsection{E$^{2}$-I columns} 

We are now interested in expanding the previous framework to the case where we have many E$^{2}$-I models stacked on top of each other, with the input to each layer being the output of the previous layer. Our goals here are:
\begin{itemize}
    \item study how a columnar architecture is functional at discriminating small gradients in the input;
    \item study the stability of the columnar architectures.
\end{itemize}
\begin{figure}[!ht]
    \centering
    \includegraphics[width=0.5\linewidth]{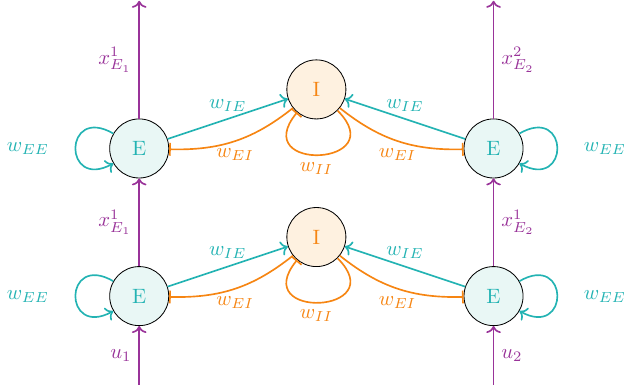}
    \caption{Schematic representation of a lateral inhibitory column composed by E$^{2}$-I networks, where the output of one layer becomes the input for the next layer.}
    \label{fig:cort-col}
\end{figure}

Let $\mathrm{L}\in\mathbb{N}$ be the number of layers in the network. We define $N_{E}=2\mathrm{L}$ to be the total number of excitatory neurons and $N_{I}=\mathrm{L}$ the total number of inhibitory neurons. The total number of neurons in the network is then $N=N_{E}+N_{I}=3\mathrm{L}$. 

We now specialize the construction of the dissipation and synaptic matrices to the case of the E$^{2}$-I column with $\mathrm{L}\in\mathbb{N}$ layers. Let $z=(d_E,d_I,d_E)\in\real^{3}$ and define
\begin{itemize}
    \item Dissipation map $D_{\mathrm{L}}:\mathbb{N}\times\real^{3}\to\real^{N\times N}$.
    \begin{equation}
        D=D_{\mathrm{L}}(z)=\diag(\underbrace{z,\dots,z}_{\mathrm{L}}).
    \end{equation}
    \item Define the intra-layer synaptic matrix
    \begin{equation}\label{eq: Syn_in}
        W_{in} = \begin{pmatrix}
            w_{EE} & -w_{EI} & 0\\
            w_{IE} & -w_{II} & w_{IE}\\
            0 & -w_{EI} & w_{EE}
        \end{pmatrix},
    \end{equation}
    and the inter-layer
    synaptic matrix
    \begin{equation}\label{eq: SynExt}
        W_{ex} = \begin{pmatrix}
            w_{EE} & 0 & 0\\
            0 & 0 & 0\\
            0 & 0 & w_{EE}
        \end{pmatrix}.
    \end{equation}
    Therefore, the synaptic matrix $W$ of the E$^{2}$-I column is
    \begin{equation}\label{eq: W_cortical_col}
        W_{\mathrm{L}}=W=
        \begin{pmatrix}
            W_{in} & \vectorzeros & \vectorzeros & \vectorzeros\\
            W_{ex} & \ddots & \vectorzeros & \vectorzeros\\
            \vectorzeros & \ddots & \ddots & \vectorzeros\\
            \vectorzeros & \vectorzeros & W_{ex} & W_{in}
        \end{pmatrix}
    \end{equation}
    with $\mathrm{L}$ diagonal blocks of intra-layer connectivity and $\mathrm{L}-1$ sub-diagonal blocks of extra-layers connectivity.
\end{itemize}

The maximal in-degree and out-degree of the excitatory-to-excitatory block of the E$^2$-I column are $k_{in}^E=2$ and $k_{out}^E=2$, respectively. According to Theorem~\ref{thm: LDS}, the $\mathcal{LDS}$ condition for the E$^{2}$-I column is ensured if 
\begin{align}
    d_{E}&>w_{EE}(k^E_{in}+k^E_{out})/2\nonumber\\
         &=\frac{4}{2}w_{EE}=2w_{EE}
\end{align}
This bound is sufficient, but might not be necessary. In fact, in what follows, we study a new sufficient condition to ensure $D-W\in\mathcal{P}$ and then use it to refine the bound for $\mathcal{LDS}$.

\begin{proposition}\label{prop: P-cort}
    $D-W=D_\mathrm{L}(z)-W_\mathrm{L}\in\mathcal{P}$ for all $\mathrm{L}\in\mathbb{N}$ iff $d_{E}>w_{EE}$.
\end{proposition}
\begin{proof}
We need to prove both directions of the implication.
\begin{itemize}
    \item[$\Rightarrow$] Apply Proposition~\ref{prop: EIuniq}.
    \item[$\Leftarrow$] We will proceed with a proof by induction.
    \begin{enumerate}
        \item[(I)] $\mathrm{L}=1$.\\
        We start from a one-layer E$^{2}$-I network, with
        \begin{align}
            D_1(z)&=\diag(d_E,d_I,d_E)\\
            W_1&=W_{in}
        \end{align}
        Consider now the matrix $A=D-W$, which has components
        \begin{equation}
        A =
            \begin{pmatrix}
                d_{E}-w_{EE} & w_{EI} & 0\\
                -w_{IE} & d_{I}+w_{II} & -w_{IE}\\
                0 & w_{EI} & d_{E}-w_{EE}
            \end{pmatrix}
        \end{equation}
        and we now focus on computing its principal minors.
        \begin{itemize}
            \item order 1:
            \begin{align}
                d_{E}-w_{EE}&>0\\
                d_{I}+w_{II}&>0
            \end{align}
            \item order 2:
            \begin{align}
                (d_{E}-w_{EE})(d_{I}+w_{II})+w_{EI}w_{IE}&>0\\
                (d_{E}-w_{EE})^{2}&>0
            \end{align}
            \item order 3:
                \begin{align}
                (d_{E}-w_{EE})[(d_{E}-w_{EE})(d_{I}+w_{II})+2w_{EI}w_{IE}]>0
            \end{align}
        \end{itemize}
        \item[(II)] L>1.\\
        We now suppose that $D_{\mathrm{L}-1}(z)-W_{\mathrm{L}-1}\in\mathcal{P}$, and we want to understand whether also $D_{\mathrm{L}}(z)-W_{\mathrm{L}}\in\mathcal{P}$. We start by noticing that we can write
        \begin{equation}
            D_{\mathrm{L}}(z)-W_{\mathrm{L}}=
            \begin{pmatrix}
                D_{\mathrm{L}-1}(z)-W_{\mathrm{L}-1} & \vectorzeros[3(\mathrm{L}-1)\times 3]\\
                \mathcal{O}_{ex} & D_{1}(z)-W_{in}
            \end{pmatrix}
        \end{equation}
        where $\mathcal{O}_{ex}=(\vectorzeros[3\times 3(\mathrm{L}-2)]\:\:W_{ex})$. We can now address the last three principal minors, for which we need to consider all the possible combinations of deletions of rows and columns. For simplicity, we will denote with $\upalpha_{k}$, for $k=1,\dots,3(\mathrm{L}-1)$, any of the principal minors of $D_{\mathrm{L}-1}(z)-W_{\mathrm{L}-1}$ of order $k$. Notice in particular that we have $\upalpha_{k}>0$ for all $k=1,\dots,3(\mathrm{L}-1)$.
        \begin{itemize}
            \item order $3\mathrm{L}-2$:\\
            If we delete two rows and two columns with the same indices from the last block $W_{in}$, then the principal minors are
            \begin{align}\label{eq: minor3d-2}
                (d_E-w_{EE})\upalpha_{3\mathrm{L}-3}&>0\\
                (d_I+w_{II})\upalpha_{3\mathrm{L}-3}&>0
            \end{align}
            If we delete one row and one column from the last $W_{in}$ block and one from the block $D_{\mathrm{L}-1}(z)-W_{\mathrm{L}-1}$, then the principal minors are 
            \begin{align}
                [(d_E-w_{EE})(d_I+w_{II})+w_{EI}w_{IE}]\upalpha_{3\mathrm{L}-4}&>0\\
                (d_E-w_{EE})^2\upalpha_{3\mathrm{L}-4}&>0
            \end{align}
            Finally, if we delete only two rows and columns with same index in $D_{\mathrm{L}-1}(z)-W_{\mathrm{L}-1}$ we have
            \begin{equation}
                (d_E-w_{EE})[(d_E-w_{EE})(d_I+w_{II})+2w_{IE}w_{EI}]\upalpha_{3\mathrm{L}-5}>0
            \end{equation}
            \item order $3\mathrm{L}-1$:\\
            By deleting one row and one column with same index from the last $W_{in}$ block, the principal minors are
            \begin{align}
                [(d_E-w_{EE})(d_{I}+w_{II})+w_{EI}w_{IE}]\upalpha_{3\mathrm{L}-3}&>0\\
                (d_E-w_{EE})^{2}\upalpha_{3\mathrm{L}-3}&>0
            \end{align}
            Instead, by deleting one row and one column with same index from the block $D_{\mathrm{L}-1}(z)-W_{\mathrm{L}-1}$, the principal minors have the form
            \begin{equation}
                (d_E-w_{EE})[(d_E-w_{EE})(d_I+w_{II})+2w_{IE}w_{EI}]\upalpha_{3\mathrm{L}-4}>0
            \end{equation}
            \item order $3\mathrm{L}$:\\
            \begin{equation}\label{eq: minor3d}
                (d_E-w_{EE})[(d_E-w_{EE})(d_I+w_{II})+2w_{IE}w_{EI}]\upalpha_{3\mathrm{L}-3}>0
            \end{equation}
        \end{itemize}
    \end{enumerate}
\end{itemize}
\end{proof}

Notice that by imposing $ \upalpha_{0}=1$ and $\upalpha_{-1},\upalpha_{-2}=0$, we can use~\eqref{eq: minor3d-2} to~\eqref{eq: minor3d} as recursive expressions for the principal minors of the matrix $D_{\mathrm{L}}(z)-W_{\mathrm{L}}$ for all $\mathrm{L}\in\mathbb{N}$.

We are now interested at studying the $\mathcal{LDS}$ condition in the specific case of the E$^2$-I column, possibly applying the results of Theorem~\ref{thm: LDS}. Rewriting the dissipation and synaptic matrices for the E$^{2}$-I column

according to the formalism of~\eqref{eq: dis} and~\eqref{eq: syn}, $N_E=2\mathrm{L}$, $N_I=\mathrm{L}$, and the diagonal components of the synaptic matrix $W_\mathrm{L}$ are
\begin{align}
    W_E^E&=\begin{pmatrix}
        w_{EE} & 0 & 0 & \dots & 0\\
        0 & w_{EE} & 0 & \dots & 0\\
        w_{EE} & 0 & \ddots & \dots & \vdots\\
        0 & \ddots & 0 & \ddots & 0\\
        0 & \dots & w_{EE} &0 & w_{EE}
    \end{pmatrix}\\
    W_I^I&=\diag(\underbrace{w_{II},\dots,w_{II}}_{\mathrm{L}}).
\end{align}
Notice that we do not consider the extra-diagonal blocks of asymmetric connections, as under Assumption~\ref{ass: balance} they do not contribute to the $\mathcal{LDS}$ treatment. 

\begin{proposition}[\textbf{$\mathcal{LDS}$ of E$^{2}$-I columns}]
    Let $d_E>2w_{EE}$. Then
    \begin{enumerate}
        \item[(i)] $W_{\mathrm{L}}-D_{\mathrm{L}}(z)\in\mathcal{LDS}$;
        \item[(ii)] the eigenvalues of $D_{\mathrm{L}}(z)-W_{\mathrm{L}}$ are
        \begin{equation}\label{eq:eigenvalues-Toeplitz}
            \uplambda_{k}=2p_E(d_E-w_{EE})-2p_E w_{EE}\cos\left(\frac{k\uppi}{N+1}\right)\quad k=1,\dots,N.
        \end{equation}
    \end{enumerate}
\end{proposition}
\begin{proof}
    In order to apply the result of Theorem~\ref{thm: LDS}, we need to study the matrix $(D_E-W_E^E)^{\top}P_E+P_E(D_E-W_E^E)$, for which the only non-zero coefficients are
    \begin{align}
        \{(D_E-W_E^E)^{\top}P_E+P_E(D_E-W_E^E)\}_{i,i}&=2p_E(d_E-w_{EE})\\
        \{(D_E-W_E^E)^{\top}P_E+P_E(D_E-W_E^E)\}_{i,i\pm 2}&=-p_E w_{EE}
    \end{align}
    Therefore, applying Gershgorin circle theorem~\citep{HRA-JCR:85}, we obtain the condition $d_E>2w_{EE}$. Upon closer inspection, it can be observed that the connections are disjoint for every pair of neurons, and therefore it is possible to find a row $Q$ and columns $Q^{\top}$ permutation matrix such that the transformed $Q\{(D_E-W_E^E)^{\top}P_E+P_E(D_E-W_E^E)\}Q^{\top}$ excitatory block becomes a tri-diagonal matrix.
    \begin{equation}
        Q\{(D_E-W_E^E)^{\top}P_E+P_E(D_E-W_E^E)\}Q^{\top}=
        \begin{pmatrix}
            2p_E(d_E-w_{EE}) & -p_E w_{EE} &  \dots & 0\\
            -p_E w_{EE} & 2p_E(d_E-w_{EE}) & \ddots & \vdots\\
            0 & \ddots & \ddots & -p_E w_{EE}\\
            0 & \dots & -p_E w_{EE} & 2p_E (d_E-w_{EE})
        \end{pmatrix}
    \end{equation}
    which is a Toeplitz matrix, and a formula for the eigenvalues is known~\citep{RMG:06}.
    Specifically, the eigenvalues of this matrix are given by~\eqref{eq:eigenvalues-Toeplitz}.
\end{proof}

Notice in particular that the bound given by Gershgorin circle theorem becomes tight only in the limit $N\to\infty$, since we get
\begin{align}
    \lim_{N\to +\infty} \uplambda_{1}&= \lim_{N\to +\infty} 2p_E(d_E-w_{EE})-2p_E w_{EE}\cos\left(\frac{\uppi}{N+1}\right)\nonumber\\
    &=2p_E(d_E-w_{EE})-2p_E w_{EE}
\end{align}

\begin{proposition}[Required height of $\mathrm{E}^{2}$-$\mathrm{I}$ column for WTA]\label{prop:E2I CC}
    Consider an E$^{2}$-I column with \(l \in \mathbb N\) layers under inhibitory compensation and synaptic matrix as given in \eqref{eq: W_cortical_col}. 
    Let each E$^{2}$-I unit comprising the E$^{2}$-I column have parameters such that the minimum separation of the inputs for lateral inhibition is \( \updelta\).
    Assume that the desired precision of the E$^{2}$-I column for WTA is \(\upepsilon <  \updelta\).
    Then, it must hold that
    \begin{equation}
        \mathrm{L} \geq 1+\Bigg\lceil\frac{\ln{(\upepsilon/\updelta)}}{\ln{(d_{E}/w_{EE} - 1)}}\Bigg\rceil
    \end{equation}
    and the E$^{2}$-I column achieves full WTA for input means within the range defined in Corollary~\ref{cor:umean_range}.

\end{proposition}
\begin{proof}
    Since it is given that the input signals have a separation of \(\upepsilon <  \updelta\), the first layer will not achieve full WTA.
    Instead, using the equilibrium equations, and the fact the input means are compensated against,
    \begin{equation}
        d_{E}(x_{E_1} - x_{E_2}) = w_{EE}(x_{E_1} - x_{E_2}) + 2\upepsilon.
    \end{equation}
    Let \(\upepsilon^{(1)}\) be the separation in the inputs to the second layer. Then,
    \begin{equation}
            \upepsilon^{(1)} \coloneqq \frac{w_{EE}}{2}(x_{E_1} - x_{E_2}) = \frac{\upepsilon}{(d_{E} - w_{EE})/w_{EE}}.
    \end{equation}
    For \(L\) layers, this is compounded multiplicatively, and we get that
    \begin{equation}
        \upepsilon^{(\mathrm{L}-1)} \coloneqq \frac{x^{\mathrm{L}}_{E_1} - x^{\mathrm{L}}_{E_2}}{2} = \frac{\upepsilon}{((d_{E} - w_{EE})/w_{EE})^{\mathrm{L}-1}}
    \end{equation}
    where \(\upepsilon^{(\mathrm{L}-1)}\) is the half-separation in the inputs to the \(\mathrm{L}\)-th layer. To achieve proper WTA phenomenology, we need to impose the condition
    \begin{equation}
        \updelta = \frac{\upepsilon}{((d_{E} - w_{EE})/w_{EE})^{\mathrm{L}-1}}
    \end{equation}
    Taking the natural logarithm and recalling that \(w_{EE} < d_{E}\), we get the desired result.
\end{proof}

\begin{remark}[\textbf{Positivity of the bound on the number of layers}]
    As we have seen earlier in Section~\ref{sec:param_E2I} under Remark~\ref{rem:crit_prec}, when $\updelta>1/2$, the choice of the parameters is no longer dependent on the precision $\updelta$, but exclusively on the $\mathcal{LDS}$ bound $d_{E}^{-1}w_{EE}< 1$. 
    Therefore, it is important to check if for all $\updelta \in(0,1/2)$, the minimum number of layers needed is a well-defined quantity.
    From Proposition~\ref{prop:E2I CC}, 
    \begin{equation*}
        \mathrm{L}\ge 1 + \Bigg\lceil\frac{\ln(\upepsilon/\updelta)}{\ln(d_{E}/w_{EE}-1)}\Bigg\rceil.
    \end{equation*}

    The numerator of the fraction in the r.h.s. is negative, since $\upepsilon<\updelta$. For $b=0$, we start from~\eqref{eq:bounds-cond-lemma}
    \begin{equation*}
        d_{E}^{-1}w_{EE}-d_I^{-1}w_{IE}> 1-\updelta,
    \end{equation*}
    from which it follows that
    \begin{equation}
        d_{E}^{-1}w_{EE}> 1-\updelta.
    \end{equation}
    Reorganizing the inequality as
    \begin{equation}
        \frac{d_E}{w_{EE}}<\frac{1}{1-\updelta}
    \end{equation}
    and subtracting $1$ on both sides we get
    \begin{equation*}
        \frac{d_{E}}{w_{EE}}-1< \frac{\updelta}{1-\updelta}<1
    \end{equation*}
    for $\updelta\in(0,1/2)$. Under the LDS condition, \(d_{E}/w_{EE}>1\) and \(d_{E}/w_{EE} - 1 > 0\). 
    Consequently, we also have that $\ln(d_{E}/w_{EE}-1)<0$ and
    \begin{equation}
        \frac{\ln(\upepsilon/\updelta)}{\ln(d_{E}/w_{EE}-1)}\geq0
    \end{equation}
    for all \(\delta \in (0, 1/2)\), and the number of layers is thus well-defined.
\end{remark}

\subsubsection{E$^{k}$-I columns}

We now desire to extend our results to more general lateral inhibitory column comprised of E$^{k}$-I units at each level for \(k \in\ \mathbb N\), but first address the problem of what happens if \(\bar u \neq 0\).

\begin{proposition}[Input normalization of E$^{k}$-I unit]\label{prop:EkI compensation}
    Consider the FR network in Definition~\ref{def:FR} with the synaptic matrix as in \eqref{eq:EkI synapse} satisfying \(w_{EE} < 1\).
    Let \(\bar u \in \mathbb R\) be the incoming input baseline and \(b_E\) the excitatory internal bias. 
    Let \(\bar \updelta\) be the difference between the highest and second highest signals, and let \(\bar x_I\) be the inhibitory neuron equilibrium when the input baseline \(\bar u=0\) and the highest and second highest input signals are separated by at least \(\updelta\).
    Let \(b_I\) be the component of the internal inhibitory bias. 
    Then, the output baseline is restored to \(b_E\) for any nonzero input baseline \(\bar u\) satisfying
    \begin{equation}\label{eq: bar u range EkI}
        \bar u \in [-w_{EI}\bar x_I, w_{EI}(1 - \bar x_I)]
    \end{equation}
    if the inhibitory internal bias is given by
    \begin{equation}\label{eq:bI-compensation EkI}
        b_I = \bar b_I + \left[\frac{w_{II}+1}{w_{EI}}\right] \bar u.
    \end{equation}
\end{proposition}
\begin{proof}
    Consider the equilibrium equation for the I neuron at the desired excitatory equilibrium and use the given compensatory bias in~\eqref{eq:bI-compensation EkI}. Then we get
    \[x_I = \left[w_{IE}\bar x_E - w_{II}x_I + \bar b_I + \frac{w_{II}+1}{w_{EI}} \bar u\right]_0^1.\]
    Under the given range of \(\bar u\), we have that
    \[x_I = \frac{\bar u}{w_{EI}} + \bar x_I\]
    at the unique equilibrium. 
    Thus, the excitatory equilibrium equation when written out gives us, for the neuron \(j\) with the highest input
    \begin{align}
        x_{E_j} &= [w_{EE}x_{E_j} - w_{EI}\bar x_I + \bar u + \updelta_j]_0^1\\
        x_{E_j} &= [w_{EE}x_{E_j} - w_{EI} x_I- \updelta_j]_0^1
    \end{align}
    of which \(x_{E_j}=1\) is a solution, following from the inequalities in~\eqref{eq:EkI ineq}.
    Similarly, for other excitatory neurons, denoted by \(i\),
    \begin{align}
        x_{E_i} &= [w_{EE}x_{E_i} - w_{EI} \bar x_I + b_E + \bar u - \updelta_i]_0^1\\
        x_{E_i} &= [w_{EE}x_{E_i} - w_{EI} x_I + b_E - \updelta_i]_0^1
    \end{align}
    which also admits \(x_{E_i}=0\) as a solution due to~\eqref{eq:EkI ineq}.
    Thus, since the firing rate equilibrium is unique under the assumption of \(w_{EE} < d_E\), this is the unique equilibrium point that the E$^{k}$-I unit is stable at.
\end{proof}

We now put all the results together in the following.

\begin{proposition}[FLI of inhibitory-compensated E$^{k}$-I column]
    Consider a \(\mathrm{L}\)-layered E$^{k}$-I column with the assumption of normalized inputs at each layer as in Proposition~\ref{prop:EkI compensation} and let each layer have a finite precision \(\updelta>0\). 
    Let the input \(u\) to the first layer be such that 
    \begin{equation}
        u = [u_{E_1}, \dots, u_{E_k}] = \bar u \mymathbb{1} + [\upepsilon_1, \dots, \upepsilon_k]
    \end{equation} 
    where the input baseline \(\bar u\) satisfies~\eqref{eq: bar u range EkI}. 
    Assume that there exists \(j \in\{1,\dots,k\}\) and $0<\upepsilon<\updelta$  such that
    \begin{align}
        \upepsilon_j &> \upepsilon\\
        \upepsilon_i &< -\upepsilon\quad \forall i\neq j.
    \end{align}
    Then, the E$^k$-I column exhibits FLI if
    \begin{equation}
        \mathrm{L} \geq 1+\Bigg\lceil\frac{\ln{(\upepsilon/\updelta)}}{\ln{(d_{E}/w_{EE} - 1)}}\Bigg\rceil
    \end{equation}
    and specifically the \(j^{th}\) excitatory neuron is eventually the only one active.
\end{proposition}

The proof follows a similar route to that of Proposition~\ref{prop:E2I CC} and is therefore omitted.
\stopcontents[sections]

\end{document}